\documentclass[12pt]{article}
\usepackage[utf8]{inputenc}
\usepackage{amsmath}
\usepackage{amsfonts}
\usepackage{amssymb}
\usepackage{amsthm}
\usepackage{bm, bbm}
\usepackage{natbib}
\usepackage{mathrsfs}
\defcitealias{AOT}{AOT}
\makeatletter
\@ifpackageloaded{hyperref}{\newcommand{\mylabel}[2]{\protected@write\@auxout{}{\string\newlabel{#1}{{#2}{\thepage}{\@currentlabelname}{\@currentHref}{}}}}}{\newcommand{\mylabel}[2]{\protected@write\@auxout{}{\string\newlabel{#1}{{#2}{\thepage}}}}}
\makeatother

\usepackage{makeidx}
\usepackage{graphicx}
\usepackage{xcolor}
\usepackage{subcaption}
\usepackage{mwe}
\usepackage{enumerate,enumitem}

\usepackage{float}
\usepackage{longtable}
\usepackage{sectsty}
\usepackage{tabulary}
\usepackage{booktabs}

\usepackage[colorlinks=true,linkcolor=black,anchorcolor=black,citecolor=black,filecolor=black,menucolor=black,runcolor=black,urlcolor=black]{hyperref}

\usepackage{chngcntr}
\usepackage{apptools}
\AtAppendix{\counterwithin{proposition}{section}}
\AtAppendix{\counterwithin{lemma}{section}}
\AtAppendix{\counterwithin{corollary}{section}}

\usepackage{tikz}
\usetikzlibrary{arrows,shapes,trees,cd}
\usepackage[normalem]{ulem}
\newtheorem{definition}{Definition}

\newtheorem{corollary}{Corollary}
\newtheorem{lemma}{Lemma}

\newtheorem{proposition}{Proposition}

\usepackage[margin=1in]{geometry}

\def\FT{{\hat{\mathbf{F}}}\!{\,^\top}}

\title{SYSTEMIC RISK IN FINANCIAL NETWORKS REVISITED: DEBT DILUTION AS A BACKDOOR BAIL-IN\thanks{For helpful comments, thanks to
Franklin Allen,
Vladimir Asriyan, Ana Babus,
Bruno Biais,
Philip Bond,
Agostino Capponi,
Tetiana Davydiuk,
Doug Diamond,
Darrell Duffie,
Phil Dybvig,
Douglas Gale,
Todd Gormley,
Denis Gromb,
Zhiguo He,
Florian Heider,
Laurie Hodrick,
Sebastian Infante,
Dalida Kadyrzhanova,
Martin Oehmke,
Mina Lee,
Francesco Palazzo,
Uday Rajan,
Adriano Rampini,
Victoria Vanasco,
Laura Veldkamp,
Thao Vuong,
Chaojun Wang,
David Webb,
Alexander Zentefis,
Jeff Zwiebel,
and audiences at the 2018 AFA,
Banca d'Italia,
EIEF,
the 2017 ECWFC Meeting,
the 2017 Summer Meeting of the FTG,
the 2017 FTG Summer School,
the FRA,
the 2017 Summer Symposium at Gerzensee,
the 2018 Maryland Junior Finance Conference, Washington University in St.\ Louis, and the 2017 Yale Junior Finance conference.
Thanks to Shunsuke Matsuno and Dayou Xi for research assistance.
This paper replaces \cite{Netting}.}}

\author{ Jason Roderick Donaldson\footnote{USC and CEPR.}
     \and Giorgia Piacentino\footnote{USC, CEPR, ECGI, and NBER.}
      \and Xiaobo Yu\footnote{CU--Boulder.}   }
\date{   August 12, 2026
}

\begin{document}

\maketitle

\begin{abstract}

\noindent We develop a model of interbank networks with random liquidity shocks. Networks of dilutable debt---e.g., long-term, unsecured---facilitate efficient liquidity transfers: Shocked banks pledge interbank claims as collateral for new senior debt, diluting existing debt. Unlike with non-dilutable debt, indebtedness and connectedness are sources of stability, not fragility. {}Dilution is thus a ``backdoor bail-in'' that reallocates losses absent a resolution authority, trigger security, or ex post renegotiation. We uncover a class of networks, ``exponential networks,'' that implement optimal contingent transfers via plain debt. Yet exponential networks are not pairwise stable, whereas some core--periphery networks are, rationalizing observed interbank structures and their under-insurance against crises.

\end{abstract}

\thispagestyle{empty}

\newpage

\section{Introduction} \setcounter{page}{1} \thispagestyle{empty}

Interbank debts are viewed as a threat to financial stability, partly due to a theory literature that shows how a shock to one bank can transmit to its counterparties; as a result, tightly interconnected networks of debt are ``robust yet fragile,'' absorbing everyday shocks but amplifying extraordinary ones (see, notably, \cite*{AOT}, hereinafter AOT, \cite{Allen-Gale-2000}, and \cite{Eisenberg-Noe-2001}). \label{pp:bail-in-intro}{}This fragility has motivated policies that shift a distressed bank's losses onto its healthy counterparties, aiming to reallocate liquidity and contain distress---statutory bail-ins, contingent convertibles, and regulator-coordinated private recapitalizations.\footnote{Examples include statutory bail-in powers (\cite{Bolton-Oehmke-2019} and \cite{Kanik-2022}); contingent convertible securities (\cite{Flannery-2014} and \cite{Sundaresan-Wang-2015}); private recapitalizations and rescue arrangements (\cite{Leitner-2005}, \cite{Rogers-Veraart-2013}, \cite{Kanik-2020}, and \cite*{Bernard-Capponi-Stiglitz-2022}); and public interventions targeted to financial networks (\cite{Erol-2019}, \cite*{Capponi-Corell-Stiglitz-2022} and \cite{Jackson-Pernoud-2024}).} {}Such policies generally require something beyond ordinary debt: a resolution authority, an explicit contingent security, or an ex post agreement among counterparties.

The corporate finance literature points to a different possibility: Ordinary debt may itself implement such transfers. Debt that can be diluted---so new senior debt can be issued before existing debt is paid---allows a distressed debtor to raise new debt at the expense of old, thereby avoiding inefficient liquidation.\footnote{See, among others, \cite{Stulz-Johnson-1985}, \cite{Diamond-1993}, \cite{Hart-Moore-1995}, and \citeauthor*{Paradox} (\citeyear{Paradox}, \citeyear{Donaldson-Gromb-Piacentino-2025}); see also \cite*{Collateral} for a survey of the literature on collateral, dilution, and priority.} The mechanism, which relies on gross interbank claims, suggests a downside of netting out.\footnote{Existing evidence suggests that a substantial fraction of interbank debt has maturity long enough for dilution to be relevant. E.g.,\ several papers using German data report average maturities exceeding one year, with overnight debt accounting for only about 10\% of exposures (e.g., \cite{Bluhm-et-al-2016}, \cite{Craig-Ma-2021}, \cite{Craig-vonPeter-2014}, \cite{gabrieli2014network}, \cite{upperEstimatingBilateralExposures2004}). \cite{Kuo-et-al-2014} document the scarcity of U.S. maturity data and infer from payment records that roughly one quarter of interbank debt has term maturity.}$^,$\footnote{See, e.g., \cite{Allen-Gale-2000}, \cite{AOT}, \cite{ALLE/BABU/CARL/12}, \cite{Kusnetsov-Veraart-2018}, \cite{He-Li-2022} and \cite{Shu-2025}.}
Can gross positions be a source of insurance rather than contagion in the financial system?
If so, would netting them out destroy coinsurance?

 We develop a model of a financial system with dilutable interbank debt that suggests that the answer to both questions is yes: Gross debt positions insure banks against liquidity shocks, so netting them out destroys valuable risk sharing. In diametric contrast to what happens in networks of non-dilutable debt, high indebtedness and connectedness become sources of stability: More debt provides collateral against which distressed banks can borrow, and denser connections allow liquidity to flow from banks that have it to banks that need it.

High indebtedness alone does not guarantee efficiency. The distribution of debt across banks matters too. However, we identify a heretofore unexplored class of networks, which we call ``exponential networks,'' that implement the efficient allocation of liquidity for every realization of shocks. They are {}robust and never fragile, implementing a ``backdoor bail-in'': Ordinary debt achieves the desired contingent transfers without a resolution authority, contingent security, or ex post negotiation. \label{pp:R1-2}

{}That does not mean there is no role for outside intervention. Exponential networks are unlikely to arise via free contracting (in the jargon, they are not pairwise stable). Core--periphery networks like those observed in practice\footnote{See, e.g., \cite{Craig-vonPeter-2014}, \cite{gabrieli2014network}, and \cite{Craig-Ma-2021}.} are more likely to arise but under-insure against crises. Our analysis therefore suggests that, when the loss-absorbing capacity comes from within the financial system through dilution, some banks being ``too big to fail'' can improve rather than undermine systemic stability.

{}Prima facie, the model tells a policy maker to impose an exponential network of interbank debts. We would not be so literal given how abstract the model is. But two lessons seem robust: Symmetric networks cannot insure banks against large shocks and networks of non-dilutable debt cannot insure them at all. Policies that compress gross exposures, cap the largest ones, or curtail the superpriority of new secured claims may therefore trade real insurance for the appearance of safety.\footnote{Examples include EMIR portfolio-compression requirements, Basel capital deductions for banks' holdings of other G-SIBs' TLAC instruments and large-exposure limits, and U.S. QFC stay and clean-holding-company rules.}

To study debt dilution in the financial system, we embed the liquidity-risk framework of \cite{Holmstrom-Tirole-1998} in a network. Banks have imperfectly pledgeable long-term assets and may suffer liquidity shocks before those assets mature. They also owe debts to and from one another, which can be seen as cross-holdings of loans or bonds. In the baseline, interbank debt matures after the shock, so it can be diluted with new senior debt, representing, e.g., repos, which have superpriority in bankruptcy (\cite{Duffie-Skeel-2012}). In a benchmark, it matures at the time of the shock, too soon to be diluted to raise liquidity.

The benchmark turns out to be isomorphic to \citetalias{AOT}'s model (Lemma~\ref{l:isomorphism}), allowing us to import their measures of connectedness and recover (and sometimes strengthen) their results: With non-dilutable debt, netting improves stability, connectedness reduces it, and there is a ``default radius'' around a shocked bank (Lemmata~\ref{l:ST netting}--\ref{l:ST bottleneck}).

Our first set of results shows how each of these conclusions reverses when debt can be diluted. In that case, netting undermines stability (Proposition~\ref{p:netting}), connectedness enhances it (Propositions~\ref{p:delta} and~\ref{p:bottleneck}), and there is a ``salvation radius'' around a healthy bank (Proposition \ref{p:radius}). Behind the reversal is a change in what the network transmits: Rather than spreading liquidity shortages from shocked banks, it spreads liquidity surpluses from healthy ones.

High connectedness and indebtedness are not enough for (constrained) efficiency. The distribution of links matters too. For instance, symmetric networks tend to be inefficient because they allocate liquidity equally to all shocked banks, ``wasting'' it on banks that cannot be saved.

In our second set of results, we explore what we call exponential networks.
We show that these networks, which are tightly connected but asymmetric, implement an efficient priority rule, allocating all free liquidity to the largest set of banks that can be saved and writing off other shocked banks entirely (Proposition~\ref{p:constrained efficiciency}): An optimal state-contingent transfer scheme built entirely from plain, non-contingent debt.
That shows how debt dilution extends the contingencies that ordinary debt can implement beyond what can be done via default alone  (\cite{Allen-Gale-1998}, \cite*{Dubey-et-al-1988}, \cite{Zame-1993}).\footnote{Since dilution requires new debt to be issued before outstanding debt matures, the mechanism also points to a different way that maturity can substitute for explicit state contingency, complementing \cite{Angeletos-2002} and  \cite{Gale-1990}.}

Our third set of results pertains to a version of endogenous network formation (pairwise stability of the face values between linked banks\footnote{So netting out is always allowed, but forming completely new links is not.}).
The exponential network is---alas!---not pairwise stable. Versions of core--periphery networks are. Together these results suggest our model can explain why core–periphery structures can persist and how banks underinsure against crises--the networks banks sustain privately can be pairwise stable but fail to provide the socially optimal allocation of liquidity.

\label{pp:hetero}{}For our fourth set of results, we extend the model, allowing banks to differ in their assets and liquidity shocks. We show that the planner's problem is equivalent to what is known as the knapsack problem (Proposition~\ref{p:PP=KP}) and that the exponential network implements the greedy algorithm, a method used to solve it approximately (Proposition~\ref{p:exponential_implement_greedy}). Here the network's structure fixes the order in which banks are saved: The same ranking applies in every state, so the network is efficient if the planner's priorities are not state contingent as is the case, e.g., if banks have identical liquidity needs (Corollary~\ref{c:optimality of greedy}); otherwise it admits explicit efficiency guarantees (Corollary~\ref{c:greedy small}). Exponential networks are robust not only to which banks are shocked, but also to heterogeneity across banks.

    {}We analyze three other extensions, the first two of which qualify our rosy view of large long-term debts (Section~\ref{s:extensions}). In one, liquidation is efficient, so not every shocked bank should necessarily be rescued; in the other, any default is socially costly, not only those that induce inefficient liquidation. In both cases, we characterize debt levels that implement the efficient outcome, albeit only for specific network structures. The message is the same as in each: As in the baseline, debts should be large enough to provide liquidity insurance when liquidation is inefficient, but, unlike in the baseline, not so large that they induce inefficiency (be it excessive continuation or excessive default).\label{pp:extensions}

\label{pp:theta}{}In our final extension, we let banks choose their asset pledgeability at a cost. We find, albeit in a two-bank example, that cross-holdings of dilutable debt can support efficient pledgeability choices: As long as pledgeability is worth having in autarky, every socially optimal profile of choices is an equilibrium. The reason is that a bank's liquidity reaches its counterparty only in the states in which the bank is not shocked itself, so a bank that cuts back loses its own protection first. As elsewhere, interbank debts cannot be too low: It is the claim on a counterparty that makes its liquidity available.

\label{pp:inside-liquidity}{}The results offer a new perspective on inside liquidity. In \citeauthor{Holmstrom-Tirole-1998} (\citeyear{Holmstrom-Tirole-1998}, \citeyear{Holmstrom-Tirole-2011}), there is a continuum of agents with idiosyncratic shocks. They show that everyone holding a share of a diversified market index implements efficiency. That is not so in our setting with a finite number of agents, as the cross-section of shocks is a source of risk. In this case, we show that efficiency requires triage among shocked banks: Saving some and writing off others. That could be implemented by pooling assets in a portfolio, \`a la \citeauthor{Holmstrom-Tirole-1998}, with triage implemented by an active fund manager. But intermediation is costly and could introduce frictions of its own. Our answer requires no intermediary at all. Bilateral debts encode the priority ordering ex ante in the profile of face values. Inside liquidity can thus be reallocated among a few large, interconnected institutions with no fund, no manager, and no contingent claims.

   {}Overall, we provide a new way to think about financial stability and bank resolution. Much of the literature studies mechanisms to reallocate losses once a bank is in distress; we show that ordinary debt can do the same thing. Building on the corporate-finance insight that debt dilution provides liquidity, we show that interbank debt acts as a private bail-in in our model, reallocating losses through contractual priority rather than regulatory intervention. As a result, gross interbank debt need not propagate illiquidity; it can propagate liquidity too. As such,  netting could reduce risk sharing and, moreover, large, highly interconnected banks stabilize rather than destabilize the financial system. Debt dilution thus reverses the financial stability role of financial networks.\footnote{Surveys of the financial networks literature include \cite*{allen2009financial}, \cite{Allen-Walther-2021}, \cite{glasserman2016contagion}, and \cite{jackson2021systemic}.}\label{pp:R1-2b}

{}The bail-in mechanisms in the literature tend to require that someone do something or hold something beyond plain debt: Statutory bail-ins and public interventions require an authority to act ex post
(\cite{Bolton-Oehmke-2019}, \cite{Erol-2019}, \cite{Kanik-2022},
\cite*{Capponi-Corell-Stiglitz-2022}, \cite{Jackson-Pernoud-2024});
contingent capital requires an explicit contingent security, the trigger of which must be designed with care (\cite{Flannery-2014},
\cite{Sundaresan-Wang-2015}); private recapitalizations require creditors to be willing to inject liquidity or accept haircuts ex post (\cite{Leitner-2005}, \cite{Rogers-Veraart-2013}, \cite{Kanik-2020}, \cite*{Bernard-Capponi-Stiglitz-2022}). Dilution requires none of these: It is private---no authority stands between distress and transfer; it is written in plain debt---the contingency comes from the shocked bank's borrowing, not from a contractual trigger, so there is no trigger to design and no trigger-induced multiplicity (our payment equilibrium is generically unique; Proposition~\ref{p:existence}); and it embeds
commitment for free---banks cannot refuse to be diluted.

Still, the mechanism has its own limitations: It is only approximately efficient when banks are heterogeneous (Section~\ref{s:hetero}); it relies on high debts, which can induce excessive continuation or costly default (Section~\ref{s:extensions}); because the right network is unlikely to arise via free contracting (Section~\ref{s:stability}), it calls for regulation of network structure; and it requires the debt market to be open for shocked banks. But each limitation comes with a mitigant: The approximation admits explicit efficiency guarantees
(Corollaries~\ref{c:optimality of greedy} and~\ref{c:greedy small});
debt levels can be calibrated to trade insurance off against excessive
continuation and default; the regulation required is ex ante---of the
network, not of banks in distress---and so creates no expectation of
rescue to distort incentives; and the required market is for safe,
secured debt, mitigating concerns of information-based freezes.

The rest of the paper proceeds as follows. Section \ref{s:model} presents the model. Section \ref{s:ST} considers the benchmark without dilution. Section \ref{s:LT} states the qualitative properties of networks when debt can be diluted. Section \ref{s:exp} analyzes the exponential network. Section \ref{s:stability} asks which networks arise from decentralized contracting. Section \ref{s:hetero} allows banks to be heterogeneous. Section \ref{s:extensions} analyzes further extensions.
Section \ref{s:conclusion} concludes.  The Appendix contains all proofs, a worked example on the ring network, and a table of notations.

\section{Model} \label{s:model}

    We consider a model with two dates $t \in \{1 , 2 \}$ and $N \geq 2$ agents $\mathrm{B}_1,..., \mathrm{B}_N$, which we refer to as ``banks.'' Each has the risk of a liquidity shock $\ell$ at Date~1  and a long-term asset worth $y$ if held until Date~2 but with value only $\theta y < y$ to outsiders (at either date).\footnote{\label{fn:deposits}{}The liquidity shock can be interpreted as a bank run. Under that interpretation, banks that are \emph{not} run should still owe deposits at Date~2---senior claims that, unlike interbank debts, cannot be diluted. Adding such deposits, of size $\tau\ell$ with $\tau\ell<\theta y$, gives back the baseline model with transformed parameters: Replace $y,$ $\ell,$ and $ \theta$ with $y^\tau:=y-\tau\ell$, $\ell^\tau:=(1-\tau)\ell$, and $\theta^\tau:=(\theta y-\tau\ell)/(y-\tau\ell)$, and all results apply as stated. Economically, Date-2 deposits encumber part of a healthy bank's pledgeable assets, reducing the liquidity it can supply to others without changing the mechanism. (If $\tau\ell\geq\theta y$, healthy banks have no unencumbered liquidity at all and the mechanism shuts down.) We thank a referee for suggesting this extension. (Cash holdings are the mirror image; see footnote~\ref{fn:cash}.)}\label{pp:deposits}$^,$\footnote{\label{fn:cash}{}The model can also be interpreted as including  cash holdings: If cash is perfectly pledgeable, then a bank endowed with  cash $y^{\$}$ and productive assets $y^{\neg\$}$ with pledgeable fraction  $\theta^{\neg\$}$ is just the baseline bank relabelled, with $y=y^{\neg\$}+y^{\$}$ and $\theta y=\theta^{\neg\$}y^{\neg\$}+y^{\$}$, and  all results apply verbatim. Cash is the mirror image of the deposits in footnote~\ref{fn:deposits}: Deposits encumber a bank's pledgeable assets; cash augments them, without changing the mechanism. We thank a referee for  suggesting this interpretation. Section~\ref{s:endogenous theta} lets banks choose $\theta$, and hence, under this interpretation, how much cash to hold.}
    We refer to $\theta y$ as the ``pledgeable part'' of a bank's (non financial) asset and to $(1-\theta)y$ as the ``non-pledgeable'' part.\footnote{The formulation can be interpreted literally in terms of private benefits or cash diversion; for micro-foundations in terms of other agency problems, see, e.g., \cite{DeMarzo-Fishman-2007b} and \cite*{Reallocation}.}  There is universal risk neutrality and no discounting. {}All consumption is at Date~2.\label{pp:date2}

    Banks also have debts both to and from other banks---``interbank liabilities'' and ``interbank claims''---due at Date 2.\label{pp:AE-1c}{}    The (gross) face value of B$_i$'s liability to B$_j$ is denoted by $F_{i \to j} \geq 0$\label{pp:AE-1d}, of all its liabilities to other banks by $F_{i\rightrightarrows} :=\sum_{j\neq i} F_{i\to j}$, and of its claims on other banks by $F_{i\leftleftarrows}:=\sum_{j\neq i} F_{ j\to i}$; $\mathbf{F}_{\rightrightarrows}$ denotes the vector with $i$th element $F_{i \rightrightarrows}$.
    The matrix $\textbf{F} := [F_{i \to j}]_{ij}$ defines the interbank network.\footnote{{}For most of the analysis, we focus on what debts in place do, staying agnostic about where they come from, be it to allocate liquidity, to incentivize monitoring, or to make markets; in Section~\ref{s:stability}, however, we explore endogenous debts.\label{pp:origination}}$^,$\footnote{\label{pp:debt-vs-equity}{}The interbank positions are modeled as plain debts. But many of our results are more general: Relying only on dilutability, they should also apply to cross-holdings of equity as in \cite*{ELLI/GOLU/JACK/14} (see Section~\ref{s:LT}). Other results, like those on network design (Section~\ref{s:exp}), are specific to debt. Indeed, we see the strength of these results as implementing state-contingent transfers with non-contingent contracts, i.e.\ with debt.}
    Following \citetalias{AOT}, we assume throughout that  $F_{i\leftleftarrows}=F_{i\rightrightarrows}$ for all $i$,  i.e.\ banks have zero net interbank positions.\footnote{\label{pp:zero net1}{}The zero-net assumption simplifies much of the analysis and is necessary for Definition~\ref{d:HD}. \label{pp:zero net2}{}But the proofs of many results do not use it, notably those of Propositions~\ref{p:existence}, \ref{p:netting}, and~\ref{p:delta} and, under strong connectedness (directed paths between all pairs of banks), that of Lemma~\ref{l:high debt}. And nothing in our efficiency benchmark relies on it: The planner in Definition~\ref{d:constrained efficiency} is free to choose any transfers, so a network that attains its solution---as the exponential network does---is best among all networks, zero-net or not.} That is a reasonable approximation of reality, as gross interbank positions are often an order of magnitude larger than net.\footnote{\label{fn:balancesheets}{}For example, in 2021, Barclays PLC's net interbank position was about an eighth of its gross, Lloyd's about a fourteenth, and HSBC's about a fifth.
          Specifically, their loans to and from other banks were, respectively, about 13.9 and 16.4, 7.0 and 7.6, and 83.1 and 101.1 billion GBP; see \href{https://home.barclays/content/dam/home-barclays/documents/investor-relations/reports-and-events/annual-reports/2021/Barclays-Bank-PLC-2021-AR.pdf}{home.barclays/content/dam/home-barclays/documents/}{investor-relations/}{reports-and-events/}{annual-reports/}{2021/}{Barclays-Bank-PLC-2021-}{AR.}{pdf}, p.\ 207, \href{https://www.lloydsbankinggroup.com/assets/pdfs/investors/annual-report/2021/2021-lbg-annual-report.pdf}{lloydsbanking}{group.com/}{assets/pdfs/}{investors/}{annual-report/}{2021/}{2021-}{lbg-}{annual-}{report.}{pdf}, pp.\ 207--208, and \href{https://www.hsbc.com/investors/results-and-announcements/annual-report}{hsbc.com/}{investors/}{results-and-announcements/}{annual-report},  p.\ 310.}
        Whereas the face values $F_{i \to j}$ of interbank debt are exogenous for now (see, however, Section~\ref{s:stability}), the associated repayments  are to be determined in equilibrium.{}
B$_i$'s repayment to B$_j$ is denoted by $R_{i\rightarrow j}$, its total repayment to all other banks by $R_{i\rightrightarrows}:=\sum_{j\neq i}R_{i\rightarrow j}$, and its total repayment from all other banks by
$R_{i\leftleftarrows}:=\sum_{j\neq i}R_{j\rightarrow i}$;
$\textbf{R}_{\rightrightarrows}$ denotes the vector with $i$th element
$R_{i\rightrightarrows}$.\label{pp:typo}

     At Date~1, banks' liquidity shocks are realized.\footnote{{}\label{pp:R1-MC-1}Given the two-date structure, all liquidity shocks arrive at Date~1.  What matters is not that they arrive at the same time, but rather while gross positions are still outstanding. That seems reasonable to the extent that interbank debts are long term. In contrast, dilution does nothing for shocks that arrive after the debts they would dilute have matured; see the discussion of random-time liquidity shocks in \cite{donaldson2018resaleable}.}
    We write $\sigma_i =1$ if B$_i$ is shocked and $\sigma_i = 0$ otherwise. There is no other risk, so the ``state,'' which is realized at Date 1, is the profile $\{\sigma_i\}_i =: \bm{\sigma}$; the number of shocked banks is denoted by $S =: \sum \sigma_i$. (We impose no restrictions on the shock distribution.)
    If the shock exceeds B$_i$'s total pledgeable assets,
    \begin{equation} \label{eq:liquidation}
        \theta y + R_{i \leftleftarrows} < \ell \sigma_i,
    \end{equation}
    it repays zero to other banks and gets payoff zero itself.\footnote{\label{fn:seniority}The assumption that interbank repayments are zero follows \citetalias{AOT}, in which liquidity shocks---their ``outside obligations''---are senior to interbank debts.
    The opposite assumption makes the repayment in equation~\eqref{eq:optimal_repayment_cases} discontinuous and can lead to multiple payment equilibria (see Definition~\ref{d:equilibrium} below).
    It does not, however, alter our central insights on the benefits of high indebtedness and connectedness with long-term debt; see \cite{Netting}.}
    The non-pledgeable part of its asset $(1- \theta ) y$ is destroyed.  We say B$_i$ is ``liquidated.''
    Otherwise, it continues to Date~2.

    \label{pp:AE-1a-1} Our definition of liquidation assumes that B$_i$'s interbank liabilities, which are due at Date~2, do not impede its ability to meet the shock, which occurs at Date~1.
    That contrasts with the networks literature (see \citetalias{AOT} and Section~\ref{s:ST}), but is in keeping with papers in corporate finance (the assumptions are  close to \citeauthor*{Donaldson-Gromb-Piacentino-2025}'s (\citeyear{Donaldson-Gromb-Piacentino-2025})).
    It captures the idea that B$_i$ can issue new liabilities of high priority, paid ahead of existing interbank liabilities at Date 2, which are thereby diluted, in that their claims on assets are now subordinated to new creditors'.\footnote{Note that dilution, as used here, need not decrease the value of existing liabilities. Although existing liabilities have a smaller slice of the asset pie, these assets might be more valuable if dilution increases the size of the pie (e.g., by helping to meet a liquidity shock).}
    These new liabilities could represent repos, which have  ``super-senior'' claims on assets in bankruptcy.\footnote{As such, the model could reflect how banks suffering liquidity shocks use super-senior (repo) financing to relax their borrowing constraints, something LTCM, Bear Stearns, and Lehman Brothers all did (or tried to do). (See, e.g., \citeauthor{Jorion-2000} (\citeyear{Jorion-2000}, pp.\ 282--284) on LTCM, \citeauthor*{BearCase} (\citeyear{BearCase}, pp.\ 11--13) on Bear, and \citeauthor{LehmanReport} (\citeyear{LehmanReport} , pp.\ 3 and 9--10) on Lehman.\label{fn:lehman})}$^,$\footnote{\label{fn:priority}Thus, there are two priority classes of debt in our model: new repo-type debt paid first and interbank debts paid next pro rata. This is a good approximation of reality, in which there are two main priority classes: secured debt paid first and unsecured paid next pro rata (see, e.g., \cite{Schwartz-1989}). (\citetalias{AOT} also features two priority classes, but the senior debt is in place at inception.) {}What matters is that the new claim is senior, not that it is debt: A shocked bank can raise cash only against a claim that its existing counterparties cannot absorb first.}
    As a result, banks can borrow against all their pledgeable assets---the pledgeable part $\theta y$ of their long-term assets and their interbank claims $R_{i\leftleftarrows}$.\footnote{\label{fn:microfoundation}{}We interpret $(1-\theta)y$ as cash that insiders can freely
divert at Date~2---a pure private benefit, consumed by the bank---so
creditors can reach at most $\theta y$ in any state, in or out of default.
In some contexts, this is a catch-all for richer agency frictions, such as asymmetric information as \cite{DeMarzo-Fishman-2007b} show formally. However, {}the interpretation can matter for welfare. For example, if the private benefit were to stem from socially wasteful behavior, liquidation need not be inefficient, something we show in Section~\ref{s:extensions} below.}

 At Date~2, a bank that is not liquidated continues to produce $y$. As $(1-\theta)y$ is not pledgeable, its total
    pledgeable assets are therefore
    $\theta y-\ell\sigma_i+R_{i\leftleftarrows}$.\footnote{\label{fn:AE-2}{}Two things are implicit behind this expression. As banks consume only at Date~2 and have no outside technologies,  there is no use for funds raised at Date~1 except to meet liquidity shocks. In principle, we could allow them to borrow, hold cash, and repay. That would net out, so we abstract from it for notational simplicity.

   You could imagine allowing for consumption at Date~1, which could capture, e.g., dividend payouts. We assume that is not allowed because the law places tight limits on
    payouts by distressed firms (e.g., unlawful dividend rules, fraudulent transfer
    law, and bankruptcy clawback provisions).
    That is not inconsistent with distressed firms issuing high priority debt, which is less restricted in practice
    (e.g., DIP financing, priming debt, or liability-management transactions). Moreover, bank debt in particular carries few covenants (\cite{Goyal-2005}) and negative pledge clauses, where present, are deemed ineffective (\cite{Bjerre-1999}; \cite*{Donaldson-Gromb-Piacentino-2025}).

    Our analysis below suggests a possible efficiency rationale for this asymmetry between limits on payouts and dilution, namely that     coinsurance relies on banks being able to dilute existing creditors to raise
    liquidity, but not simply to expropriate value, which would lead the coinsurance mechanism to unravel.}
     If they exceed its debt
$F_{i\rightrightarrows}$, it repays and keeps the remainder of its total assets. Otherwise, i.e.\ if
    \begin{equation}\label{e:default}
        \theta y-\ell\sigma_i+R_{i\leftleftarrows}
        <
        F_{i\rightrightarrows},
    \end{equation}
    it defaults, its creditors get its pledgeable assets, and it keeps its non-pledgeable assets. In summary, if not liquidated, B$_i$'s creditors get
    $\min\big\{\theta y-\ell\sigma_i+R_{i\leftleftarrows},
    F_{i\rightrightarrows}\big\}$
   and B$_i$ gets
    $\max\big\{(1-\theta)y,\,
    y-\ell\sigma_i+R_{i\leftleftarrows}-F_{i\rightrightarrows}\big\}$.

\label{pp:AE-3}{}
    We assume that in default a bank's total pledgeable assets go to its counterparties with its total interbank repayments being divided in pro rata shares:
    \begin{equation}\label{eq:pro_rata}
        R_{i\to j}
        =
        \hat F_{i\to j}R_{i\rightrightarrows},
    \end{equation}
    where $\hat F_{i\to j} := F_{i\to j}/F_{i\rightrightarrows}$ if $F_{i\rightrightarrows} \neq 0$ and $\hat F_{i\to j} := 0 $ otherwise.
    The pro rata rule follows the literature (\cite{Eisenberg-Noe-2001}) and reflects
    bankruptcy law and practice.\footnote{{\cite{csoka2021axiomatization} provides an axiomatic foundation for the pro rata assumption.}} (The identity also holds out of default, as $F_{i \to j} = \hat F_{i \to j} F_{i \rightrightarrows}.$)

    Whereas liquidation entails a deadweight loss (the only one in the model), default is just a transfer from creditors to debtors (except in extensions; see Section \ref{s:risky assets} and Section \ref{s:default costs}).

    To define the equilibrium, we also require that markets clear: The repayments B$_i$ receives from other banks coincide with the repayments other banks make to it:
    \begin{equation}\label{eq:payment_clearing}
    R_{i\leftleftarrows}=\sum_{j\neq i}R_{j\to i} .
    \end{equation}

\begin{definition}[Payment equilibrium]\label{d:equilibrium}
        A \emph{payment equilibrium} is a repayment vector  $\{R_{i\to j}\}_{i\neq j}$ for each state $\bm \sigma$ such that the repayments
        \begin{itemize}
            \item [(i)] are sequentially rational, or, combining the above,
            \begin{equation}
            R_{i\rightrightarrows} \label{eq:optimal_repayment_cases}
            =
            \renewcommand\arraystretch{1.5}
            \left\{
            \begin{array}{cl}
            0 & \text{ if } \theta y-  \ell\sigma_i  + R_{i\leftleftarrows} \leq 0 \,,
            \\
            \theta y-  \ell\sigma_i  + R_{i\leftleftarrows}
            & \text{ if }
            \theta y-  \ell\sigma_i  + R_{i\leftleftarrows}\in (0, F_{i\rightrightarrows} ] \,,
            \\
            F_{i\rightrightarrows} & \text{otherwise}
            \end{array}
            \right.
        \end{equation}
            \item [(ii)] are paid pro rata (equation (\ref{eq:pro_rata})), and
            \item [(iii)] clear the market (equation (\ref{eq:payment_clearing})).
        \end{itemize}

    \end{definition}
\noindent It is convenient to write the sequential rationality condition as
    \begin{equation} \label{e:total repayment}
            R_{i \rightrightarrows} = \max\Big\{ \, 0 \, , \, \min\big\{ \,  \theta y -  \ell\sigma_i  + R_{i\leftleftarrows}
            \, , \,
            F_{i\rightrightarrows}
            \, \big\} \, \Big\} .
    \end{equation}
    That can be combined with the other equilibrium conditions to write a vector fixed point equation:
\begin{equation}
\label{eq:payment_equilibrium}
    \textbf{R}_{\rightrightarrows}
    =
    \Big[
    \min
        \big\{\textbf{F}_{\rightrightarrows} \, , \, \theta y \mathbf{1} - \ell\bm{\sigma}
        + \FT\textbf{R}_{\rightrightarrows}
        \big\}
        \Big]^+ ,
    \end{equation}
    a solution of which is often called the ``clearing vector.'' ($\mathbf{1}$ denotes the vector of $N$ ones: $(1,  ... , 1) \in \mathbb{R}^N$.)

        As liquidation is the only inefficiency in the model, we adopt the following notion of efficiency. \label{pp:AE-5}

\begin{definition}[Efficiency] \label{d:efficiency}
{}One network is (weakly) more \emph{efficient} than another if (weakly) fewer banks are liquidated in equilibrium for every state $\bm \sigma$.

\end{definition}

\noindent This is a strengthening of \citetalias{AOT}'s notions of ``stability''  (fewer liquidations on average for a given number of shocks) and ``resilience'' (fewer liquidations in the worst case scenario) in that if one network is more efficient than another it is more stable and more resilient too. The efficiency ranking is not a total order on the set of networks, but we derive strong enough results that it suffices for our purposes (except in the extension with heterogeneous banks in Section~\ref{s:hetero}, where we modify it; see Definition \ref{d:pp}).

For several of our results, it is useful to define the following network structures:

\begin{definition}[Network typology] \label{d:typology}
    {}A network $\mathbf{F}$ is
    \emph{regular} if $F_{i\rightrightarrows}=F_{j\rightrightarrows}$ for $i \neq j$,
    \emph{symmetric} if $F_{i\to j}=F_{j\to i}$ for $i \neq j$,
     \emph{complete} if $F_{i\to j} > 0$ for $i \neq j$, \emph{uniform complete} if it is complete and  $F_{i\to j}$ is constant,
     \emph{star} if $F_{i \to j} = F_{j \to i}$ is constant for some $i$ and all $j$ and $F_{j \to k} = 0$ for all $j, k \neq i$,
     and
    \emph{ring}  if $F_{i\to j}= 0$ unless $j = i +1$ (mod $N$).

    \end{definition}

    \noindent In words, in a regular network, each bank has the same total liabilities; in a symmetric network, each pair of banks has zero net positions; in a complete network, every pair of distinct banks is directly connected (every ordered pair has a positive exposure); in a uniform complete network, every pair of distinct banks is directly connected by a liability of the same size; in a star network, one bank---the ``core''---owes and is owed the same amount by every other, and the others owe nothing to each other; in a ring network, each bank owes only one other bank and is owed by one other bank in a circle.

It is also useful to define several properties of networks, capturing how closely banks are connected to one another.

\begin{definition}[Delta connectedness] \label{d:delta} A network $\mathbf{F}$ is $\delta$-connected if there is a non-empty proper subset of banks $\mathscr{B}$ such that $\hat F_{i \to j} \leq \delta$ and $\hat F_{j \to i} \leq \delta$ for all $i \in \mathscr{B}$ and $j \in \mathscr{B}^c$.

   It is \emph{connected} if it is not $\delta$-connected for $\delta = 0$.
\end{definition}

\noindent {}In words, $\delta$ measures the strength of the weakest link between two parts of the network. Small $\delta$ means the network can be partitioned into two groups connected only by weak liabilities. At $\delta=0$ the two groups have no liabilities between them at all, so a network that is \emph{not} $0$-connected is connected in the usual sense of graph theory: There is an undirected path between every pair of banks. Without loss, we focus on connected networks for the following definitions and results unless otherwise stated.

\begin{definition}[Harmonic distance] \label{d:HD}
For a regular network $\mathbf{F}$, the \emph{harmonic distance} from B$_i$ to B$_j$
is the solution to $d_{i \to j} :=
        1+\sum_{k\neq i} d_{i\to k}\hat F_{k\to j}$ for  $i \neq j$  and $d_{i \to i} = 0$.

        \end{definition}

\noindent{}In words,  $d_{i \to j}$ is the liability-weighted distance from $i$ to $j$, which captures how easily liquidity (or the lack thereof) can flow from B$_i$ to B$_j$:  It is small when payments travel in few steps over links that carry a large share of each bank's liabilities. What matters is the share, not the count: A bank connected to many others by thin links can be ``far'' from all of them.

\begin{definition}[Bottleneck parameter] \label{d:bottleneck}
    For a regular network $\mathbf{F}$, the
     \emph{bottleneck parameter}  is
    \begin{equation}
        \beta =\min _{\varnothing\neq\mathscr{B}\subsetneq\{1,\ldots,N\}} \sum_{i\in \mathscr{B}} \sum_{j\in \mathscr{B}^c}\frac{\hat F_{i\to j}}{|\mathscr{B}||\mathscr{B}^c |}.
    \end{equation}

\end{definition}
\noindent {}In words, $\beta$ is the narrowest pipe out of any group of banks: A network has low $\beta$ if one of its components has relatively low liabilities to the rest of it, so that little can drain out of that component. It is similar to $\delta$ above, but directional (in that it is silent about the liabilities from the rest to the component).

\section{Benchmark Without Dilution} \label{s:ST}

Here we consider the network in which the interbank liabilities $F_{i \to j}$ are due at Date 1 instead of Date 2.  This benchmark helps us both to compare our model to the literature and to contrast our results to those therein.

Now interbank liabilities, being due immediately, cannot be diluted with new debt at Date 1. Thus B$_i$ is liquidated if its pledgeable assets are insufficient to cover not only its liquidity needs $\ell\sigma_i $ but also its interbank liabilities $F_{i \rightrightarrows}$, or
    \begin{equation} \label{e:ST default}
    \theta y +    R_{i\leftleftarrows} <
        \ell\sigma_i + F_{i\rightrightarrows}.
    \end{equation}

    \noindent   This condition for liquidation coincides with that for default.
       In that respect, the benchmark contrasts with the baseline, in which the conditions are different.
       Nonetheless, the equations for the clearing vector are the same in both versions, as the liquidation value at Date~1 coincides with the pledgeable value at Date~2: Banks that would default at Date 2 in the baseline model make the same repayment when they are liquidated at Date~1 here. Thus we do not need to adjust the equilibrium definition here. Only efficiency changes.

    Without the distinction between liquidation and default, this benchmark coincides with \citetalias{AOT}.
\begin{lemma}[Isomorphism between benchmark and \citetalias{AOT}] \label{l:isomorphism} {}There is an isomorphism  between payment equilibria in our  benchmark and those in AOT in the case in which long-term assets are fully destroyed by default and cash holdings are zero.\footnote{This is the case under which \citetalias{AOT} derive most of their results.}{}The payment equilibria in the two models have the same liquidation outcomes; hence, they have generically the same efficiency ranking of networks.\footnote{Here and hereinafter, ``generically'' means everywhere except on the ``knife-edge'' set of parameters on which the equilibrium is not unique; see footnote~\ref{fn:knife-edge} below for the specifics. The efficiency ranking in the lemma holds there too under any selection rule applied to both models.}
\end{lemma}

\label{pp:AE-4}

\noindent (The isomorphism between our benchmark and \citetalias{AOT}, while mechanical, was unexpected to us, as the models seemed different prima facie. Our model seemed to be about liquidity, theirs about solvency. We now see both as about both: In both, a shock decreases total asset value (solvency) and might not be met only because long-term assets are not fully pledgeable (liquidity).)

We now restate, and sometimes strengthen, several of \citetalias{AOT}'s results, focusing on those that contrast with our results on the long-term debt network.

\begin{lemma}[Netting in benchmark] \label{l:ST netting}
        For any network $\mathbf{F}$ and any $\alpha > 1$, $\alpha\mathbf{F}$ is generically less efficient than $\mathbf{F}$.
\end{lemma}

\noindent This generalizes \citetalias{AOT}'s Proposition 3 (p.~574) in two ways: To an arbitrary number of shocked banks (they prove it for one) and beyond regular networks.
It says that less debt is a good thing. Indeed, it would be better to have none whatsoever ($\alpha = 0$).
Intuitively, when one bank defaults on its liability to another, the other finds it harder to pay its liability to yet another.
So distress propagates from shocked banks to otherwise healthy ones, especially when debts are high $(\alpha > 1)$.

We now turn to network connectedness.

\begin{lemma}[Delta connectedness in benchmark] \label{l:ST delta}
Suppose $\ell$ is sufficiently large and exactly one bank is shocked ($\lvert \bm \sigma \rvert= 1$).

\begin{itemize}

    \item[(i)] A ring network with $F > (N-1) \theta y$ is the least efficient among all networks.

    \item[(ii)] Any $\delta$-connected network with $N \delta \max_k F_{k\rightrightarrows}<\theta y$ is strictly more efficient than the ring.

\end{itemize}

\end{lemma}

\noindent This strengthens some of the statements in \citetalias{AOT}'s Proposition 6 (p.\ 577--578) by adapting them to our notion of efficiency (Definition \ref{d:efficiency}) and dropping regularity. It says that less connectedness is a good thing. The ring network, in which every bank has a large exposure to another, is the worst. It is better to weaken the exposures in the sense of lowering delta connectedness. Intuitively, $\delta$ captures how much risk can transmit between two components, so, unlike in the ring network, risk cannot spillover from one to another when $\delta$ is low.

We now show that there is a ``default radius'' around a shocked bank.

\begin{lemma}[Default radius in benchmark] \label{l:ST radius}
    Let $\mathbf{F}$ be a regular network with $F_{i \rightrightarrows} \equiv F$ and suppose that exactly one bank, say B$_j$, is shocked and does not meet its liquidity shocks ($R_{j\rightrightarrows}=0$).  Define $d^{ST}:=\frac{F}{\theta y}$.
    \begin{enumerate}
        \item[(i)]  If $d_{j\to i}<d^{ST}$, then B$_i$ is liquidated.
        \item[(ii)]  If all banks are liquidated, then $d_{j\to i}<d^{ST}$ for all $i$.
    \end{enumerate}
\end{lemma}

\noindent This is \citetalias{AOT}'s Proposition 8 (p.\ 579). It says that the harmonic distance $d$ is, in a sense, the right measure of one bank's exposure to another, in that it captures exactly whether a shocked bank's distress will transmit to an otherwise healthy bank through the network.
It defines a radius around a shocked bank within which all banks are liquidated.

 \citetalias{AOT}  link the harmonic distance $d$  to the bottleneck parameter $\beta$ using Markov chains.
 They show, roughly, that $d_{i \to j}$ is the mean hitting time of a Markov chain from state $i$ to $j$ and that the bottleneck parameter is closely related to the ``conductance'' of a graph, which measures how hard it is for a Markov chain on a graph to leave a set of nodes.
    Hence the next result:

\begin{lemma}[Bottleneck connectedness in benchmark] \label{l:ST bottleneck}

    Suppose the conditions of Lemma \ref{l:ST radius} are met and that, additionally, the network $\mathbf{F}$ is symmetric.
    Define  $\beta^{ST} := 4 \sqrt{\frac{\theta y}{NF}}$ and  $\beta_{ST} := \min\left\{ \, \frac{\theta y}{2 N F}\, , \, 1 \right\}$.
    \begin{enumerate}
        \item[(i)] If $\beta>\beta^{ST}$, then all banks are liquidated.
        \item[(ii)] If $\beta<\beta_{ST}$, then at least one bank is not liquidated.
    \end{enumerate}

\end{lemma}

\noindent This is \citetalias{AOT}'s Corollary 2 (p.\ 581).
It captures the idea that if all banks are closely connected then risk is so easily transmitted to other banks that a shock at one can lead all to fail, whereas if they are not, it cannot. Specifically, if at least two components are not closely connected, so $\beta$ is small, risk in one of them cannot spread to the other.

\section{Properties of Dilutable Debt Networks}  \label{s:LT}

We now turn to the properties of networks in our baseline model, in which interbank debt is dilutable. For each result in the non-dilutable benchmark, we prove a counterpart with the opposite sign: Whereas indebtedness and connectedness do harm when debt cannot be diluted, they do good when it can. The reason is that dilution changes what the network transmits---rather than propagating liquidity shortages from shocked banks, it lets liquidity flow from healthy banks to shocked ones.

We start with existence and uniqueness.

\begin{proposition}[Existence and uniqueness] \label{p:existence}
        For any network $\mathbf{F}$, a payment equilibrium exists and is generically unique.
    \end{proposition}

\noindent Existence, as ever, follows from Tarski's fixed-point theorem.  {}Uniqueness comes from payments being pinned down by resources, not expectations:  Default is only a transfer (equation~(\ref{e:default})), in which a bank hands its creditors all of its pledgeable resources, so its repayment moves at most one for one with the repayments it receives (equation \eqref{e:total repayment})---a bank with a cushion absorbs part of any shortfall and a bank already repaying nothing passes on none of it.
Any shortfall thus dissipates as it circulates through the network, no matter what banks expect of one another.
\footnote{{}In related models, an equilibrium in which all banks in a cycle repay, and hence each has the resources to repay, can coexist with a default equilibrium with the opposite outcome.
There, default is costly: A counterparty's default inflicts a discrete loss on its creditors, so a shortfall compounds as it circulates, and pessimism around a cycle of claims validates itself (\cite{Rogers-Veraart-2013}; \cite{Jackson-Pernoud-2024}); relatedly, in  \cite*{Roukny-Battiston-Stiglitz-2018} banks' default conditions depend on their counterparties' defaults around a credit cycle.
Liquidation destroys value in our model too, but only the non-pledgeable part $(1-\theta)y$, which creditors could never have reached, so not even liquidation inflicts a discrete loss on creditors.
(That relies on liquidity shocks being senior to interbank debts; see footnote~\ref{fn:seniority}.)}
{}That said, there is a knife-edge case in which shocks exactly exhaust the total resources so repayments exactly cancel out, leading to multiple possible repayments supporting the same outcome.\footnote{\label{fn:knife-edge}{}For a connected network, the knife-edge condition is $N\theta y=S\ell$.
More generally, in a disconnected zero-net-position network, multiplicity can arise only if $|K|\theta y=S_K\ell$ for some connected component $K$ where $S_K$ is the number of shocked banks in $K$.
To see where it comes from, suppose two banks, say B$_i$ and B$_j$, owe each other $F$ and have no outside liabilities. Knife-edge multiplicity can arise in the case in which they both default.
Then, from equation~\eqref{e:total repayment}, $R_{i \to j} =\theta y-\ell \sigma_i +R_{j \to i}$ and $R_{j \to i} =\theta y-\ell \sigma_j +R_{i\to j}$. Adding equations says any solution must satisfy $2 \theta y = (\sigma_i + \sigma_j) \ell$, the two-bank version of the knife-edge condition.
In this case, any pair of repayments such that the difference in payments equals the difference in liquidity shortfalls, $R_{i \to j} - R_{j \to i} = \theta y - \sigma_i \ell =  \sigma_j \ell-\theta y$ is a solution. In particular, if $\sigma_i  = \sigma_j = 1$ then that the last condition holds for any $R_{i \to j} = R_{j \to i}\in[0,F]$:
The payments cancel out, so it does not matter what they are.  }

We now turn to the role of debt levels.
\begin{proposition}[Netting] \label{p:netting}
        For any network $\mathbf{F}$, $\alpha\mathbf{F}$ is more efficient than $\mathbf{F}$ whenever $\alpha > 1$.

\end{proposition}

\noindent This is the counterpart of Lemma \ref{l:ST netting}.
It says that  \emph{more} (dilutable) debt is a good thing here.\footnote{Although, per the result, increasing debt cannot hurt if the relative debts stay the same---if $\alpha \mathbf{F}$ increases but $\mathbf{F}$ stays constant---it can if they change---i.e.\ if some debts increase and some do not---as illustrated by how making a network ``more symmetric'' can decrease efficiency (in Section \ref{s:exp}).}
{}The reason is that high debts $(\alpha > 1)$ allow shocked banks to weather shocks.
    To see why, suppose two banks, B$_i$ and B$_j$, have perfectly off-setting debts, owing each other the same amount: $\alpha F_{i \to j} = \alpha F_{j \to i} = \alpha F.$
    Suppose B$_i$ suffers a liquidity shock and B$_j$ does not.
    In this case, B$_i$ gets liquidity by pledging both its own assets and the claim it has on B$_j$ to raise $\theta y + R_{j \to i }$.
    The larger B$_i$'s claim $\alpha F$ on B$_j$ is, the higher its value $R_{j \to i}$ is, and the more liquidity it can raise.
    \label{pp:AE-1b}

      True, increasing $\alpha$ gives B$_i$ not only  a larger claim to pledge on the left-hand side of its balance sheet, but also a larger liability to repay on the right.
      However, B$_i$, being shocked, is likely to default, whereas, not being shocked, B$_j$ is not.\footnote{\label{pp:sc7}{}Indeed, the shocked bank generally defaults, whereas the not-shocked one repays in full (cf.\ equation \eqref{e:default}): The shocked bank's default passes B$_j$ a shortfall $\ell - \theta y$, which B$_j$'s own slack $\theta y$ covers given $\ell < 2\theta y$ (the algebra is that of footnote~\ref{fn:knife-edge} with one shock). Moreover, the mechanism is self-enforcing, in that, as banks cannot avoid being diluted, they are committed ex ante to transfers they would prefer not to make ex post. \cite{Leitner-2005} uncovers another self-enforcing mechanism to transfer liquidity in financial networks: Healthy banks commit to transfer liquidity to distressed ones by exposing themselves to their default through the network.}
      {}Therefore, the claim B$_i$ has on B$_j$ is worth more than the liability it has to it. So the apparently offsetting debts are a transfer of value from the not-shocked to the shocked bank, providing insurance for the shocked bank.

            This transfer benefits whichever bank is shocked at the expense of whichever is not shocked.
            Liquidity gets transferred from the bank with liquidity to the bank that needs it:  Zero-net long-term debt has positive net present value.
            Per the result, increasing debt increases overall efficiency.\footnote{Efficient ``dilutable debt'' also appears in \cite{Diamond-1993}, \citeauthor*{Paradox} (\citeyear{Paradox}, \citeyear{Donaldson-Gromb-Piacentino-2025}), and  \cite{Hart-Moore-1995}.}

            That is not so in the benchmark (and most other models), in which zero-net debts have zero NPV at most (Lemma \ref{l:ST netting}).
            The reason is that short-term debt, being due right away, cannot be diluted:  The claims on the left side of bank balance sheets increase debt capacity but the liabilities on the right decrease it; they cancel each other out at best.

    We turn to network connectedness next.

\begin{proposition}[Delta connectedness]
\label{p:delta}

Suppose exactly one bank is \emph{not} shocked ($\lvert \mathbf{1} - \bm \sigma \rvert= 1$).

\begin{itemize}

    \item[(i)] A ring network with $F > \theta y$ is the most efficient among all networks.

    \item[(ii)] If $\ell \leq 2 \theta y$,\footnote{This is the condition for there to be enough liquidity in total to save a shocked bank; otherwise no network does better than the empty network (cf.\ equation~\eqref{e:L*}).} for any $\delta$, there is a $\delta$-connected network that is strictly less efficient than a ring.
\end{itemize}

\end{proposition}

This is the counterpart of Lemma \ref{l:ST delta}. It says that \emph{more} connectedness is a good thing---connectedness in the sense of the bottleneck $\delta$, which captures the strength of exposures between groups of banks, not their number. The ring network, in which every bank has a large exposure to another, is the best (given a single not-shocked bank). It is better not to weaken interbank exposures in the sense that making $\delta$ small could lead to a strictly worse outcome. Intuitively, $\delta$ captures how easily a shocked bank can raise liquidity using claims on healthy banks as collateral. When $\delta$ is low, that mechanism is less effective. The ring wins here because it sends the not-shocked bank's liquidity to a single shocked bank rather than splitting it---the principle the exponential network generalizes to every state (Section~\ref{s:exp}).

We turn to the long-term-debt counterpart to short-term debt's default radius around a shocked bank. It is a ``salvation radius'' around a healthy bank.

\begin{proposition}[Salvation radius] \label{p:radius}
    Let $\mathbf{F}$ be a regular network with $F_{i \rightrightarrows} \equiv F$ and suppose that exactly one bank, say B$_j$, is \emph{not} shocked and that it does not default.  Define $d^{LT}:=\frac{F}{\ell - \theta y}$.
    \begin{enumerate}
        \item[(i)]  If $d_{j\to i}<d^{LT}$, then B$_i$ is not liquidated.
        \item[(ii)]  If no bank is liquidated, then $d_{j\to i}\leq d^{LT}$ for all $i$.
    \end{enumerate}

\end{proposition}

\noindent This is the counterpart of Lemma \ref{l:ST radius}. It says that the harmonic distance captures not only how defaults transmit from shocked to healthy banks when debt cannot be diluted, but also how the option to dilute allows shocked banks to raise liquidity against claims on not-shocked banks. It defines a radius around a not-shocked bank within which no bank is liquidated.

\citetalias{AOT}'s link between the harmonic distance and the bottleneck parameter also applies to our long-term debt network---it relies on only network structures, not equilibrium behavior. Hence we have the next result:

\begin{proposition}[Bottleneck connectedness] \label{p:bottleneck}
    Suppose the conditions of Proposition \ref{p:radius} are met and that, additionally, the network $\mathbf{F}$ is symmetric.
    Define  $\beta^{LT} := 4 \sqrt{\frac{\ell - \theta y}{NF}}$ and  $\beta_{LT} := \min \left\{ \,  \frac{\ell - \theta y}{2 N F} \, , \, 1 \, \right\}$.
    \begin{enumerate}
        \item[(i)] If $\beta>\beta^{LT}$, then no bank is liquidated.
        \item[(ii)] If $\beta<\beta_{LT}$, then at least one bank is liquidated.
    \end{enumerate}

\end{proposition}

 \noindent This is the counterpart of Lemma \ref{l:ST bottleneck}.     It captures the idea that if all banks are closely connected, one bank's excess liquidity can flow through the system  to save all banks and, conversely, if they are not, it cannot.            Specifically, if at least two components are not closely connected ($\beta$ is small), the banks in one can raise little liquidity by diluting their liabilities to banks in the other. What they can raise can be so limited that they end up unable to save themselves from liquidation.

\section{Efficiency and Exponential Networks}  \label{s:exp}

Here we define constrained efficiency and construct a class of networks---the ``exponential networks''---that implement it. We conclude the section with an example contrasting complete and exponential networks, which captures the main ideas.

\begin{definition}[Planner's problem and constrained efficiency]\label{d:constrained efficiency} The \emph{planner's problem} is to find a set of transfers $\{t_i\}_i$ for each $\bm \sigma$, with $t_i$ paid to each B$_i$, to  minimize the number of liquidated banks $ \big\lvert \big\{ i : \theta y + t_i - \ell \sigma_i < 0 \big\} \big \vert$ subject to each bank's liquidity constraint $ t_i \geq  \min \{\ell \sigma_i -\theta y,0\}$ and to liquidity being conserved $\sum t_i \leq 0$.

A network is \emph{constrained efficient} if the equilibrium is no less efficient than the planner's solution.
\end{definition}

\noindent \label{def 7}{}In words, the planner wants to minimize the number of liquidated banks by transferring liquidity within the system. It must respect the aggregate resource constraint and the limited pledgeability friction---it cannot raise more from any one bank than the net liquidity it has (that is no more than $\theta y$ from a not-shocked bank and zero from a shocked one, per the constraint $t_i \geq \min\{\ell \sigma_i - \theta y, 0\}$, given $\ell > \theta y$).
But {}it has no other constraints. (The problem does not even depend on the network $\mathbf{F}$ as the planner's transfers can replicate---or undo---any interbank payments.)
The solution, characterized next, is triage: Tax every healthy bank to the hilt, fully rescue as many shocked banks as that budget allows, and abandon the rest, as spreading liquidity thinly saves no one.

\begin{lemma}[Constrained efficiency] \label{l:constrained efficiency} A network is constrained efficient if the number of liquidated banks is
    \begin{equation} \label{e:L*}
        L^* : =  \max   \left\{
                            \, 0 \, , \,
                          \left\lceil \frac{ S \ell - N \theta y }{\ell - \theta y} \right\rceil
                           \,
                        \right\}
    \end{equation}
    for each state $\bm \sigma$ where $S$ banks are shocked.
\end{lemma}

\noindent It turns out that the social planner should generally raise as much liquidity as possible from each not-shocked bank, levying the tax $-t_i = \theta y$ if $\sigma_i = 0$, and transfer shocked banks either just enough liquidity to survive or none at all, i.e., either $t_i = \ell - \theta y$ or $t_i = 0$ if $\sigma_i = 1$.
Using $S$ and $L$ to denote the numbers of shocked and liquidated banks, the planner's budget constraint says that the total subsidy---the transfer  $\ell - \theta y$ to each of the $S-L$ shocked banks that is not liquidated---must be less than the total tax---the transfer $\theta y$ from each of the $N-S$ banks that is not shocked:
\begin{equation} \label{e:residual liquidity}
    (S - L)(\ell - \theta y ) \leq (N - S) \theta y .
\end{equation}
Solving for the smallest non-negative integer $L$ that satisfies the above gives the result.

That argument points to two key properties of the planner's solution, both of which help avoid ``wasting liquidity'':

    \begin{enumerate}

        \item[(i)] It extracts the maximum tax from not-shocked banks, since they are not liquidated anyway.

        \item[(ii)] It gives nothing to liquidated banks, since, analogously, they are liquidated anyway.

    \end{enumerate}

\noindent We aim to construct a network with both properties (Proposition \ref{p:constrained efficiciency} below). We now build up to it in steps, showing how to achieve one and then the other, starting with the first:

\begin{lemma}[High debt mutualizes assets] \label{l:high debt} Let $\mathbf{F}$ be a connected network. If $\alpha$ is sufficiently large, then in the equilibrium of $\alpha \mathbf{F}$ either (i) all not-shocked banks make the maximum net payment,
     $
        {R}_{i \rightrightarrows}- {R}_{i \leftleftarrows} = \theta y,
    $
    or (ii) no bank is liquidated.

\end{lemma}

\noindent This says that if debts are sufficiently high and liquidity is sufficiently scarce (in the sense that at least one bank is liquidated), then each not-shocked bank provides the maximum amount of liquidity.
         Intuitively, to increase interbank debts is to make each bank's assets a larger fraction of others' balance sheets---it ``mutualizes'' the banking system, making each bank more like the whole system.\footnote{This role of  default in facilitating a socially efficient transfer of liquidity contrasts with the literature, in which it typically constitutes a social cost, facilitating  rent extraction at best (see \cite{Farboodi-2021} and \cite{Perotti-Spier-1993}).}
         As a result, when liquidity is scarce overall, no surviving bank retains excess liquidity.
         In contrast, if liquidity is not scarce then no bank is liquidated if debt levels are sufficiently high:

\begin{corollary}[First best with high debt for small shocks] \label{c:high debt}
Let $\mathbf{F}$ be a connected network and suppose that $N  \theta y > S \ell $ (i.e.\ $L^* = 0$). If $\alpha$ is sufficiently large, then no bank is liquidated in the equilibrium of $\alpha \mathbf{F}$.

\end{corollary}

The second property---that liquidated banks are transferred nothing---points to how the planner's allocation is necessarily discriminatory: It prioritizes some shocked banks over others. (If in contrast, it allocated the excess liquidity equally among all shocked banks it ends up saving none of them unless it can save them all. As we illustrate in an example below, that symmetry makes the complete network ``robust yet fragile,'' per \citetalias{AOT}'s result.)
That suggests that whenever liquidity is scarce ($S$ is large) a network must be asymmetric to be (constrained) efficient.
The following definitions characterize ways in which a network can be asymmetric.

\begin{definition}[Assortativity]\label{d:assortativity}
A network $\mathbf{F}$ is \emph{assortative} if there is a permutation
$\pi$ on $\{1,\ldots,N\}$ such that
\begin{equation}\label{eq:assortativity}
    F_{i\to\pi(k)}\geq F_{i\to\pi(l)}
\end{equation}
for all $k<l$ and all $i\notin\{\pi(k),\pi(l)\}$.
\end{definition}

\noindent {}In words, a network is assortative if whenever one bank owes more to B$_k$ than to B$_l$, every other bank owes weakly more to B$_k$ than to B$_l$. Assortativity thus allows us to rank banks by the size of the liabilities they are owed: The network has a single weak hierarchy on which all banks agree, with ties allowed,  rather than each bank owing most to its own neighbors as in, e.g., the ring. The next definition quantifies how steep that hierarchy is.

\begin{definition}[$s$-dominance] \label{d:s-dominance} For a network $\mathbf{F}$, B$_i$'s liabilities are \emph{$s$-dominated} for $s\in(0,1)$, if there is a permutation $\pi_i$ on $\{1, ..., N\}$ with $\pi_i(i) =i$, such that
    \begin{equation}\label{eq:s-dominance}
        F_{i\to \pi_i(j+k)}\leq s^k F_{i\to \pi_i(j)}
    \end{equation}
    for all $j$ and $k \geq 0$ such that $j \neq i$ and $j +k \neq i$.
\end{definition}

\noindent {}In words, for $s < 1$, B$_i$'s liabilities to others decay rapidly---its second largest liability is only at most a fraction $s$ of its largest, and so on: The pecking order is steep, with each bank's payments concentrated at the top of its ranking rather than spread down it. (The permutation $\pi_i$ in the definition just ranks the debts by size.)

Together the definitions above define what we call exponential networks:

\begin{definition}[Exponential networks] \label{d:exponential networks} A network is an \emph{exponential network} (with base $s$) if it is connected, it is assortative, and B$_i$'s debts are $s$-dominated for all $i$.

\end{definition}

\noindent {}In words, every bank's debt to B$_2$ is only at most a fraction $s$ of its debt to B$_1$ and so on---the permutation $\pi_i$ in Definition \ref{d:s-dominance} can be chosen to rank each bank's creditors the same way (assumed w.l.o.g.\ to be the same as their index ordering per Definition \ref{d:assortativity}), creating a pecking order:  Liquidity flows to the top-ranked banks first, rather than being divided evenly and wasted, as we show below.
The exponential network generates an approximately exponential distribution of bank asset size, $y + F_{i \leftleftarrows}$.\footnote{Another paper in which intermediation networks give rise to an endogenous bank size distribution is \cite*{Farboodi-et-al-2017}.}

The next two results characterize the payments made to liquidated banks in an exponential network.

    \begin{lemma}[Controlling relative payments to liquidated banks]\label{lm:exponentially_dominated}
    Suppose $\mathbf{F}$ is an exponential network with base $s$ and
    B$_{i^*}$ is liquidated. For every lower-ranked bank B$_j$, $j>i^*$,
    \begin{equation}
        R_{j\leftleftarrows}\leq s^{j-i^*}R_{i^*\leftleftarrows}.
    \end{equation}
\end{lemma}

\noindent  This says that when debt levels in the network  are exponentially controlled (given $s$-dominance), so are the repayments to liquidated banks in equilibrium.

The previous result says that $s$ controls the relative payments among liquidated banks. The next says that it controls the total payment to all of them.

        \begin{lemma}[Controlling total payments to liquidated banks]\label{lm:total_liquidity}
        Suppose $\mathbf{F}$ is an exponential network with base $s$.
        If at least one bank is liquidated in equilibrium, then
            \begin{equation}
                \sum_{i\in\mathscr{L}} R_{i \leftleftarrows} < \frac{\ell-\theta y}{1 - s},
            \end{equation}
            where $\mathscr{L}$ denotes the set of liquidated banks.

                \end{lemma}

            \noindent This says that the total transfer to liquidated banks can be made arbitrarily close to a single shocked bank's liquidity shortfall, $\ell -\theta y$, by making $s$ sufficiently small.
             Intuitively, for small $s$, the liquidity wasted by transferring it to banks that end up being liquidated anyway would barely have been enough to save even a single one of them (per property (ii) above).

To sum up, if debts are large then no liquidity is wasted on banks that would not have been liquidated anyway (Lemma \ref{l:high debt}) and if the network is exponential then little is wasted on those that would have been (Lemma \ref{lm:total_liquidity}). The next result builds on these findings to show how to achieve constrained efficiency.

\begin{proposition}[Efficiency of exponential networks] \label{p:constrained efficiciency}
Define
        \begin{equation} \label{e:s*}
            s^* := 1 - \frac{ \ell - \theta y}{N \theta y - S \ell + (1 + L^*) ( \ell - \theta y)}.
        \end{equation}
Let $\mathbf{F}$ be an exponential network with base $s \leq s^*$. For $\alpha$ sufficiently large, $\alpha \mathbf{F}$ is constrained efficient as long as  $\frac{S \ell - N \theta y}{\ell -\theta y}$ is not an integer.

\end{proposition}

\noindent In words, an exponential network with a rapidly decaying distribution of debts is optimal in all but a knife-edge case (i.e.\ all but the case in which $\frac{S\ell - N \theta y}{\ell - \theta y}$ is an integer).

This result has a limitation when achieving constrained efficiency requires almost all of the liquidity available, i.e.\ when the slack in inequality \eqref{e:residual liquidity} becomes small for $L = L^*$ ($(N - S) \theta y - (S-L^*) ( \ell - \theta y) \to 0$).
In this case, $s^*$ becomes small so the largest bank in the exponential network is many times larger than the second largest one, and so on. Moreover, in the limit when it equals zero, the network does not achieve constrained efficiency---that is the integer case from the proposition.
But the next result suggests this limitation is not too worrisome if we accept a weaker notion of efficiency:

\begin{proposition}[Approximate efficiency of exponential networks] \label{p:exponential almost} Let $\mathbf{F}$ be an exponential network with base $s \leq 1/2$. For $\alpha$ sufficiently large, $\alpha \mathbf{F}$ is ``almost constrained efficient'' in that at most $L^* + 1$ banks are liquidated.

\end{proposition}

 Finally, we point out that even in the knife-edge case in which the exponential network does not achieve constrained efficiency, no other fully connected network does either.

\begin{lemma}[Inefficiency of other networks] \label{lm:inefficiency other networks} Suppose $\frac{S \ell - N \theta y}{\ell - \theta y}$ is an integer, $S<N$, and $S \ell > N \theta y$ and that $\mathbf{F}$ is fully connected in that $F_{i \to j} > 0$ for all $i\neq j$. The equilibrium is not constrained efficient.

\end{lemma}

\noindent This result suggests that the exponential network is the ``best'' no matter the parameters.\footnote{As the proof of Lemma \ref{lm:inefficiency other networks} implies, a network that is not fully connected could do better for some realizations of $\bm{\sigma}$ but not for all.}

\textbf{Efficiency example.} Finally, we consider an example to illustrate (i) how a complete network leads to inefficient liquidation by allocating liquidity to banks that end up being liquidated in equilibrium and (ii) how an exponential network solves the problem.

The illustration requires that shocks are large enough that at least one bank is liquidated in the constrained-efficient outcome (otherwise a complete network can save all banks). Hence we consider three banks, two of which are shocked: There are $N=3$ banks with assets $y=2$, a fraction $\theta = 1/2$ of which is pledgeable. Exactly two banks suffer shocks, $S \equiv \sum \sigma_i = 2$, of size $\ell$. We assume that $\ell = 8/5$ so that in the constrained-efficient outcome exactly one bank is liquidated in each state: $ \ell < 3 \theta y <  2 \ell$ (i.e.\ $8/5 < 3 < 16/5$).
Observe that covering either shocked bank's  liquidity shortfall requires at least $3/5$ of the other's surplus: $\ell - \theta y  = 3 \theta y / 5$.

We start with the complete network as benchmark and show that both shocked banks are always liquidated. Then we illustrate how the exponential network saves one of them, achieving the constrained-efficient outcome.

\emph{Complete network benchmark.} Suppose each bank has total liabilities $F_{i \rightrightarrows} \equiv F$ to others ($F/2$ to each of the other two). For each state $\bm \sigma$, equation \eqref{eq:optimal_repayment_cases} gives the system the clearing vector $\mathbf{R}_{\rightrightarrows} = [R_{i \rightrightarrows}]_{i}$ must solve (the full matrix of equilibrium repayments is then given by the pro rata shares $R_{i \to j} = \hat F_{i \to j} R_{i \rightrightarrows}$ by equation \eqref{eq:pro_rata}).
When the shocked banks are B$_1$ and B$_2$, the system is:
\begin{equation}
\renewcommand\arraystretch{1.5}
   \left\{ \begin{array}{l}
        R_{1\rightrightarrows}    =
            \max\Big\{ \, 0 \, , \, \min\left\{ \,  - \frac 3 5  + \frac 1 2 R_{2 \rightrightarrows} +  \frac 1 2 R_{3 \rightrightarrows}
            \, , \,
            F
            \, \right\} \, \Big\},
            \\
        R_{2\rightrightarrows}
        = \max\Big\{ \, 0 \, , \, \min\left\{ \,   -  \frac 3 5  + \frac 1 2 R_{1 \rightrightarrows} +  \frac 1 2 R_{3 \rightrightarrows}
            \, , \,
            F
            \, \right\} \, \Big\},
            \\
        R_{3\rightrightarrows}
        =
        \max\Big\{ \, 0 \, , \, \min\left\{ \,  \ 1   + \frac 1 2 R_{1 \rightrightarrows} +  \frac 1 2 R_{2 \rightrightarrows}
            \, , \,
            F
            \, \right\} \, \Big\}.
    \end{array}
    \right.
\end{equation}
The first two lines of the system illustrate the problem with the complete network. The repayments from the bank with excess liquidity, B$_3$, are allocated equally between the two banks with a liquidity shortfall, B$_1$ and B$_2$---each  gets $\frac 1 2 R_{3 \rightrightarrows}$ (plus a symmetric transfer from the other shocked bank, $\frac 1 2 R_{2 \rightrightarrows}$ or $\frac 1 2 R_{1 \rightrightarrows}$).
But each bank needs more than that to survive. Allocating scarce resources equally means that no one has enough.

As the system is symmetric, the problem is analogous when the other pairs of banks are shocked.  Solving it gives $R_{i \rightrightarrows} = 0$ if $\sigma_i =1$ and $R_{i \rightrightarrows} = \min \{F, 1\}$ otherwise, implying that both shocked banks are always liquidated (equation \eqref{eq:optimal_repayment_cases}).

\emph{Exponential network.}
Now we turn to an exponential network with base $s = 1/2$:
\begin{equation} \label{e:F example}
\renewcommand\arraystretch{.75}
\textbf{F}=
\left[
\begin{array}{ccc}
    0 & 2 & 1 \\
    2 & 0 & \frac{1}{2} \\
    1 & \frac{1}{2} & 0 \\
\end{array}
\right].
\end{equation}
Note that the off-diagonal entries in each row and column are decreasing (assortativity per Definition \ref{d:assortativity}) and that each is  at most a fraction $1/2$ of the previous ($s$-dominance per Definition \ref{d:s-dominance}).
We can compute each bank's total liabilities ${F}_{i \rightrightarrows}$ (the row sums of $\mathbf{F}$) and the fraction of its payments it makes to each other bank ($\mathbf{F}$ normalized by $F_{i \rightrightarrows}$):
\begin{equation} \label{e:F and F hat}
\renewcommand\arraystretch{.75}
\textbf{F}_{\rightrightarrows}=
\left[
\begin{array}{c}
    3\\
    \frac{5}{2}\\
    \frac{3}{2}
\end{array}\right]
\ \ \& \ \
\mathbf{\hat F} =
\left[
\begin{array}{ccc}
            0 & \frac{2}{3} & \frac{1}{3} \\
            \frac{4}{5} & 0 & \frac{1}{5} \\
            \frac{2}{3} & \frac{1}{3} & 0 \\
    \end{array}
    \right] .
\end{equation}

It turns out that no matter what pair of banks is shocked, only one is liquidated. To see why, consider the case in which B$_1$ and B$_2$ are shocked.
In that case, the clearing vector solves (substituting from equation \eqref{e:F and F hat} into equation \eqref{eq:optimal_repayment_cases}):

\begin{equation}
\renewcommand\arraystretch{1.5}
   \left\{ \begin{array}{l}
        R_{1\rightrightarrows}    =
            \max\Big\{ \, 0 \, , \, \min\left\{ \,  - \frac 3 5  + \frac 4 5 R_{2 \rightrightarrows} +  \frac 2 3 R_{3 \rightrightarrows}
            \, , \,
             3
            \, \right\} \, \Big\},
            \\
        R_{2\rightrightarrows}
        = \max\Big\{ \, 0 \, , \, \min\left\{ \,   - \frac 3 5  + \frac 2 3 R_{1 \rightrightarrows} +  \frac 1 3 R_{3 \rightrightarrows}
            \, , \,
            \frac 5 2
            \, \right\}\,  \Big\},
            \\
        R_{3\rightrightarrows}
        =
        \max\Big\{ \, 0 \, , \, \min\left\{ \,   \ 1  + \frac 1 3 R_{1 \rightrightarrows} +  \frac 1 5 R_{2 \rightrightarrows}
            \, , \,
            \frac 32
            \, \right\} \, \Big\}.
    \end{array}
    \right.
\end{equation}
The first two lines of the system illustrate how the exponential network allocates liquidity efficiently. The repayments from B$_3$, the bank with excess liquidity, are allocated primarily to one of the two banks with the liquidity shortfall---B$_1$ gets $\frac 2 3$ of B$_3$'s total repayment, B$_2$ only $\frac 1 3$ of it.
And that allows B$_1$ to survive. B$_2$ is liquidated. But that is (constrained)  efficient as there are not enough resources to save them both anyway.

Solving the system gives the clearing vector
     $\textbf{R}_{\rightrightarrows}=
     \left(
        \frac{3}{35},
        0,
        \frac{36}{35} \right)$,
        which, having only one zero entry, affirms that only $\mathrm{B}_2$ is liquidated (equation \eqref{eq:optimal_repayment_cases}).

The other cases are analogous. When B$_1$ and B$_3$ are shocked the clearing vector is
$    \textbf{R}_{\rightrightarrows}= \left(
        \frac{3}{7},
        \frac{9}{7},
        0 \right),
$
implying only B$_3$ is liquidated, and when B$_2$ and B$_3$ are shocked it is      $\textbf{R}_{\rightrightarrows}=
    \left(    \frac{39}{35},
        \frac{1}{7},
        0
    \right)$, implying only B$_3$ is.

\emph{Ring network.} \label{pp:ring example}  {}Consider also the ring network, in which each B$_i$ owes $F$ to B$_{i+1}$ (mod 3) and nothing to the other bank. It is also constrained efficient in this example: When B$_1$ and B$_2$ are shocked, the clearing vector is $\mathbf{R}_{\rightrightarrows} = \left( \frac{2}{5}, 0, 1 \right)$ for $F \geq 1$---the not-shocked bank B$_3$ makes the maximum payment, all of it to B$_1$, which survives---and the other states are analogous. But efficiency here is an artifact of the three-bank example, in which the shocked banks are necessarily adjacent. In Appendix~\ref{a:ring}, we show that with four banks and a larger shock, the shocked banks can be separated by the not-shocked ones, in which case each not-shocked bank pays a different shocked bank and neither is saved.\footnote{{}Regularity and constrained efficiency generally push in opposite directions: One requires symmetry in banks' total interbank debts, the other asymmetry in which shocked banks are liquidated. The comparison of the complete network (regular) and exponential network (not) captures that. The ring with $N=3$ is a caveat, as it can be both regular and constrained efficient. But, as Appendix~\ref{a:ring} shows, this need not remain true once $N\geq4$.}

\section{Pairwise Stability}\label{s:stability}

We now explore what networks can arise through decentralized contracting. We start by defining a notion that captures that: Do banks have the incentive to deviate from their links?

\begin{definition}[Pairwise stability w.r.t.\ $G$] \label{d:pairwise}
Fix a distribution $G$ of shocks $\bm{\sigma}$. A network $\mathbf{F}$ is \emph{pairwise stable} if no pair B$_i$ and B$_j$ of banks with a direct link ($F_{i \to j} > 0$ or $ F_{j \to i}> 0$) can change their face values in a way that strictly increases both of their expected payoffs.
\end{definition}

\noindent {}In words, a network is pairwise stable if it does not admit bilateral deviations: No two linked banks can privately change their mutual debts and both come out strictly ahead.
\label{pp:dontnet}{}Netting out is always an admissible deviation, corresponding to reducing offsetting face values. As such, a network with gross positions being pairwise stable captures banks choosing not to net out.

The definition is weak in that deviations must be strict---ties are broken toward the status quo---and some types of conceivable deviations are not allowed---no new links can be formed.
(Deviations need not be symmetric, so, unlike in the baseline, we do not impose zero-net positions; Proposition~\ref{p:existence} still applies, as
its proof does not use that assumption.)
Nonetheless, our first result here is negative: The exponential network, despite being efficient, is not pairwise stable:

\begin{proposition}[Exponential not pairwise stable]\label{p:exponential_not_pairwise}
Suppose $\ell < 2 \theta y$, $1\leq L^*$, B$_{N-1}$ and B$_N$ are linked, and $G$ puts non-zero weight on exactly one of them being shocked (i.e.\  $\mathbb{P}[\sigma_{i}( 1 - \sigma_j )=1 ] > 0$ for $(i,j) \in \{(N-1,N), (N, N-1)\}$).

Let $\mathbf{F}$ be an exponential network with $s\leq s^*$ and $\alpha$ be sufficiently large.  $\alpha \mathbf{F}$ is not pairwise stable w.r.t.\ $G$ (where $L^*$ and $s^*$ are as in equations~(\ref{e:L*}) and~(\ref{e:s*})).
\end{proposition}

The proposition illustrates that constrained efficiency does not imply pairwise stability. The same priority ordering that maximizes liquidity insurance creates profitable bilateral deviations. In particular, B$_{N-1}$ and B$_N$, who, by construction, always receive liquidity last in the exponential network, can trap some liquidity between them by increasing the gross debt between them. That suggests that the networks that arise via decentralized contracting are likely to be inefficient.

Hence we turn to networks that resemble those observed in the real world, which are largely found to be core--periphery networks. We take two approaches to capturing the core--periphery networks, corresponding to different interpretations of the banks in our model, both defined in Definition~\ref{d:typology}.
The first is the star network, which is the simplest core--periphery network.
The second is the uniform complete network, in which all pairs of banks are connected by equal debts. This is not a core--periphery network in itself. It captures the core--periphery structure to the extent that the core is connected and the long-term assets $y$ and liquidity shocks $\ell$ in our model capture claims and liabilities on the periphery.

\begin{proposition}[Star pairwise stable] \label{p:star stable}
   Suppose $1<S<N$, $L^*<S$, $S\ell>(2N-S)\theta y$, and $ (N-S)   F \geq (\ell - \theta y)$. The star network is pairwise stable w.r.t.\ $G$ for all $G$.
\end{proposition}

\noindent This suggests that our model can explain observed interbank network structures. The result is strong in that it holds for all shock distributions $G$. The proof exploits the high-debt region in which the core bank always defaults.
 Hence we prove an auxiliary result in which it does not, suggesting the conclusion is robust:

\addtocounter{proposition}{-1}
\renewcommand{\theproposition}{\arabic{proposition}$'$}

\begin{proposition}[Star still pairwise stable]\label{p:star still stable}
   Suppose $\mathbf{F}$ is a star network of $N\geq3$ banks, $S = N -1$, and
    \begin{equation}
\max\left\{
\frac{\theta y}{N-1},
\frac{(2\theta-1)y}{N-2}
\right\}
<
F=\ell-\theta y
<
\frac{\theta y}{N-2}.
    \end{equation}
    The star network is pairwise stable w.r.t.\ the uniform distribution of shocks (every $\bm \sigma$ s.t.\  $\vert \bm \sigma \vert = N-1$ is equally likely).
\end{proposition}

\renewcommand{\theproposition}{\arabic{proposition}}

\noindent The proof of this version relies on banks being on the brink of liquidation: Increasing pairwise debts would harm them, because they would be liquidated when they are shocked. That sounds more knife-edge than we think it is. As our model is bang-bang---you are liquidated or not---that is the analog of a first-order condition---you take on offsetting debts until the risk of liquidation from the liability outweighs the insurance benefit of the claim.

The complete network is likewise pairwise stable:

\begin{proposition}[Uniform complete pairwise stable]\label{p:complete stable large shock}
Suppose \(N\geq3\) and \(\ell >N\theta y\). Then there exists \(\bar F\) such that for every \(F\geq\bar F\), the uniform complete network with common face value \(F\) is pairwise stable w.r.t. $G$ for all \(G\). \end{proposition}

\noindent The proof of this result, which requires all banks be highly indebted ($F \geq \bar F)$, exploits that shocked banks are liquidated and not-shocked banks default, offering a possible explanation for why banks are indebted and crises are systemic: It deters deviations in networks.

Pairwise stability says that, per Definition~\ref{d:pairwise},  no pair of banks can do anything to improve their payoffs. The outcomes are nonetheless inefficient for banks collectively. In both cases, the networks are ``too symmetric,'' in that they often allocate liquidity equally to banks instead of prioritizing them, violating the ``no wasting liquidity'' principle (see Section~\ref{s:exp}).
In the star network, only the core bank can be saved. Other shocked banks are always liquidated:
\begin{proposition}[Star inefficient]\label{p:star inefficient} Suppose $N\theta y<S\ell$. Let $\mathbf{F}$ be a star network.    All shocked banks are liquidated if either (i) the core bank is not shocked or (ii) it is shocked and $(N-S) F < (\ell-\theta y)$. It is not constrained efficient whenever $L^*< S$.
\end{proposition}

The uniform complete network can be even worse. There, all shocked banks are liquidated whenever shocks are large:
\begin{proposition}[Uniform complete inefficient] \label{p:complete inefficient}  Let $\mathbf{F}$ be the uniform complete network in Proposition~\ref{p:complete stable large shock}. Whenever
    \begin{equation} \label{e:complete sigma}
    S \ell > N \theta y ,
    \end{equation}
    all shocked banks are liquidated.
\end{proposition}

\noindent This says that whenever the constrained efficient number of liquidations $L^*$ is not zero, all shocked banks are liquidated. (That can  be seen from substituting the condition~\eqref{e:complete sigma} on the number of shocked banks into the definition of $L^*$ in equation~\eqref{e:L*}.) Even though, per Section~\ref{s:LT}, dilutable debt necessarily helps, the result still affirms the  idea that observed financial networks are ``robust yet fragile,'' do well for small shocks but struggle to sustain large ones.

\section{Heterogeneous Banks} \label{s:hetero}

So far, we assumed that all banks had the same assets in place and the same size liquidity shocks. Now we allow them to be heterogeneous, with $y_i$ denoting B$_i$'s assets and $\ell_i>\theta y_i$ the size of its potential liquidity shock.
We denote the deadweight loss of liquidating B$_i$ by $\Delta_i \geq 0$; in the baseline it is the non-pledgeable value destroyed, $\Delta_i = (1-\theta)y_i$.

We start with a generalization of the planner's problem to this environment:

\begin{definition}[Generalized planner's problem (PP)]\label{d:pp} The \emph{(generalized) planner's problem} is to find a vector of transfers $\textbf{\emph t}$ for each $\bm \sigma$, with $t_i$ paid to each B$_i$, to minimize the total deadweight loss of liquidation,
    \begin{equation} \label{e:t DWL}
    \mathrm{ minimize }  \ \  \sum_{i=1}^N \mathbbm{1}_{\{\theta y_i+t_i < \ell_i \sigma_i\}}
    \Delta_i ,
\end{equation}
subject to (i) each bank's liquidity constraint
\begin{equation}\label{eq:pp:liquidity_constraint}
    t_i\geq \min\{\ell_i\sigma_i-\theta y_i,0\}
\end{equation}
and (ii) total liquidity being conserved
\begin{equation}\label{eq:pp:liquidity_conservation}
    \sum_{i=1}^N t_i\leq 0.
\end{equation}
\end{definition}

\noindent Observe that if $\Delta_i$ is the same for all banks, then the objective is just to minimize the number of liquidations. If $\ell_i - \theta y_i$ is also the same, the constraints are too and the problem coincides with the planner's problem in Definition~\ref{d:constrained efficiency}.

We next show that the problem of choosing transfers among banks is equivalent to that of choosing the set of banks to liquidate, with  effectively the same objective and subject to the same constraint on aggregate liquidity:

\begin{definition}[Knapsack problem (KP)] \label{d:kp}
Find a vector of binary variables $\textbf{\emph x} \in \{0, 1\}^N$ for each $\bm \sigma$, with $x_i = 0$ if B$_i$ is liquidated, to minimize the deadweight loss,
    \begin{equation} \label{e:x DWL}
      \mathrm{ minimize } \ \  \sum_{i = 1}^N (1 - x_i) \Delta_i,
    \end{equation}
    subject to liquidity being conserved
    \begin{equation}\label{eq:aggregate_pledge}
    \sum_{i=1}^N x_i\sigma_i(\ell_i-\theta y_i)\leq \sum_{i=1}^N(1-\sigma_i)\theta y_i.
\end{equation}

\end{definition}

\noindent In computer science, this problem is called the ``knapsack problem,'' as it describes a problem of finding the most valuable set of objects to put in a knapsack subject to a constraint on their total weight, just as ours describes finding the most valuable set of banks to save subject to a constraint on the total liquidity required, an equivalence formalized in the next result:

\begin{proposition}[The planner's problem is the knapsack problem] \label{p:PP=KP}
PP in Definition \ref{d:pp} is equivalent to KP in Definition \ref{d:kp} in that
    (i)    if $\mathbf{\hat t}$ solves PP then $\mathbf{\hat x}$ with
    $
        \hat x_i := \mathbbm{1}_{\{\theta y_i+ \hat t_i \geq \ell_i\sigma_i\}} $ solves KP and (ii) if $\mathbf{\check x}$ solves KP
        then $\mathbf{\check t}$ with $\check t_i :=\sigma_i \check x_i(\ell_i-\theta y_i)-(1-\sigma_i)\theta y_i$ solves PP.
\end{proposition}

\noindent The knapsack problem is hard to solve (it is NP-hard, in the language of computer science). Hence algorithms have been developed that deliver approximate solutions quickly. One is the ``greedy algorithm,'' defined in our context as follows:

\begin{definition}[Greedy algorithm] \label{d:greedy}
Suppose banks are ranked by a permutation $\pi^{-1}$ on $\{1, ... , N\}$, so $\pi(r)$ is the index of the $r$-th highest ranked bank. For each B$_{\pi(r)}$, set $x_{\pi(r)} = 1$ if either $\sigma_{\pi(r)} = 0$ or if
\begin{equation} \label{e:greedy}
    \sum_{r'= 1 }^r \sigma_{\pi(r')}   \big(\ell_{\pi(r')} - \theta y_{\pi(r')} \big) \leq
    \sum_{r'= 1}^N \big(1 - \sigma_{\pi(r')} \big) \theta y_{\pi(r')}.
\end{equation}
\end{definition}

\noindent {}In words, for any ranking of banks, the greedy algorithm goes through them sequentially, saving the highest-ranked shocked banks until liquidity runs out (and saving all not-shocked banks as well).\footnote{{}The knapsack problem and the greedy algorithm have recently appeared elsewhere in the financial networks literature: \cite{Jackson-Pernoud-2024} study a model in which, if banks are connected in a star network, the optimal ex post bailout policy solves a knapsack problem, and they discuss a greedy algorithm as an approximate way to solve it. In our model, the planner's problem is always a knapsack problem---it is independent of the network---and the exponential network turns out to be an exact implementation of the greedy algorithm in equilibrium.\label{fn:AE-8}}

We now show that an exponential network can implement the greedy algorithm:

\begin{proposition}[Exponential networks implement the greedy algorithm]\label{p:exponential_implement_greedy}
    Let banks be ordered by the ranking $\pi$ in Definition~\ref{d:greedy}: $i = \pi(i)$ and suppose inequality (\ref{e:greedy}) is strict for all $\bm\sigma$ and $r$.
    There exists a threshold $s^*$ such that the exponential network $\alpha \mathbf{F}$ with base $s<s^*$ implements the outcome obtained by the greedy algorithm when $\alpha$ is sufficiently large.
\end{proposition}

\noindent This result embodies the idea that the exponential network, like the greedy algorithm, prioritizes what  banks should be saved. The result does not depend on the priority ranking given by $\pi$.

But the efficiency of the outcome does depend on the priority ranking. The next definition helps define a useful one:

\begin{definition}[Profitability index] The \emph{profitability index} of a bank B$_i$ is the ratio of the benefit to the cost of avoiding liquidation when it is shocked:
    \begin{equation}
        \mathrm{PI}_i := \frac{\Delta_i}{\ell_i - \theta y_i}.
    \end{equation}

\end{definition}

\noindent The profitability index here is akin to the eponymous ratio in capital budgeting, namely the ratio of the payoff to costs. The greedy algorithm corresponds to the rule prescribed for investment under capital constraints and mutually exclusive projects: Undertake those with the highest profitability indices.

When banks are ranked by their profitability indices, the greedy algorithm, and hence an exponential network, is optimal in some circumstances and nearly optimal in many others:

\begin{corollary}[Optimality of exponential network with common costs] \label{c:optimality of greedy} Suppose all shocked banks have the same liquidity shortfall: $\ell_i - \theta y_i = \ell_j - \theta y_j$ for all $i$ and $j$ and inequality~\ref{e:greedy} never binds for any $\bm\sigma$ and $r$.  Let banks be ordered by their profitability indices, $\text{PI}_{i} \geq \text{PI}_{j}$ for $i \leq j$.
An exponential network solves the planner's problem in Definition~\ref{d:pp}.

\end{corollary}

\noindent Given Proposition \ref{p:exponential_implement_greedy}, the result says that if the cost of saving every bank is the same, then saving those with the highest benefit delivers the optimum.

The algorithm need not be optimal, but it is nearly optimal if the liquidated banks are small:

\begin{corollary}[Approximate optimality of exponential network with small banks] \label{c:greedy small}
As in Lemma \ref{lm:exponentially_dominated}, let $i^*$ be the index of the first liquidated bank. An exponential network can deliver a deadweight loss within $\Delta_{i^*}$ of the solution to the planner's problem in Definition \ref{d:pp}.

\end{corollary}

\noindent {This result, which follows from an application of linear programming to the knapsack problem in \cite{dantzig1957discrete},  implies that the relative inefficiency of the greedy algorithm is small if banks are small, in particular if the deadweight loss of liquidating a single bank, namely B$_{i^*}$, is small.
The result also has an analogy in capital budgeting heuristics: If you use your entire budget, then doing those investments with highest profitability is optimal.
In general, the risk is that, due to indivisibility, doing those investments might leave some of your budget unused.
That risk is small if investments are small.}
Although the greedy algorithm could still be far from optimal when $\Delta_{i^*}$ is large, limiting results elsewhere in that literature suggest that the greedy algorithm is close to optimal on average (see \cite{calvin2003average}).

Overall, the results here add support to our finding that exponential networks implement a robust policy. The (approximate) efficiency of the priority rule is not specific to the case in which all banks have the same size of assets in place or liquidity shocks.

\section{Further Extensions}  \label{s:extensions}

Under our baseline assumptions (i) liquidation is always inefficient, (ii) default absent liquidation is costless, and (iii) banks' pledgeability $\theta$ is given. Here we relax these assumptions (albeit only for specific networks). Relaxing the first two, we show how to choose debt levels to balance the benefit of high debt in providing liquidity to avoid inefficient liquidation per the baseline with the cost of inducing excessive continuation/default included here. Relaxing the third, we find a further benefit of high debt: It leads banks to choose their pledgeability efficiently.

\subsection{Risky Assets and Too Few Liquidations} \label{s:risky assets}

So far, we assumed that $y$ was sufficiently large that liquidation was always inefficient. Now we assume that $y$ can have any value (but is the same for all banks). Thus it is efficient for all shocked banks to be liquidated if $y$ is low but not if it is high.
Here we denote the threshold below which liquidation is efficient by $y^*$ and we show how to choose debt levels in a complete network to implement the efficient liquidation policy.\footnote{\label{fn:efficiency}If paying $\ell$ represents a physical cost, then $(1-\theta )y^* = \ell$, i.e.\ the cost of liquidation equals the cost of continuation.
If it includes a transfer to unmodeled creditors (like \citetalias{AOT}'s ``outside obligation'' $v$), then $y^* < \ell$. Working with a general $y^*$ allows us to stay agnostic on the interpretation of $\ell$.}

We consider a complete network with debt levels $F$, focusing on the case in which $S\ell <N\theta y$ for all $y>y^*$ (so, in principle, no bank need be liquidated: $L^* = 0$ in Lemma \ref{l:constrained efficiency}).  As the network is symmetric,  each shocked/not-shocked bank makes and receives the same payments; we index all shocked banks' payments by $s$ and not-shocked banks' by $n$. From equation \eqref{eq:payment_equilibrium},  the payment to each type is a pro rata share of the repayment made by all other banks of that type and by all banks of the other type. Hence the equilibrium equations for any $y$ are:
\begin{equation}
    \begin{cases}
    R_{s\rightrightarrows}=
        \left[
        \min\left\{ \, F \, , \, \theta y - \ell +  \dfrac{1}{N-1}
        \Big( (S-1)
            R_{s\rightrightarrows}+ (N-S)R_{n\rightrightarrows}
            \Big)
        \right\}
        \right]^+ ,
        \\[2em]
    R_{n\rightrightarrows}    =
    \left[ \min \left\{
    \, F \, , \, \theta y  + \dfrac{1}{N-1} \Big( S R_{s\rightrightarrows}+ (N-S-1) R_{n\rightrightarrows}  \Big) \right\} \right]^+.
    \end{cases}
\end{equation}
Solving gives the equilibrium repayments:
        \begin{equation}
            \begin{cases}
                R_{s\rightrightarrows}&= \left[ F- \frac{N-1}{N-S} ( \ell - \theta y) \,\right]^+, \\[1em]
                R_{n\rightrightarrows}&= F     ,
            \end{cases}
        \end{equation}
    with shocked banks being liquidated whenever $F <\frac{N-1}{N-S} ( \ell - \theta y).$ Thus the  efficient outcome is implemented---banks are liquidated if $y < y^*$ but not if $y \geq y^*$---by setting
             $F = \frac{N-1}{N-S} ( \ell - \theta y^*) .$

    Intuitively, no matter the value of $y$, high debts allow banks to raise liquidity. In the baseline, that is only a good thing. Here, it can be bad.  But an appropriately chosen debt level implements the efficient outcome no matter the value of $y$.

    \subsection{Costly Default} \label{s:default costs}

    So far, we assumed that, while liquidation entailed a deadweight loss, default was just a transfer.
    Thus we found that higher debts always (weakly) increased efficiency (in the sense of Proposition \ref{p:netting}).
    An alternative notion of efficiency could minimize defaults as well as liquidations.
    Here we show that an exponential network can achieve this goal if the debt levels are not too high, albeit only in an example.

    Here we return to the three-bank exponential network in the example in Section \ref{s:exp} and replace the network $\mathbf{F}$ in equation \eqref{e:F example} with $\alpha \mathbf{F}$, so increasing $\alpha$ increases indebtedness.
    Recall that the parameters are such that with two shocked banks the constrained efficient number of liquidations is one.
    As shocked banks always default (see equation \eqref{eq:optimal_repayment_cases}), the constrained efficient number of defaults is two.
    In the example in Section \ref{s:exp}, we achieve the efficient number of liquidations, but not of defaults (all three banks default). Here we do both.

 We achieve both goals by reducing the debt.
    To see that, set $\alpha = 5/8$ and observe that the clearing vector equilibrium payment is $
    \textbf{R}_{\rightrightarrows}= (
        \frac{1}{40},
        0,
        \frac{15}{16})$:  B$_1$ is not liquidated and $\mathrm{B}_3$  does not default.
        (But reducing the debt too much undermines the flexibility to avoid liquidation. If debt is too low, say $\alpha = 1/2$, then the clearing vector is $\textbf{R}_{\rightrightarrows}=
        (0,
        0,
        \frac{3}{4})$: Both B$_1$ and B$_2$ are liquidated.)

    Overall, debt should be high enough to allow banks to raise liquidity, but, as in Section~\ref{s:risky assets}, it should not be too high, this time because that would lead to unnecessary costly defaults.

\subsection{Endogenous Pledgeability} \label{s:endogenous theta}

So far, we assumed that pledgeability $\theta$ was exogenous. Now we let each bank choose its own $\theta_i$ at a private cost $c(\theta_i)$, after the network is in place but before shocks are realized.\footnote{Per
footnote~\ref{fn:cash}, the choice can be read as how much cash to hold, with $c$ capturing the forgone return on the assets it replaces.} The concern is that coinsurance crowds out self-insurance: A bank might choose to be illiquid and free ride on its counterparty. Here we show that it need not, albeit only in a two-bank example.

We consider two banks with reciprocal debts $F_{1\to2}=F_{2\to1}=F\geq\ell$, each shocked independently with probability $p\in(0,1)$. We denote the cost of increasing pledgeability by $c$, with $c(0) = 0$, $c', c'' > 0$. Write $\theta^a:=\ell/y$ for the pledgeability a bank needs to meet a shock unaided; it is what a bank would choose on its own under the assumption that $c(\theta^a)\leq p(y-\ell)$, which we maintain
throughout.

The network changes what each bank needs. When only one bank is shocked, it can pledge its counterparty's assets as well as its own, so it survives whenever $\theta_1+\theta_2\geq\theta^a$: The pair need hold only what one bank would have held alone. When both are shocked, nothing flows between them, and each survives only on its own $\theta_i$. The candidates are therefore $(0,0)$, shared provision $(\theta^a/2,\theta^a/2)$, concentrated provision $(0,\theta^a)$, and full self insurance $(\theta^a,\theta^a)$; which is optimal turns on how convex $c$ is. But whichever it is, banks choose it:

\begin{proposition}[Endogenous pledgeability]\label{p:endogenous theta}
If $c(\theta^a)\leq p(y-\ell)$, then every weakly socially optimal pledgeability profile (i.e.\ every maximizer of expected surplus net of pledgeability costs) is a Nash equilibrium.
\end{proposition}

\noindent The reason is that a bank's pledgeability reaches its counterparty only in the states in which the bank is not shocked itself, so a bank that cuts back loses its own protection first. At the shared profile, cutting to zero leaves the pair short of $\theta^a$, so the deviator is liquidated when it alone is shocked; the inequality making shared provision socially better than $(0,0)$ is exactly the one making that deviation privately unprofitable. Because each bank both provides and receives insurance, the externality appears on both sides of the planner's comparison and cancels.

As in Section~\ref{s:LT}, this relies on debts being high: It is the claim on a counterparty that makes its liquidity available. (The one fragile profile is the concentrated one, in which a single bank bears the cost while both enjoy the protection. It survives only because our maintained condition makes that bank willing to hold $\theta^a$ in autarky anyway---concentration is privately
sustainable exactly when it is not really a sacrifice.)

\section{Conclusion}  \label{s:conclusion}

  We revisit systemic risk in financial networks, allowing existing interbank debts to be diluted with new debt, as they often can be in practice. In the model, a network of dilutable debts provides insurance among banks, and the right one implements the constrained-efficient transfers of liquidity. Even absent state-contingent contracts, dilution makes repayments state contingent: A shocked bank borrows against its interbank claims, and its counterparties bear the cost.

Dilutability can thus reverse the usual reading of what makes a system fragile. High gross debts and close ties spread distress when debts cannot be diluted, but liquidity when they can. Dilution can then substitute for policies that impose losses on creditors administratively (e.g., \cite*{Bernard-Capponi-Stiglitz-2022}): It implements a ``backdoor bail-in,'' in which debtors avoid insolvency by shifting the cost of distress onto diluted creditors.

That does not mean there is no case for regulation. True, given the right exposures, the bail-in requires no intervention. But those exposures need not emerge from bilateral contracting: Banks do not internalize the insurance they provide to the rest of the system, so they have an incentive to deviate from the socially efficient exposures. Policy aimed at the exposures banks build up before a crisis may thus substitute for bail-outs during one.

     Nor should regulation flatten exposures or restrict dilution: It is asymmetry and dilutability that do the insuring. Policies that compress gross positions, cap the largest ones, or curtail the superpriority of new claims may trade real insurance for the appearance of safety.

\newpage
\appendix

\section{Proofs}

\subsection{Proof of Lemma \ref{l:isomorphism} (\nameref{l:isomorphism})}

  In \citetalias{AOT}, the clearing vector is completely described by (Lemma B2, equation (B3)):
    \begin{equation}
        \textcolor{blue!70}{\bm{x}}=
        \Big[\min \big\{\,\textcolor{red!70}{y \bm{1}} \, , \, \textbf{Q} \textcolor{blue!70}{\bm{x}} +\bm{e}+\zeta A\bm{1}\, \big\}\Big]^+.
    \end{equation}
In the case in which long-term assets are fully destroyed in default, the $\zeta=0$ case, the equation can be re-written as
\begin{equation}
    \textcolor{blue!70}{\bm{x}}=
    \left[\min \left\{\, \textcolor{red!70}{y \bm{1}} \, , \, \textbf{Q} \textcolor{blue!70}{\bm{x}} +\textcolor{cyan!70}{\frac{a-v}{a-v+A}} \textcolor{violet!70}{(a-v+A)\bm{1}} -\textcolor{teal!70}{\epsilon\bm{\sigma}}
    \, \right\}
    \right]^+,
\end{equation}
having used that, by definition, $\bm{e}=\bm{a}-\bm{v} - \epsilon {\bm \sigma}$ (with no cash, AOT's $c_i \equiv 0$) and denoting the profile of shock indicator, a notation \citetalias{AOT} do not use, by $\bm \sigma$, as in our baseline.
This is equivalent to the equilibrium in our benchmark model, which per Definition \ref{d:equilibrium} is given by
\begin{equation}
    \textcolor{blue!70}{\textbf{R}_{\rightrightarrows}}=\Big[
    \min\big\{\, \textcolor{red!70}{\textbf{F}_{\rightrightarrows}} \, , \, \FT \textcolor{blue!70}{\textbf{R}_{\rightrightarrows}} +\textcolor{cyan!70}{\theta} \textcolor{violet!70}{y \bm{1}} -\textcolor{teal!70}{\ell \bm{\sigma}} \big\}\Big]^+,
\end{equation}
where the color coding represents the mapping between the notation in the two papers, as described in Table \ref{t:AOT correspondence}.

In both models, the banks default whenever they cannot repay the face value of their debts; hence the sets of defaulting banks coincide. Likewise, in both, all defaulting banks are liquidated; hence efficiency coincides too (see Definition \ref{d:efficiency}).
\qed

\begin{table}

    \centering

        \caption{Notations in  \citetalias{AOT} and here.\label{t:AOT correspondence}}

    \begin{tabular}{l|rr}
        &AOT & This paper  \\
        \hline
        Face value of debt    & $y_{ji}$     & $F_{i \to j}$
        \\
        Payment received &$[\textbf{Q}\bm{x}]_i$ & $[\FT
        \textbf{R}_{\rightrightarrows}]_i$
        \\
        Negative shock &$ a-v-e_i$ & $\ell \sigma_i$\\
        Total assets &$a-v+A$ &  $y$
        \\
        \quad {Pledgeable assets} &$a-v$ &  $\theta y$
        \\
        \quad {Non-pledgeable assets} &$A$ & $(1-\theta)y$\\
        \end{tabular}

\end{table}

\subsection{Proof of Lemma \ref{l:ST netting} (\nameref{l:ST netting})}

This proof generalizes \citetalias{AOT}'s proof of their Proposition 3. The idea is to show that in the equilibrium of $\alpha \mathbf{F}$ for $\alpha > 1$, each bank's shortfall $\mathbf{F}_\rightrightarrows - \mathbf{R}_\rightrightarrows$ is greater than it is in the equilibrium of $\mathbf{F}$ and therefore so is the number of defaults.

\begin{lemma}\label{lm:shortfall} Define the mapping
    \begin{equation}
    \Psi^\alpha : \mathbf{D} \mapsto  \Big[
        \min\{ \alpha \mathbf{F}_{\rightrightarrows} \, , \,  \FT \mathbf{D}-\theta y \bm{1}+\ell\bm{\sigma}\} \, \Big]^+.
    \end{equation}
If $\mathbf{R}^\alpha_{\rightrightarrows}$ is a clearing vector of $
\alpha \mathbf{F}$, then the ``shortfall'' $\mathbf{D}^\alpha := \alpha \mathbf{F}_\rightrightarrows - \mathbf{R}^\alpha_{\rightrightarrows}$ is a fixed point of $\Psi^\alpha.$
\end{lemma}

\begin{proof}
We compute, using $\mathbf{Q} \equiv \FT$:
\begin{align}
    \alpha \textbf{F}_{\rightrightarrows}-\textbf{R}^\alpha_{\rightrightarrows}
    &=
    \alpha \textbf{F}_{\rightrightarrows}-\max \big \{\bm{0},\min\{\alpha \textbf{F}_{\rightrightarrows},\textbf{Q}\textbf{R}^\alpha_{\rightrightarrows}+\theta y \bm{1}-\ell\bm{\sigma}\} \big\}\tag{Payment Eqm.}\\
    &=\min\big\{\alpha \textbf{F}_{\rightrightarrows}\, , \, \max\{\mathbf{0}, \alpha \textbf{F}_{\rightrightarrows}-\textbf{Q}\textbf{R}^\alpha_{\rightrightarrows}-\theta y \bm{1}+\ell\bm{\sigma}\}\} \tag{Combining}\\
    &=\Big[\min\big\{\alpha \textbf{F}_{\rightrightarrows} \, , \, \alpha \textbf{F}_{\rightrightarrows}-\textbf{Q}\textbf{R}^\alpha_{\rightrightarrows}-\theta y \bm{1}+\ell\bm{\sigma} \big\}\Big]^+\tag{Interchange min/max}\\
    &=\Big[ \, \min \big\{\alpha \textbf{F}_{\rightrightarrows} \, , \,  \textbf{Q}(\alpha \textbf{F}_{\rightrightarrows}-\textbf{R}^\alpha_{\rightrightarrows})-\theta y \bm{1}+\ell\bm{\sigma} \big\}\Big]^+\tag{zero-net debt}\end{align}
Substituting from the definitions of $\mathbf{D}^{\alpha}$ and $\Psi^{\alpha}$ gives the result.
\end{proof}

Now we show that, for $\alpha>1$, $\Psi^\alpha$ has a fixed point that is no smaller than $\mathbf{D}^1$.

\begin{lemma}Let $\mathbf{D}^1$ be a fixed point of $\Psi^1$ and define
    \begin{equation}
    \mathscr{H}^\alpha := \prod_{i=1}^N \big[ D_{i}^1,\alpha F_{i\rightrightarrows}]
    \end{equation}
For $\alpha > 1$, $\Psi^\alpha$ maps  $\mathscr{H}^\alpha$ into itself, i.e.\ $\Psi^\alpha \big( \mathscr{H}^\alpha \big)
        \subset
        \mathscr{H}^\alpha.$

\end{lemma}

\begin{proof}

The upper bound, i.e.\ that $\Psi^\alpha(\mathbf{D}^\alpha) \leq \alpha \mathbf{F}_{\rightrightarrows}$, follows immediately from the definition of $\Psi^\alpha$ as a minimum.

So we need only to  show the lower bound, i.e.\ that $\Psi^\alpha(\mathbf{D}^\alpha) \geq \mathbf{D}^1$ for all $\mathbf{D}^\alpha \in \mathscr{H}^\alpha$. We have that for $\mathbf{D}^\alpha \in \mathscr{H}^\alpha$, $\mathbf{D}^1 \leq \mathbf{D}^\alpha$ by definition of the domain. Thus we can compute:
    \begin{align}
 \Psi^\alpha (\mathbf{D}^\alpha)
        &=\Big[ \min\big\{\alpha \textbf{F}_{\rightrightarrows} \, , \, \textbf{Q}\textbf{D}^\alpha-\theta y \bm{1}+\ell\bm{\sigma}\big\} \Big]^+\\
                        &\geq
                        \Big[ \min\big\{\alpha \textbf{F}_{\rightrightarrows} \, , \, \textbf{Q}\textbf{D}^1-\theta y \bm{1}+\ell\bm{\sigma}\big\} \Big]^+\\
                        &\geq \Big[ \min\{ \textbf{F}_{\rightrightarrows},\textbf{Q}\textbf{D}^1-\theta y \bm{1}+\ell\bm{\sigma}\} \Big]^+\\
                        &\equiv \Psi(\mathbf{D}^1) \equiv\textbf{D}^1,
    \end{align}
    since $\mathbf{D}^1$ is a fixed point of $\Psi^1$ by definition.
\end{proof}

Combining the two lemmata above and applying Brouwer's theorem yields a fixed point $\mathbf{D}^\alpha\in\mathscr{H}^\alpha$. Reversing the algebra in the proof of Lemma~\ref{lm:shortfall} shows that this fixed point corresponds to a clearing vector of $\alpha\mathbf{F}$. Generically, it is therefore the equilibrium shortfall vector of $\alpha\mathbf{F}$. Hence $\mathbf{D}^\alpha\geq\mathbf{D}^1$, so every bank liquidated under $\mathbf{F}$ is also liquidated under $\alpha\mathbf{F}$. \qed

\subsection{Proof of Lemma \ref{l:ST delta} (\nameref{l:ST delta})}

This proof mirrors \citetalias{AOT}'s proof of their Proposition 6; we translate it to our notation, adapt it to our notion of efficiency, and add some details.

For statement (i), w.l.o.g., we write $F := F_{i \rightrightarrows}$ for the ring.

\textbf{Proof of statement (i).} Take $\ell>N\theta y$, which is sufficient for $\ell$ to be sufficiently large. We show, by verification, that in equilibrium all banks are liquidated and, therefore, no network is less efficient than the ring. Assuming all banks are liquidated and letting the single shocked bank be B$_1$, w.l.o.g., the equilibrium is the solution to
        \begin{equation}
    \left\{
        \begin{array}{rll}
        R_{1 \to 2} &=  \big[ \theta y - \ell  + R_{N \to 1} \big]^+
        \\
        R_{i \to i + 1} &= \theta y + R_{i- 1 \to i}
        \quad\quad\quad  i \in \{2, ... , N \} ,
        \end{array}
        \right.
    \end{equation}
    with the convention that $N+1 := 1$ for indices.
Solving, we have the equilibrium:
$R_{1\rightrightarrows}=0$ and, for $i > 1$, $R_{i\to i+1}=(i-1)\theta y\leq(N-1)\theta y,$ which is less than $F$ by hypothesis.
To see that there is no other equilibrium, observe that every equilibrium satisfies $R_{i\to i+1}\leq\theta y+R_{i-1\to i}$ for $i\geq2$, and hence $R_{N\to1}\leq R_{1\to2}+(N-1)\theta y$. The shocked bank's repayment equation therefore gives
\begin{equation}
R_{1\to2}\leq\big[R_{1\to2}+N\theta y-\ell\big]^+.
\end{equation}
Since $\ell>N\theta y$, this requires $R_{1\to2}=0$. The remaining repayment equations then uniquely give $R_{i\to i+1}=(i-1)\theta y$ for $i>1$. I.e.,\ all banks are liquidated.

\textbf{Proof of statement (ii).} Here we show that if  $\mathbf{F}$ is $\delta$-connected and $\delta$ is small, then not all banks are liquidated and, therefore, the network is more efficient than the ring (in which they are).
Let $\mathscr B$ be one side of a cut witnessing that $\mathbf F$ is $\delta$-connected, chosen so that the single shocked bank is not in $\mathscr B$; if it lies in the original set, take the complement. To do so, we show two lemmata.

\begin{lemma} Let $\mathscr B$ be so chosen. If $N\delta\max_k F_{k\rightrightarrows}<\theta y,$ then
        \begin{equation}\label{eq:total_payment_B_to_Bc}
        \sum_{i\in\mathscr{B}}\sum_{j\in\mathscr{B}^c}R_{i\to j}<\theta y|\mathscr{B}|.
    \end{equation}
\end{lemma}
The result says that the total payment from banks in $\mathscr{B}$ to those in $\mathscr{B}^c$ is small when $\delta$ small.

\begin{proof} By the definition of delta-connectedness,
    $
            F_{i\to j}=\hat F_{i\to j}F_{i\rightrightarrows}\leq \delta F_{i\rightrightarrows},$ for all  $ (i,j)\in \mathscr{B}\times \mathscr{B}^c
$. Thus
$
            R_{i\to j}\leq F_{i\to j}\leq \delta F_{i\rightrightarrows},$ for all  $ i \in \mathscr{B}$ and  $ j\in \mathscr{B}^c.
$
Now summing $i$ over $\mathscr{B}$, summing $j$ over $\mathscr{B}^c$,     and using $\delta|\mathscr{B}^c|\max_k F_{k\rightrightarrows}<N\delta\max_k F_{k\rightrightarrows}<\theta y$, gives
        \begin{equation}
            \sum_{i\in\mathscr{B}}\sum_{j\in\mathscr{B}^c}R_{i\to j}\leq \delta|\mathscr{B}^c|\sum_{i\in\mathscr B}F_{i\rightrightarrows}\leq \delta\max_k F_{k\rightrightarrows}|\mathscr{B}||\mathscr{B}^c|<\theta y|\mathscr{B}|.
        \end{equation}

\end{proof}

\begin{lemma} Maintain the assumptions that the single shocked bank is not in $\mathscr{B}$ and that $\mathbf{F}$ is $\delta$-connected with $N\delta\max_k F_{k\rightrightarrows}<\theta y$. If inequality \eqref{eq:total_payment_B_to_Bc} holds, then
not all banks in $\mathscr{B}$ are liquidated.
\end{lemma}
\begin{proof}
Suppose, in anticipation of a contradiction, that all banks in $\mathscr{B}$ are liquidated. Thus for each  {B}$_j$  in $\mathscr{B}$,    $R_{j\rightrightarrows}=R_{j\leftleftarrows}+\theta y$. Summing over banks in $\mathscr{B}$, we have that   $\sum_{j\in \mathscr{B}} R_{j\rightrightarrows}= \sum_{j\in \mathscr{B}} R_{j\leftleftarrows}+\theta y|\mathscr{B}|$.
Now just expand $R_{j\rightrightarrows}$ and $R_{j\leftleftarrows}$ into their component payments,
    \begin{equation}
        \sum_{j\in \mathscr{B}} \left(\sum _{i\in \mathscr{B}}R_{j\to i}+\sum _{i\in \mathscr{B}^c}R_{j\to i}\right)= \sum_{j\in \mathscr{B}} \left(\sum _{i\in \mathscr{B}}R_{ i\to j}+\sum _{i\in \mathscr{B}^c}R_{ i\to j}\right)+\theta y|\mathscr{B}|,
    \end{equation}
    and cancel $\sum_{j\in \mathscr{B}} \sum _{i\in \mathscr{B}}R_{j\to i}$  to get that
    \begin{equation}
         \sum_{j\in \mathscr{B}} \sum _{i\in \mathscr{B}^c}R_{j\to i}= \sum_{j\in \mathscr{B}} \sum _{i\in \mathscr{B}^c}R_{ i\to j}+\theta y|\mathscr{B}|\geq \theta y|\mathscr{B}|
    \end{equation}
    This contradicts equation \eqref{eq:total_payment_B_to_Bc}, which says the total payment from $\mathscr{B}$ to $\mathscr{B}^c$ is small.
\end{proof}

\subsection{Proof of Lemma \ref{l:ST radius} (\nameref{l:ST radius})}

The result is the same as \citetalias{AOT}'s. Hence we omit the proof. \qed

\subsection{Proof of Lemma \ref{l:ST bottleneck} (\nameref{l:ST bottleneck})}

Statement (i) follows by combining \citetalias{AOT}'s Proposition 8 with their bound relating the bottleneck parameter to harmonic distance. For statement (ii), that bound is not directly applicable because its lower bound concerns the maximum harmonic distance over all ordered pairs, whereas the identity of the shocked bank is fixed. We therefore use a direct cut-payment argument, analogous to \citetalias{AOT}'s proof of their Proposition 6(ii).

We begin with a lemma that connects the bottleneck parameter $\beta$ to the harmonic distance $d$:

        \begin{lemma} \label{l:beta and d}
            Let $\mathbf{F}$ be a symmetric regular financial network of $N$ banks. The bottleneck parameter $\beta$ satisfies
        \begin{equation}
            \frac{1}{2N\beta }\leq \ \max_{i, k \,  :\, i \neq k} d_{k\to i}\leq \frac{16}{N\beta^2} .
        \end{equation}
    \end{lemma}

    \begin{proof} This is \citetalias{AOT}'s Lemma 1 (p.\ 580). Hence we omit the proof.
    \end{proof}

    With this we proceed to the two statements of the lemma.

\textbf{Proof of statement (i).}  For $\beta > \beta^{ST} \equiv 4 \sqrt{ \frac{\theta y }{N F} } $, we have, from Lemma \ref{l:beta and d}, that
    \begin{equation}
        d_{j \to i} \leq \frac{16}{N \beta^2} < \frac{ 16}{N (\beta^{ST})^2}
            = \frac{F}{\theta y} \equiv d^{ST} \  \text{ for all } \  i\neq j,
    \end{equation}
    per the definition of $d^{ST}$ in Lemma \ref{l:ST radius}. That lemma implies that when   B$_j$ is shocked, then each B$_i$ defaults.

\textbf{Proof of statement (ii).} Let $\mathscr B$ be a nonempty proper subset attaining the minimum in Definition~\ref{d:bottleneck}. Since the network is symmetric, replacing $\mathscr B$ with $\mathscr B^c$ leaves the cut value unchanged, so choose the minimizing side such that the shocked bank B$_j$ is not in $\mathscr B$.

Suppose, in anticipation of a contradiction, that all banks are liquidated. Every bank in $\mathscr B$ is not shocked and, since liquidation and default coincide in the benchmark, equation~\eqref{e:total repayment} implies
\begin{equation}
    R_{i\rightrightarrows}-R_{i\leftleftarrows}=\theta y
    \quad\text{for every }i\in\mathscr B.
\end{equation}
Summing over $\mathscr B$, cancelling payments internal to $\mathscr B$, using $R_{i\to k}\leq F_{i\to k}$, and then using regularity and the definition of $\beta$ gives
\begin{align}
    \theta y|\mathscr B|
    &=\sum_{i\in\mathscr B}\sum_{k\in\mathscr B^c}
    \big(R_{i\to k}-R_{k\to i}\big)\\
    &\leq\sum_{i\in\mathscr B}\sum_{k\in\mathscr B^c}F_{i\to k}\\
    &=F\beta|\mathscr B||\mathscr B^c|\\
    &\leq FN\beta|\mathscr B|.
\end{align}
But $\beta<\beta_{ST}\leq\theta y/(2NF)$ makes the last expression strictly less than $\theta y|\mathscr B|/2$, a contradiction. Thus at least one bank in $\mathscr B$ is not liquidated. \qed

\subsection{Proof of Proposition \ref{p:existence} (\nameref{p:existence})\label{pr:existence}}

\emph{Existence.} A payment equilibrium is a fixed point of the mapping
\begin{equation}\label{e:Phi map}
    \Phi(\mathbf{R})
    :=
    \Big[
    \min\Big\{
    \mathbf{F}_{\rightrightarrows}\,,\;
    \theta y\bm{1}-\ell\bm\sigma+\FT\mathbf{R}
    \Big\}
    \Big]^+
\end{equation}
(equation \eqref{eq:payment_equilibrium}). $\Phi$ maps the complete lattice
$\prod_{i=1}^N[0,F_{i\rightrightarrows}]$ into itself and is monotone,
since $\FT$ is entrywise non-negative and $x\mapsto[\min\{F,x\}]^+$ is
nondecreasing. By Tarski's fixed point theorem, a payment equilibrium
exists; moreover, there is a least payment equilibrium
$\mathbf{R}_{\rightrightarrows}$.

\emph{Uniqueness.} Let $\mathbf{R}'_{\rightrightarrows}$ be any payment
equilibrium, so
$\mathbf{R}'_{\rightrightarrows}\geq\mathbf{R}_{\rightrightarrows}$; write
$\Delta R_i:=R'_{i\rightrightarrows}-R_{i\rightrightarrows}\geq0$ and let
$\mathscr{C}:=\{i:\Delta R_i>0\}$. (Banks with $F_{i\rightrightarrows}=0$ repay
zero in every equilibrium, so every bank in $\mathscr{C}$ has
$F_{i\rightrightarrows}>0$.) Since, per equation \eqref{e:total repayment}, each
bank's repayment increases at most one for one with the repayments it
receives,
\begin{equation}\label{e:one for one}
    \Delta R_i
    \;\leq\;
    R'_{i\leftleftarrows}-R_{i\leftleftarrows}
    \;=\;
    \sum_{k\neq i}\hat{F}_{k\to i}\,\Delta R_k
    \qquad
    \text{for every } i .
\end{equation}
Suppose $\mathscr{C}\neq\emptyset$. Summing inequality~(\ref{e:one for one}) over $i\in\mathscr{C}$, using $\Delta R_k=0$ for $k\notin\mathscr{C}$ and $\sum_{i\neq k}\hat{F}_{k\to i}=1$, gives
\begin{equation}\label{e:sandwich}
    \sum_{i\in\mathscr{C}}\Delta R_i
    \;\leq\;
    \sum_{k\in\mathscr{C}}\Delta R_k
    \sum_{i\in\mathscr{C}\setminus\{k\}}\hat{F}_{k\to i}
    \;\leq\;
    \sum_{k\in\mathscr{C}}\Delta R_k .
\end{equation}
The two ends coincide, so both inequalities bind. The second binds only if
$\sum_{i\in\mathscr{C}}\hat{F}_{k\to i}=1$ for every $k\in\mathscr{C}$,
i.e.\ only if banks in $\mathscr{C}$ have no liabilities outside $\mathscr{C}$. The first binds only if inequality~(\ref{e:one for one}) binds for every $i\in\mathscr{C}$, which, given the piecewise form of equation  \eqref{eq:optimal_repayment_cases}---slope zero below zero and above $F_{i\rightrightarrows}$, slope one in between--- requires every $i\in\mathscr{C}$ to be in the middle region in both equilibria:
\begin{equation}\label{e:interior}
    R_{i\rightrightarrows}
    =\theta y-\ell\sigma_i+R_{i\leftleftarrows}
    \geq0
    \qquad\text{and}\qquad
    R'_{i\rightrightarrows}
    =\theta y-\ell\sigma_i+R'_{i\leftleftarrows}
    \leq F_{i\rightrightarrows} .
\end{equation}
Summing the first equation in condition~(\ref{e:interior}) over
$i\in\mathscr{C}$ and cancelling the payments within $\mathscr{C}$ (which
appear on both sides, as banks in $\mathscr{C}$ pay only within
$\mathscr{C}$) leaves
\begin{equation}\label{e:zero slack}
    S_{\mathscr{C}}\,\ell
    \;=\;
    \lvert\mathscr{C}\rvert\,\theta y
    \;+\;
    \sum_{k\notin\mathscr{C}}\;\sum_{i\in\mathscr{C}}
    \hat{F}_{k\to i}\,R_{k\rightrightarrows} ,
\end{equation}
where $S_{\mathscr{C}}$ is the number of shocked banks in $\mathscr{C}$:
The shocks in $\mathscr{C}$ exactly exhaust its pledgeable assets and inflows.
For any fixed repayment regime outside $\mathscr C$, the last term on its right-hand side is pinned down by banks outside $\mathscr C$:
Because banks in $\mathscr C$ have no liabilities outside $\mathscr C$, their resources and repayments do not enter the clearing problem on $\mathscr C^c$. Within each such regime, repayments outside $\mathscr C$ are piecewise affine and weakly increasing in $\theta y$. Hence the right-hand side of equation~(\ref{e:zero slack}) strictly increases in $\theta y$, because it contains the term $\lvert\mathscr C\rvert\theta y$, whereas the left-hand side is fixed for fixed $\ell$.
There are only finitely many choices of $\mathscr C$ and repayment regimes, so it holds only on a lower-dimensional subset of parameters, and the payment equilibrium is generically unique.

\emph{Payoff equivalence at the knife edge.} Banks outside $\mathscr{C}$ are unaffected by the multiplicity. Since no repayment differences leave $\mathscr{C}$, they receive identical payments from banks in $\mathscr{C}$; and since repayments among banks outside $\mathscr{C}$ are identical by construction, their total inflows are the same in both equilibria. Hence their liquidation decisions and payoffs coincide. For every $i\in\mathscr C$, condition~\eqref{e:interior} and $\Delta R_i>0$ imply $0\leq R_{i\rightrightarrows}<R'_{i\rightrightarrows}\leq
F_{i\rightrightarrows}$. Hence B$_i$ is not liquidated in either
equilibrium and obtains $(1-\theta)y$ in both. Thus all payment
equilibria induce the same liquidation decisions and payoffs.
\hfill $\qed$

\subsection{Proof of Proposition \ref{p:netting} (\nameref{p:netting})}

The proof is similar to that of Lemma \ref{l:ST netting}, but simpler because we work with repayments $\mathbf{R}_{\rightrightarrows}$ directly instead of ``shortfalls'' $\mathbf{F}_{\rightrightarrows} - \mathbf{R}_{\rightrightarrows}.$

Define the mapping
    \begin{equation}
        \Phi^\alpha ( \mathbf{R} ) :=   \Big[
                                    \min    \big\{ \,  \alpha \mathbf{F}_{\rightrightarrows} \, , \,
                                    \mathbf{Q} \mathbf{R} + \theta y \bm{1}-\ell \bm{\sigma}
                                            \big\}
                                \Big]^+.
    \end{equation}
     Keeping in mind that $\mathbf{Q} \equiv \FT$, the fixed point of $\Phi^1$ is a clearing vector of $\mathbf{F}$.  \qed

     \begin{lemma} \label{l:Phi} Let  $\mathbf{R}^1_{\rightrightarrows}$ be a fixed point of $\Phi^1$ and define
    \begin{equation}
        \mathscr{I}^\alpha  := \prod_{i = 1}^N \big[R^1_{i \rightrightarrows} , \alpha F_{i \rightrightarrows} \big].
    \end{equation}
    For $\alpha > 1$,  $\Phi^\alpha$ maps     $\mathscr{I}^\alpha$ into itself, i.e.\ that $\Phi^\alpha( \mathscr{I}^\alpha) \subset \mathscr{I}^\alpha$.
    \end{lemma}

    \begin{proof} The upper bound, i.e.\ that $\Phi^\alpha( \mathbf{R}^\alpha_{\rightrightarrows}) \leq \alpha \mathbf{F}_{\rightrightarrows}$, follows immediately from the definition of $\Phi^\alpha$ as a minimum.

    So we need only to show the lower bound, i.e.\ that $\Phi^\alpha(\mathbf{R}^\alpha_{\rightrightarrows}) \geq \mathbf{R}^1_{\rightrightarrows}$.
    For $\mathbf{R}^\alpha_{\rightrightarrows} \in \mathscr{I}^\alpha$,  we can compute:
    \begin{align}
    \Phi^\alpha (\mathbf{R}^\alpha_{\rightrightarrows})
        &=\Big[ \min\big\{\alpha \textbf{F}_{\rightrightarrows} \, , \, \textbf{Q}\textbf{R}^\alpha_{\rightrightarrows} +\theta y \bm{1}-\ell\bm{\sigma}\big\} \Big]^+\\
        &\geq \Big[ \min\big\{\alpha \textbf{F}_{\rightrightarrows} \, , \, \textbf{Q}\textbf{R}^1_{\rightrightarrows} +\theta y \bm{1}-\ell\bm{\sigma}\big\} \Big]^+
        \\
        &\geq \Big[ \min\big\{ \textbf{F}_{\rightrightarrows} \, , \, \textbf{Q}\textbf{R}^1_{\rightrightarrows} +\theta y \bm{1}-\ell\bm{\sigma}\big\} \Big]^+
        \\
        &\equiv \Phi^1(\mathbf{R}^1_{\rightrightarrows} ) \equiv \mathbf{R}^1_{\rightrightarrows},
    \end{align}
    since $\mathbf{R}^1_{\rightrightarrows}$ is a fixed point of $\Phi$ by definition.

    \end{proof}
        Given the lemma, we can apply Brouwer's theorem, to conclude that for $\alpha > 1$, an equilibrium of $\alpha \mathbf{F}$ is a fixed point of a mapping on $\mathscr{I}^\alpha$. Therefore the generically unique\footnote{Generic global uniqueness follows from Proposition \ref{p:existence}. When the equilibrium is not unique, let $\mathbf{R}_{\rightrightarrows}^1$ be the least fixed-point as per Proposition~\ref{p:existence}. The outcome-equivalence part of that proposition implies that all equilibria of each LT network induce the same liquidation set. Therefore comparison of these two fixed points suffices for the efficiency ranking.}
 clearing vector, being in $\mathscr{I}^\alpha$, exceeds $\mathbf{R}^1_{\rightrightarrows}$: All repayments are higher when $\alpha$ is higher, so there are fewer liquidations. \qed

\subsection{Proof of Proposition \ref{p:delta} (\nameref{p:delta})}

Note that this proof makes use of the minimum number of liquidated banks, $L^*$, derived in Lemma \ref{l:constrained efficiency}, even though that result comes later in the text.

Throughout we assume, w.l.o.g., that the B$_1$ is the not-shocked bank.

\textbf{Proof of statement (i).} Equation \eqref{eq:optimal_repayment_cases} and the fact that banks have zero net positions (so no shocked bank can repay in full) imply that any shocked bank ($i \geq 2$) repays
\begin{align}
    R_{i \rightrightarrows} \equiv R_{i \to i+1}
    &= \big[\theta y - \ell + R_{i \leftleftarrows} \big]^+
    \\
    &= \big[\theta y - \ell + R_{i-1 \to i} \big]^+,
\end{align}
having used the definition of the ring network (Definition \ref{d:typology}).

The expression for $R_{i\rightrightarrows}$ implies that if B$_{i-1}$ is liquidated, then B$_i$ is too. Thus the number of banks that are not liquidated is the maximum index $i$ for which $ R_{i - 1 \to i} \geq \ell - \theta y$.

We can now expand the condition recursively for any B$_i$ that is not liquidated:
    \begin{align}
    \ell - \theta y \leq R_{i - 1 \to i}
    &=   \theta y - \ell + R_{i - 2 \to i - 1 }
    \\
    &= k (\theta y - \ell) + R_{i - (k+1) \to i - k} \ \ \ \text{(for $k \in \{1,..., i - 2\})$}
    \\
    &= (i-2)(\theta y - \ell) + R_{1 \to 2}.
    \end{align}
    So the number of banks that are not liquidated is
    \begin{equation}
    \max \left\{ \,  i \leq N  \, : \,  i - 1 \leq \frac{R_{1 \to 2}}{\ell - \theta y} \, \right\}
    = \min
        \left\{
        \, N \, , \,
        \left\lfloor 1 + \frac{R_{1\to 2}}{\ell - \theta y} \right\rfloor \, \right\}.
    \end{equation}
    The number of liquidated banks is $N$ minus the above:
    \begin{equation}
    L = \max \left\{ \, 0 \, , \, \left\lceil N - 1 - \frac{R_{1 \to 2}}{\ell - \theta y} \right\rceil \, \right\}
    .
    \end{equation}
    Suppose first that $(N-1)(\ell-\theta y)>\theta y$. At least one shocked bank must be liquidated: Otherwise, iterating their repayment equations around the ring and using $R_{1\to2}\leq\theta y+R_{N\to1}$ would imply $(N-1)(\ell-\theta y)\leq\theta y$, a contradiction. Once one shocked bank is liquidated, every shocked bank after it is too, so B$_N$ is liquidated and $R_{N\to1}=0$. Hence $R_{1\to2}=\theta y$, since $F>\theta y$. Substituting into the preceding expression for $L$ and comparing with equation~(\ref{e:L*}) at $S=N-1$ gives
    \begin{equation}
        L=\left\lceil N-1-\frac{\theta y}{\ell-\theta y}\right\rceil=L^*.
    \end{equation}
    Suppose instead that $(N-1)(\ell-\theta y)\leq\theta y$. Then the repayment vector
    \begin{equation}
        R_{1\to2}=F
        \qquad\text{and}\qquad
        R_{i\to i+1}=F-(i-1)(\ell-\theta y)
        \quad\text{for }i\in\{2,\ldots,N\}
    \end{equation}
    is an equilibrium in which no bank is liquidated: The repayment equations hold for every shocked bank and $\theta y+R_{N\to1}\geq F$ makes B$_1$ repay in full. The outcome-equivalence part of Proposition~\ref{p:existence} implies the same liquidation outcome in every payment equilibrium, so $L=0=L^*$. In either case, the ring attains the minimum number of liquidations and is therefore most efficient among all networks.

    \textbf{Statement (ii).} A trivial example suffices: The ``linkless'' network---$F_{i \to j} = 0$ for all $i$ and $j$---is $\delta$-connected for any $\delta$; in it, all $N-1$ shocked banks are liquidated. Since $\ell\leq2\theta y$ implies $L^*\leq N-2$, that is less efficient than the ring network with $F > \theta y$, per statement (i). \qed

\subsection{Proof of Proposition \ref{p:radius} (\nameref{p:radius})}

We prove the two statements in turn. The arguments build on \citetalias{AOT}'s proof of their Proposition 8 (pp.\ 602--603). Ours are a bit more complicated because we cannot work with the clearing vector $\mathbf{R}_{\rightrightarrows}$, but have to work with the shortfall $\mathbf{D} \equiv \mathbf{F}_{\rightrightarrows} - \mathbf{R}_{\rightrightarrows}$ instead.

\textbf{Statement (i).} The proof comprises (somewhat involved) calculations using the shortfall $\mathbf{D}$, which ultimately allow us to bound the harmonic distance of the not-shocked banks to any liquidated bank.

Throughout we denote the set of defaulting banks by $\mathscr{D}$, that of liquidated banks by $\mathscr{L}$, and, hence, that of those that default but are not liquidated by $\mathscr{D}\setminus \mathscr{L}$. $\mathbf{1}$ denotes the  vector of all ones and $\mathbf{I}$ the identity matrix, each of appropriate dimension determined by the context. Given the network is regular, we can write ${F}_{i \rightrightarrows} \equiv F.$

We start with two lemmata. Each takes as its starting point the equilibrium equation for the shortfall
\begin{equation} \label{e:shortfall}
\mathbf{D} =  \Big[
        \min\{  \mathbf{F}_{\rightrightarrows} \, , \,  \mathbf{Q}\mathbf{D}-\theta y \bm{1}+\ell\bm{\sigma}\} \, \Big]^+,
\end{equation}
which follows from setting $\alpha=1$ in Lemma~\ref{lm:shortfall}.
The first lemma develops the equation for banks that default but are not liquidated; the second for banks that are liquidated.

\begin{lemma} \label{l:radius 1}
 $\mathbf{D}_{\mathscr{D}\backslash \mathscr{L}}
            =\big(\mathbf{I} -\mathbf{Q}_{\mathscr{D}\backslash \mathscr{L},\mathscr{D}\backslash \mathscr{L}}\big)^{-1}  \big(\mathbf{Q}_{ \mathscr{D}\backslash \mathscr{L},\mathscr{L}}\bm{1} F - ( \theta y - \ell)  \bm{1}\big)
            $.
\end{lemma}

\begin{proof}
        For banks that default, the shortfall is not zero, and for those that are not liquidated, it is less than $F$. Hence for  $\mathrm{B}_i$ in $\mathscr{D}\backslash \mathscr{L}$, the second term under the ``min'' in equation \eqref{e:shortfall} is the relevant one; writing that elementwise gives the following:
\begin{align}
    D_{ i}
            &=  \sum_{k=1}^N Q_{ik}D_{k}-\theta y+\ell
            \\
            &=  \sum_{k\in \mathscr{L}} Q_{ik}D_{k}+
                    \sum_{k\in \mathscr{D}\backslash \mathscr{L}} Q_{ik}D_{k}
                    +\sum_{k\notin \mathscr{D}} Q_{ik}D_{k}-\theta y+\ell
                    \\
                &=  \sum_{k\in \mathscr{L}} Q_{ik}D_{k}+
                    \sum_{k\in \mathscr{D}\backslash \mathscr{L}} Q_{ik}D_{k}-\theta y+\ell ,
\end{align}
having used that shortfall is zero for banks that do not default ($D_k = 0$ for $k \not \in \mathscr{D}).$
Re-writing the above in block matrix form, using that fact that liquidated banks repay nothing ($D_k = F$ for $k \in \mathscr{L}$), and rearranging gives:
\begin{align}
    \mathbf{D}_{\mathscr{D}\backslash \mathscr{L}}
            &=  \mathbf{Q}_{\mathscr{D}\backslash \mathscr{L},\mathscr{D}\backslash \mathscr{L}}\mathbf{D}_{\mathscr{D}\backslash \mathscr{L}}+
                \mathbf{Q}_{ \mathscr{D}\backslash \mathscr{L},\mathscr{L}}  \mathbf{D}_{\mathscr{L}}-( \theta y  - \ell) \bm{1}\\
            &=  \mathbf{Q}_{\mathscr{D}\backslash \mathscr{L},\mathscr{D}\backslash \mathscr{L}}\mathbf{D}_{\mathscr{D}\backslash \mathscr{L}}+
                \mathbf{Q}_{ \mathscr{D}\backslash \mathscr{L},\mathscr{L}}  F\bm{1}- (\theta y - \ell ) \bm{1} \\
            &=(\mathbf{I}-\mathbf{Q}_{\mathscr{D}\backslash \mathscr{L},\mathscr{D}\backslash \mathscr{L}})^{-1}  (\mathbf{Q}_{ \mathscr{D}\backslash \mathscr{L},\mathscr{L}}F\bm{1}- (\theta y - \ell) \bm{1}),
\end{align}
where the last expression comes from solving for   $\mathbf{D}_{\mathscr{D}\backslash \mathscr{L}}$ and rearranging. \end{proof}

\begin{lemma} \label{l:radius 2} $   \big( \mathbf{I} +  \mathbf{Q}_{\mathscr{L},\mathscr{D}\backslash \mathscr{L}} (\mathbf{I}-\mathbf{Q}_{\mathscr{D}\backslash \mathscr{L},\mathscr{D}\backslash \mathscr{L}})^{-1} \big)  \big(\ell -\theta y  \big)\bm{1}
    >
   \mathbf{\tilde Q} F\bm{1}$,
    where
    \begin{equation} \label{e:Q tilde}
   \mathbf{\tilde Q} :=   \big( \mathbf{I} -\mathbf{Q}_{\mathscr{L},\mathscr{D}\backslash \mathscr{L}} (\mathbf{I}-\mathbf{Q}_{\mathscr{D}\backslash \mathscr{L},\mathscr{D}\backslash \mathscr{L}})^{-1}  \mathbf{Q}_{ \mathscr{D}\backslash \mathscr{L},\mathscr{L}}-\mathbf{Q}_{\mathscr{L},\mathscr{L}}\big ) .
   \end{equation}
\end{lemma}

\begin{proof}
Banks that are liquidated repay zero  (equation \eqref{eq:optimal_repayment_cases}), so
\begin{align}
    F &< \sum_{k=1}^N    Q_{ik}D_{k}    -\theta y   +\ell
    \\
    &=  \sum_{k\in \mathscr{L}} Q_{ik}D_{k}+
                    \sum_{k\in \mathscr{D}\backslash \mathscr{L}} Q_{ik}D_{k}
                    +
                    \sum_{k\notin \mathscr{D}} Q_{ik}D_{k}-\theta y+\ell
                \\
        &= \sum_{k\in \mathscr{L}} Q_{ik}D_{k}+
                    \sum_{k\in \mathscr{D}\backslash \mathscr{L}} Q_{ik}D_{k}-\theta y+\ell   \\
            &=\sum_{k\in \mathscr{L}} Q_{ik}F+
                    \sum_{k\in \mathscr{D}\backslash \mathscr{L}} Q_{ik}D_{k}-\theta y+\ell ,
\end{align}
having used that shortfall is zero for banks that do not default ($D_k = 0$ for $k \not \in \mathscr{D})$ and $F$ for those that are liquidated ($D_k = F$ for $k \in \mathscr{L}$).
The above can be re-written in block-matrix notation,
 $
    F\bm{1}< \mathbf{Q}_{\mathscr{L},\mathscr{D}\backslash \mathscr{L}} \mathbf{D}_{\mathscr{D}\backslash \mathscr{L}}
    +\mathbf{Q}_{\mathscr{L}, \mathscr{L}} F\bm{1}- (\theta y -\ell) \bm{1}$, so the expression for $\mathbf{D}_{\mathscr{D}\backslash \mathscr{L}}$ from Lemma \ref{l:radius 1} can be substituted in to get:
\begin{equation}
    \mathbf{Q}_{\mathscr{L}, \mathscr{D}\backslash \mathscr{L} } (\mathbf{I}-\mathbf{Q}_{\mathscr{D}\backslash \mathscr{L}, \mathscr{D}\backslash \mathscr{L}})^{-1}  (\mathbf{Q}_{ \mathscr{D}\backslash \mathscr{L},\mathscr{L}}F\bm{1}- (\theta y  - \ell) \bm{1})
    +\mathbf{Q}_{\mathscr{L}, \mathscr{L}} F\bm{1}- (\theta y -\ell)  \bm{1}>F\bm{1} .
\end{equation}
Rearranging the above gives the expression in the lemma.
\end{proof}

Now we compute a bound on the harmonic distance $d$.
First we use the definition of $d$ (Definition \ref{d:HD}) to write in block matrix form:
\begin{equation} \label{e:block HD}
\renewcommand\arraystretch{1.5}\left\{
\begin{array}{l}
    \textbf{d}_{j \to \mathscr{L}}    =\textbf{1}+\textbf{Q}_{\mathscr{L},\mathscr{L}}\textbf{d}_{j \to \mathscr{L}}+
    \textbf{Q}_{\mathscr{L},\mathscr{D}\backslash\mathscr{L}}\textbf{d}_{j \to \mathscr{D}\backslash\mathscr{L}},
    \\
    \textbf{d}_{j \to \mathscr{D}\backslash\mathscr{L}}    =\textbf{1}+\textbf{Q}_{\mathscr{D}\backslash\mathscr{L},\mathscr{L}}\textbf{d}_{j \to \mathscr{L}}+
    \textbf{Q}_{\mathscr{D}\backslash\mathscr{L},\mathscr{D}\backslash\mathscr{L}}\textbf{d}_{j \to \mathscr{D}\backslash\mathscr{L}},
    \end{array}
\right.
\end{equation}
where $\textbf{d}_{j \to \mathscr{L}}$ and  $\textbf{d}_{j \to \mathscr{D}\backslash\mathscr{L}}$ are vectors that capture the harmonic distances from B$_j$ to each of (i) the liquidated and (ii) the defaulting but not liquidated banks, respectively (cf.\ equations (B19) and (B20) in \citetalias{AOT}). (NB: As, by hypothesis, B$_j$ is the only not-shocked bank, there are no additional terms to not defaulting banks.)

Solving for the system in equation \eqref{e:block HD}---solving for  $\textbf{d}_{j \to \mathscr{D}\backslash\mathscr{L}}$ in the second equation and substituting it into the first---gives
\begin{equation}
     \textbf{d}_{j \to \mathscr{L}}    =\textbf{1}+\textbf{Q}_{\mathscr{L},\mathscr{L}}\textbf{d}_{j \to \mathscr{L}}+
    \textbf{Q}_{\mathscr{L},\mathscr{D}\backslash\mathscr{L}} (\textbf{I}-\textbf{Q}_{\mathscr{D}\backslash\mathscr{L},\mathscr{D}\backslash\mathscr{L}})^{-1}(\textbf{1}+\textbf{Q}_{\mathscr{D}\backslash\mathscr{L},\mathscr{L}}\textbf{d}_{j \to \mathscr{L}})
\end{equation}
or, given the definition of $ \mathbf{\tilde Q}$ in equation \eqref{e:Q tilde},
 $
     \mathbf{\tilde{Q}}\textbf{d}_{j \to \mathscr{L}} =(\textbf{I}+\textbf{Q}_{\mathscr{L},\mathscr{D}\backslash\mathscr{L}} (\textbf{I}-\textbf{Q}_{\mathscr{D}\backslash\mathscr{L},\mathscr{D}\backslash\mathscr{L}})^{-1})\textbf{1}.  $
From  here, we can use Lemma \ref{l:radius 2} to write
\begin{equation}
    \mathbf{\tilde{Q}}\textbf{d}_{j \to \mathscr{L}}
       > \mathbf{\tilde{Q}}\frac{F}{\ell-\theta y}\textbf{1}.
\end{equation}
 As  $\mathbf{\tilde Q}$ is non-singular M-matrix, its inverse exists and is elementwise non-negative,\footnote{The result follows Theorem 2 of \cite{Plemmons-1977} and exercise 5.8 of \citeauthor{Berman-Plemmons-1979} (\citeyear{Berman-Plemmons-1979}, p.\ 159) given that  $\mathbf{\tilde Q}$ is the Schur complement of the non-singular $M$-matrix
    $$
    \begin{bmatrix}
    \textbf{I} - \mathbf{Q}_{\mathscr{D}\backslash \mathscr{L},\mathscr{D}\backslash \mathscr{L}} & -\mathbf{Q}_{\mathscr{D}\backslash \mathscr{L}, \mathscr{L}} \\
    -\mathbf{Q}_{ \mathscr{L},\mathscr{D}\backslash \mathscr{L}} & \textbf{I} -\mathbf{Q}_{\mathscr{L}, \mathscr{L}}
    \end{bmatrix}.
    $$
    }
 this says that if B$_i$ is liquidated, then
    \begin{equation}
   {d}_{j \to i}  \geq \frac{F}{\ell - \theta y}
   \equiv d^{LT},
    \end{equation}
    per the definition of $d^{LT}$ in the proposition.  That is the desired result.

\textbf{Statement (ii).}
    As, by hypothesis, no bank is liquidated, we have that for any B$_i$ $D_{i}<F_{i\rightrightarrows}$, from the definition of the shortfall $\mathbf{D}$.
        Thus, from the equilibrium equation for the shortfall (Lemma \ref{lm:shortfall} with $\alpha = 1$) and the observation that no shocked bank repays in full, $D_i > 0$ (equation \eqref{eq:optimal_repayment_cases} given the assumption that banks have zero net positions), we have
\begin{equation} \label{e:D is HD}
    D_{i}=\sum_{k\neq i}Q_{ik}D_{k}+\ell-\theta y    ,
\end{equation}
for any B$_i$ for $i \neq j$, where, remember, B$_j$ is the not-shocked bank. Dividing both sides of equation \eqref{e:D is HD} by $\ell - \theta y$ says that $D_i/(\ell - \theta y)$ solves $x_i = 1 + \sum_{k \neq i} Q_{ik} x_k$ for all $i$. By the definition (and uniqueness) of the harmonic distance, that implies that
    \begin{equation}
        d_{j \to i} = \frac{D_i}{\ell - \theta y}.
    \end{equation}
As $D_i < F$ by hypothesis,  $d_{j\to i} < F/(\ell - \theta y) \equiv d^{LT},$ as desired.  \qed

\subsection{Proof of Proposition \ref{p:bottleneck} (\nameref{p:bottleneck})}

We prove the two statements of the proposition sequentially. Statement (i) relies on Lemma \ref{l:beta and d} above (which, recall, depends only on the network structure, not the maturity of debt despite being stated within the benchmark of Section~\ref{s:ST}). Statement (ii) instead uses a direct cut-payment argument.

\textbf{Proof of statement (i).}   For $\beta > \beta^{LT} \equiv 4 \sqrt{ \frac{\ell - \theta y }{N F} } $, we have, from Lemma \ref{l:beta and d}, that
    \begin{equation}
        d_{j \to i} \leq \frac{16}{N \beta^2} < \frac{ 16}{N (\beta^{LT})^2}
            = \frac{F}{\ell - \theta y} \equiv d^{LT},
    \end{equation}
    per the definition of $d^{LT}$ in Proposition \ref{p:radius}. That result implies that when  B$_j$ is not shocked, then B$_i$ is not liquidated.

\textbf{Proof of statement (ii).} Let $\mathscr B$ be a nonempty proper subset attaining the minimum in Definition~\ref{d:bottleneck}. Since the network is symmetric, replacing $\mathscr B$ with $\mathscr B^c$ leaves the cut value unchanged, so choose the minimizing side such that the unique not-shocked bank B$_j$ is not in $\mathscr B$. Thus every bank in $\mathscr B$ is shocked.

Suppose, in anticipation of a contradiction, that no bank is liquidated. A shocked bank cannot repay in full: By symmetry and regularity, it receives at most its total interbank claims $F$, while $\ell>\theta y$ means that repaying its total liabilities $F$ would require an incoming payment strictly greater than $F$. Hence every bank in $\mathscr B$ defaults without being liquidated, and equation~\eqref{e:total repayment} implies
\begin{equation}
    R_{i\leftleftarrows}-R_{i\rightrightarrows}=\ell-\theta y
    \quad\text{for every }i\in\mathscr B.
\end{equation}
Summing over $\mathscr B$, cancelling payments internal to $\mathscr B$, using $R_{k\to i}\leq F_{k\to i}$, and then using symmetry, regularity, and the definition of $\beta$ gives
\begin{align}
    (\ell-\theta y)|\mathscr B|
    &=\sum_{i\in\mathscr B}\sum_{k\in\mathscr B^c}
    \big(R_{k\to i}-R_{i\to k}\big)\\
    &\leq\sum_{i\in\mathscr B}\sum_{k\in\mathscr B^c}F_{k\to i}\\
    &=F\beta|\mathscr B||\mathscr B^c|\\
    &\leq FN\beta|\mathscr B|.
\end{align}
But $\beta<\beta_{LT}\leq(\ell-\theta y)/(2NF)$ makes the last expression strictly less than $(\ell-\theta y)|\mathscr B|/2$, a contradiction. Thus at least one bank in $\mathscr B$ is liquidated. \qed

\subsection{Proof of Lemma~\ref{l:constrained efficiency} (\nameref{l:constrained efficiency})}

The argument is in the text. \qed

\subsection{Proof of Lemma~\ref{l:high debt} (\nameref{l:high debt})}

    We allow bank heterogeneity by indexing $\ell$ and $y$ with subscripts since the lemma is invoked in the proof of Proposition \ref{p:exponential_implement_greedy}.
    We show that if a not-shocked bank, say B$_i$, makes net payment less than $\theta y_i$ no matter how high $\alpha$ is, it is impossible  that another bank, say B$_j$, is liquidated.
        From equation \eqref{eq:optimal_repayment_cases} with proper indexing,
         the not-shocked bank's net payment is
        \begin{equation}
         R_{i\rightrightarrows}-R_{i\leftleftarrows}
         =
        \left\{
            \renewcommand\arraystretch{1.2}
            \begin{array}{cl}
            \theta y_i & \text{if defaults,} \\
                \alpha F_{i\rightrightarrows}-R_{i\leftleftarrows} & \text{otherwise.}
            \end{array}
            \right.
        \end{equation}

        Since the network is connected and has zero net positions, it is strongly connected,\footnote{To see this, contract each strongly connected component into one node.  If there were more than one component, the resulting finite directed acyclic graph would have a source component $\mathscr C$.  No liabilities enter $\mathscr C$, while zero net positions, summed over its banks, imply that total liabilities leaving $\mathscr C$ equal total liabilities entering it and hence are also zero.  This contradicts (weak) connectivity unless $\mathscr C$ is the whole network.} meaning there is a directed path
            \(
            i=i_0\to i_1\to\cdots\to i_m=j
            \)
        with $F_{i_t\to i_{t+1}}>0$ on every edge.

        Next we show along this path, if $R_{i_{t-1}\to i_t}\geq\alpha a_t-b_t$ for some $a_t>0$ and finite $b_t\geq0$ independent of $\alpha$, then $R_{i_t\to i_{t+1}}\geq\alpha a_{t+1}-b_{t+1}$, for some $a_{t+1}>0$ and finite $b_{t+1}\geq0$ independent of $\alpha$. To see this, set $a_1:=F_{i_0\to i_1}$ and $b_1:=0$, and recursively
                \[
                \begin{aligned}
                a_{t+1}
                &:=\hat F_{i_t\to i_{t+1}}
                \min\left\{F_{i_t\rightrightarrows},a_t\right\}>0,
                \\
                b_{t+1}
                &:=\hat F_{i_t\to i_{t+1}}
                \left[b_t+\ell_{i_t}\sigma_{i_t}-\theta_{i_t}y_{i_t}\right]^+<\infty.
                \end{aligned}
                \]

        Since B$_{i_0}$ repays in full,
                \(
                    R_{i_0\to i_1}=\alpha F_{i_0\to i_1}.
                \)
         We now apply the repayment equation successively along the path.  If $R_{i_{t-1}\to i_t}\geq\alpha a_t-b_t$ for some $a_t>0$ and finite $b_t\geq0$, then
                \begin{align*}
                R_{i_t\to i_{t+1}}
                &=\hat F_{i_t\to i_{t+1}}R_{i_t\rightrightarrows} \\
                &\geq
                \hat F_{i_t\to i_{t+1}}
                \left[
                \min\left\{
                \alpha F_{i_t\rightrightarrows},
                \alpha a_t-b_t+\theta_{i_t}y_{i_t}-\ell_{i_t}\sigma_{i_t}
                \right\}
                \right]^+ \qquad \text{since } R_{i_{t-1}\to i_t}\geq\alpha a_t-b_t\\
                &\geq\alpha a_{t+1}-b_{t+1} \qquad \text{since $[\min\{\alpha u,\alpha v-c\}]^+\geq\alpha\min\{u,v\}-[c]^+$ for $u,v\geq0$}.
                \end{align*}
         Induction therefore implies $R_{i_{m-1}\to j}\geq\alpha a_m-b_m$.  For sufficiently large $\alpha$, this exceeds B$_j$'s finite liquidity deficit, contradicting its liquidation.  Taking the largest of these finitely many thresholds over ordered pairs $(i,j)$ and shock profiles $\bm{\sigma}$ proves that, if any bank is liquidated, every not-shocked bank makes net payment $\theta y_i$.

\subsection{Proof of Corollary \ref{c:high debt} (\nameref{c:high debt})}

Suppose (in anticipation of a contradiction) that at least one  bank is liquidated. From Lemma \ref{l:high debt} and equation \eqref{eq:optimal_repayment_cases}, we know that for each B$_i$ the net payment is
\begin{equation}
    R_{i \rightrightarrows}- R_{i \leftleftarrows}
    \renewcommand\arraystretch{1.5}
    \left\{
    \begin{array}{cl}
        = \theta y - \ell \sigma_i & \text{if B$_i$ is not liquidated}, \\
        > \theta y - \ell \sigma_i & \text{if B$_i$ is liquidated.}
    \end{array}
    \right.
    \end{equation}
Combining market clearing (equation \eqref{eq:payment_clearing}) with the expression above, we have that
\begin{equation}
0 =  \sum  \big(  R_{i \rightrightarrows} - R_{i \leftleftarrows} \big) > \sum \big( \theta y - \ell  \sigma_i \big)
= N \theta y - S \ell ,
\end{equation}
contradicting the hypothesis that $N \theta y > S \ell $. Therefore no bank can be liquidated, as desired. \qed

\subsection{Proof of Lemma~\ref{lm:exponentially_dominated} (\nameref{lm:exponentially_dominated})}
Use the pro rata condition that $R_{i \to j} = \hat F_{i \to j } R_{i \rightrightarrows}$  for any $j$ (including $j = i^*$), by Definitions~\ref{d:assortativity} and~\ref{d:s-dominance}, after relabeling banks according to their common hierarchy,
    \begin{align}
    R_{i\to j}
    &=\hat F_{i\to j}R_{i\rightrightarrows}\\
    &\leq s^{j-i^*}\hat F_{i\to i^*}R_{i\rightrightarrows}\\
    &=s^{j-i^*}R_{i\to i^*}.
    \end{align}
    Since \(B_{i^*}\) is liquidated, \(R_{i^*\to j}=0\). Therefore,
        \begin{align}
        R_{j\leftleftarrows}
        &=\sum_{i\neq j,i^*}R_{i\to j}\\
        &\leq s^{j-i^*}\sum_{i\neq j,i^*}R_{i\to i^*}\\
        &\leq s^{j-i^*}R_{i^*\leftleftarrows}.
        \end{align}
          \qed

\subsection{Proof of Lemma~\ref{lm:total_liquidity} (\nameref{lm:total_liquidity})}

    Since liquidated banks make zero repayments, $R_{i \rightrightarrows} = 0$ for $i \in \mathscr{L}$, equation \eqref{eq:optimal_repayment_cases} implies that each receives payment $R_{i \leftleftarrows} < \ell - \theta y.$
    Applying this to B$_{i^*}$, the highest-ranked liquidated bank, and using Lemma~\ref{lm:exponentially_dominated}, we have the following:
    \begin{align}
        \sum_{i\in\mathscr{L}}R_{i\leftleftarrows}&\leq \sum_{i\in\mathscr{L}} s^{i-i^*} R_{i^*\leftleftarrows}
        \\
        &\leq R_{i^*\leftleftarrows}\sum_{i=0}^{\infty} s^{i}
        \\
        &=R_{i^*\leftleftarrows}\frac{1}{1-s}\\
        &<\frac{\ell-\theta y}{1-s}. \label{e:lemma 9 proof}
    \end{align}
\qed

\subsection{Proof of Proposition \ref{p:constrained efficiciency} (\nameref{p:constrained efficiciency})}

As the exponential network is connected, we know from Lemma \ref{l:constrained efficiency} and Corollary \ref{c:high debt} that if $L^* = 0$ then no bank is liquidated as long as  $\alpha$ is high.
     Hence we focus on the case in which at least one bank is liquidated in the constrained-efficient outcome, $L^* \geq 1$.

    Recall that in this case it suffices to show the following:

        \begin{itemize}

            \item[(i)] Each bank that is not liquidated makes the maximum net payment it can without being liquidated, $R_{i\rightrightarrows} - R_{i \leftleftarrows} = \theta y - \ell \sigma_i$, for B$_i$ not liquidated.

            \item[(ii)] The banks that are liquidated receive a total net payment that would be insufficient to save any one of them, $-\sum_{i \in \mathscr{L}} \big(R_{i\rightrightarrows} - R_{i \leftleftarrows}\big) < \ell -\theta y $.

        \end{itemize}

        \noindent The first property follows from Lemma \ref{l:high debt} and equation \eqref{eq:optimal_repayment_cases}.

        The second property follows from two steps. The first is to use Lemma \ref{lm:total_liquidity} and the definition of $s^*$ to bound the liquidated banks' net payment in terms of $L^*$:
        \begin{align}
        -\sum_{i \in \mathscr{L}} \big(R_{i\rightrightarrows} - R_{i \leftleftarrows}\big)
        &< \frac{\ell - \theta y}{1 - s}
        \\
        & \leq \frac{\ell - \theta y}{1 -s^*}
        \\
        &= N \theta y - S \ell + (1 +L^*)(\ell - \theta y) . \label{e:net payments bound}
        \end{align}
        The second step is to use market clearing, $\sum_{i=1}^N \big(R_{i\rightrightarrows} - R_{i \leftleftarrows}\big)  = 0$ by equation \eqref{eq:payment_clearing}, to write  the LHS above in terms of the number of liquidated banks $L$, using that (i) not-shocked banks make net payment $\theta y$ (by Lemma~\ref{l:high debt})  and (ii) shocked, not-liquidated banks make net payment $\theta y - \ell $ (by equation~\eqref{eq:optimal_repayment_cases}):
        \begin{align} \label{eq:ce_insufficient_liquidity}
        - \sum_{i \in \mathscr{L}} \big(R_{i\rightrightarrows} - R_{i \leftleftarrows}\big)
         &= \sum_{i \in \mathscr{L}^c} \big(R_{i\rightrightarrows} - R_{i \leftleftarrows}\big)
         \\
         &=
           \
            \sum_{i\in \mathscr{L}^c\, : \, \sigma_i = 0}
         \big(R_{i\rightrightarrows} - R_{i \leftleftarrows}\big)
         \ +
         \sum_{i \in \mathscr{L}^c\, : \, \sigma_i = 1}
         \big(R_{i\rightrightarrows} - R_{i \leftleftarrows}\big)
         \\
         &= (N-S) \theta y + (S - L) ( \theta y - \ell).
        \end{align}
        Combining this with the bound in equation \eqref{e:net payments bound} and canceling terms says $L < 1 + L^*$. As $L$ and $L^*$ are integers, and $L^* \leq L$ by Lemma \ref{l:constrained efficiency}, it must be that $L = L^*$.

         (The assumption that $\frac{S \ell - N \theta y}{\ell - \theta y}$ not be an integer was required for $s^* > 0$ and thus for the exponential network with base $s < s^*$ to be well defined.)
        \qed

\subsection{Proof of Proposition \ref{p:exponential almost} (\nameref{p:exponential almost})}

With the weaker notion of efficiency, we need to show only that the banks that are liquidated  receive a total net payment  insufficient to save any \emph{two} of them, $-\sum_{i \in \mathscr{L}} \big(R_{i\rightrightarrows} - R_{i \leftleftarrows}\big) < 2\big(\ell -\theta y\big)$. Given the proof of Proposition \ref{p:constrained efficiciency}, that is all we need to show.
As in that proof, it follows Lemma \ref{lm:total_liquidity} along with the definitions of $s^* (=1/2)$ and $L^*$:
        \begin{align}
        -\sum_{i \in \mathscr{L}} \big(R_{i\rightrightarrows} - R_{i \leftleftarrows}\big)
        & < \frac{\ell - \theta y}{1 -s^*}
        \\
        &= 2 (\ell - \theta y)
        \end{align}
    as desired. \qed

\subsection{Proof of Lemma~\ref{lm:inefficiency other networks} (\nameref{lm:inefficiency other networks})}

 Suppose (in anticipation of a contradiction) that  a fully connected network $\mathbf{F}$ achieves constrained efficiency, i.e.\ that the number of liquidated banks is $L^* = \frac{S\ell-N\theta y}{\ell-\theta y}$, having used the assumptions that $\frac{S\ell-N\theta y}{\ell-\theta y}$ is an integer and that $S \ell > N \theta y$ in conjunction with the definition of $L^*$ (equation \eqref{e:L*}).

 As each bank that is not shocked pays at most $\theta y$ and each that is shocked but not liquidated pays exactly  $\theta y-\ell  \sigma_i$ (equation \eqref{eq:optimal_repayment_cases}), we can use market clearing to bound the total payment to the liquidated banks as follows:
    \begin{align}
    - \sum_{i \in \mathscr{L}} \big( R_{i \rightrightarrows} - R_{i \leftleftarrows}
    \big)
    &= \sum_{i \not \in \mathscr{L}} \big( R_{i \rightrightarrows} - R_{i \leftleftarrows}
    \big)
    \\
    &\leq \sum_{i \not \in \mathscr{L}} (\theta y - \sigma_i \ell )
    \\
    &= (N - S) \theta y + (S - L^*) ( \theta y - \ell )
    \\
    &= 0.
    \end{align}
    I.e.\ liquidated banks receive no (positive) payment.

    But shocked, not-liquidated banks must receive positive payment (otherwise they would be liquidated by equation \eqref{eq:optimal_repayment_cases}).
    Given the hypothesis that the network is fully connected, that contradicts the assumption that payments are pro rata (equation \eqref{eq:pro_rata}): Any bank that has debt to a not-liquidated bank must have debt to a liquidated bank too and it cannot make a positive payment to one but not the other. \qed

\subsection{Proof of Proposition \ref{p:exponential_not_pairwise} (\nameref{p:exponential_not_pairwise})}

Since the $\alpha$ is sufficiently large, all not-shocked banks default and obtain a payoff of $(1-\theta)y$, and if exactly one of B$_{N-1}$ and B$_N$ is shocked, it is liquidated. We show they can deviate to large mutual debts which help them avoid liquidation when one of them is shocked and do not affect their payoff when both or neither of them is shocked.

Let
\(
    O_i=\alpha\sum_{k\notin\{N-1,N\}}F_{i\to k}
\) for $i\in\{N-1,N\}$
denote the aggregate liabilities of B\(_i\) to banks outside the pair. Consider a deviation $F'_{N-1\to N}=F'_{N\to N-1}=F$ such that
\[
    F>
    \max_{i\in\{N-1,N\}}
    \left\{
        \alpha F_{N-1\to N},
        \alpha F_{N\to N-1},
        \theta y,
        \frac{(\ell-\theta y) O_i}{2\theta y-\ell}
    \right\}.
\]

Consider first a state in which only one of the two banks, B$_i$, is shocked and the other bank B\(_j\) is not shocked. Since \(L^*\geq1\), B\(_i\) is liquidated and  obtains a  payoff of zero.

After the deviation, B\(_j\)'s total liabilities are \(F+O_j>\theta y\). And it pays at least $\frac{F}{F+O_j}\theta y>\ell-\theta y$ to B$_i$
where the strict inequality follows from \(F(2\theta y-\ell)>(\ell-\theta y) O_j\). Thus, B\(_i\) avoids liquidation. The not-shocked bank B\(_j\) is not worse off in this state since it obtains at least $(1-\theta)y$.  The same argument applies symmetrically when B\(_j\) is shocked and B\(_i\) is not.

If neither bank in the pair is shocked, both obtain at least the original payoff $(1-\theta)y$.

When both are shocked, if \(L^*\geq2\), both are liquidated, so neither can be made worse off. Suppose instead \(L^*=1\). Then B\(_{N-1}\) is saved and B\(_N\) is liquidated in the original equilibrium. We claim that B\(_{N-1}\) still survives after the deviation. If B\(_{N-1}\) were liquidated, then B\(_N\)'s only incoming payments would come from banks outside the pair. Given the exponential structure, \(R'_{N\leftleftarrows}\leq s R'_{N-1\leftleftarrows}<\ell-\theta y\), so B\(_N\) would be liquidated too. But once B\(_{N-1}\) and B\(_N\) are both liquidated, neither of them makes any payment, and the debts between every other pair remain unchanged, so the equilibrium payment vector is also the solution to the payment equation before deviation. But this contradicts that all payment equilibria induce the same set of liquidated banks (see the proof of Proposition~\ref{p:existence}): The constructed payment vector is an equilibrium of the pre-deviation network in which B$_{N-1}$ is liquidated, whereas in the original equilibrium it is not.

Since the shock distribution assigns positive probability to both states in which only one of B$_{N-1}$ and B$_N$ is shocked, their expected payoffs are strictly improved.

\subsection{Proof of Proposition \ref{p:star stable} (\nameref{p:star stable})}
We first present the equilibrium of the star network for different values of $F$.
\begin{lemma}
    Suppose $1<S<N$ and $\ell-\theta y\leq (N-S)\theta y$ so that the system has enough liquidity to save at least one bank.\footnote{We only need this for the case when the core bank is shocked.} Let $\bar{F}:=\frac{N}{S}\theta y$, $\bar{F}^c:=\frac{N\theta y-\ell}{S-1}$,  $\underline{F}^c:=\frac{\ell-\theta y}{N-S}$ and $\underline{F}:=\frac{\theta y}{S}$.
    \begin{itemize}
        \item When only peripheral banks are shocked, all of them are liquidated. The core bank defaults if $F>\underline{F}$ and the not-shocked peripheral banks default if $F>\bar{F}$.
        \item When the core bank is also shocked, all the shocked peripheral banks are liquidated. The core bank is liquidated if $F<\underline{F}^c$. The not-shocked banks default when $F> \bar{F}^c.$
    \end{itemize}
\end{lemma}
\begin{proof}
    The liquidation of the peripheral banks comes from the symmetry argument given insufficient liquidity in the network.

    The not-shocked peripheral banks default only if the core defaults on them and the default is large enough. When the core is not shocked, it defaults iff its payment at the verge of default is weakly less than its liabilities
    \begin{equation}
        R_{1\rightrightarrows}=(N-S-1)F+\theta y\leq (N-1)F \iff F\geq \underline{F}.
    \end{equation}
    The peripheral banks default iff
    \begin{equation}
        \frac{R_{1\rightrightarrows}}{N-1}+\theta y< F\iff F> \bar{F}.
    \end{equation}

    When the core is  shocked, it is liquidated if it doesn't receive enough liquidity at the verge of liquidation
    \begin{equation}
        (N-S)\min\{F,\theta y\}< \ell-\theta y \iff F< \underline{F}^c \quad \text{ given }\quad \ell-\theta y\leq (N-S)\theta y.
    \end{equation}
    As long as the not-shocked peripheral banks repay in full, and hence at the verge of their default, the core pays $R_{1\rightrightarrows}=(N-S)F+\theta y-\ell$ once it survives.
    The not-shocked peripheral banks default if
    \begin{equation}
        \frac{R_{1\rightrightarrows}}{N-1}+\theta y< F \iff F> \bar{F}^c.
    \end{equation}
\end{proof}

    Without deviation, given that $F\geq \underline{F}^c> \underline{F}$, the core bank is never liquidated, but always defaults and obtains  $(1-\theta)y$ in every shock state. We show that the core bank can never obtain a higher payoff in every state, no matter the deviation, by showing that the maximum residual pledgeable assets (i.e., pledgeable assets minus liabilities assuming the core does not default) are negative.  Let's consider a deviation between the core bank B$_1$ and a peripheral bank B$_2$  to $F_{1\to2}$ and $F_{2\to1}$.     Denote the liquidity deficit $D:=\ell-\theta y$. The assumptions $S\ell>(2N-S)\theta y$ and $S>1$ imply
    \begin{equation}
        \frac{D}{N-S}>\frac{2\theta y}{S}\geq\frac{\theta y}{S-1}.
    \end{equation}
    \begin{itemize}
        \item Suppose first that the core is not shocked and B$_2$ is shocked.  If the core were solvent, it would pay $F_{1\to2}$ to B$_2$ and $F$ to each of the other $N-2$ peripheral banks. Hence the core's residual pledgeable assets  are at most
        \begin{align*}
            & \theta y+(F_{1\to2}-D)^+ +(S-1)(F-D)^+ +(N-S-1)F-F_{1\to2}-(N-2)F.\\
            = & \theta y+(F_{1\to2}-D)^+-F_{1\to 2}+(S-1)[(F-D)^+-F]
        \end{align*}
        When $F\leq D$, the expression is smaller than $\theta y-(S-1)F$\footnote{We used the (self-evident) auxiliary inequality $(x-y)^+-x<0\quad \forall x,y>0$ here and below.} which is negative since $F>\frac{\theta y}{S-1}$; when $F>D$, the expression is smaller than $\theta y-(S-1)D$ which is smaller than $0$ since $D>\frac{(N-S)\theta y}{S-1}\geq \frac{\theta y}{S-1}$.

        \item Next suppose both the core and B$_2$ are not shocked (only possible when $N-S\geq 2$). Since B$_2$ can repay at most $\theta y+F_{1\to2}$, the core's residual pledgeable assets are at most
        \begin{align*}
            &\theta y+(\theta y+F_{1\to2})+S(F-D)^+ +(N-S-2)F-F_{1\to2}-(N-2)F.\\
        = &2\theta y +S[(F-D)^+-F]
        \end{align*}
        When $F\geq D$, the expression equals $2\theta y-SD<0$ since $D>\frac{2\theta y}{S}(N-S)\geq\frac{2\theta y}{S}$; when $F<D$, the expression equals $2\theta y-SF<0$ since $F\geq \frac{D}{N-S}>\frac{2\theta y}{S}$.

        \item Now suppose both the core and B$_2$ are shocked. The core's residual pledgeable assets are at most
        \begin{align*}
            &\theta y-\ell+(F_{1\to2}-D)^+ +(S-2)(F-D)^+ +(N-S)F-F_{1\to2}-(N-2)F=\\
            =&-D+[(F_{1\to2}-D)^+-F_{1\to2}]+(S-2)[(F-D)^+-F]<-D<0.
        \end{align*}

        \item Finally, suppose the core is shocked and B$_2$ is not shocked. The core's residual pledgeable assets are at most
        \begin{align*}
            &\theta y-\ell+(\theta y+F_{1\to2})+(S-1)(F-D)^+ +(N-S-1)F-F_{1\to2}-(N-2)F.\\
            =&2\theta y-\ell +(S-1)[(F-D)^+-F]\\
            =& \{2\theta y+S[(F-D)^+-F]\}-\{\ell+[(F-D)^+-F]\}
        \end{align*}
        We have proven that the term in the first curly bracket is negative and that the term in the second, $\ell+[(F-D)^+-F]\geq \ell-D=\theta y>0$, is positive, so the entire expression is negative. Intuitively, being directly hit is worse than being defaulted on by a shocked bank.
        \end{itemize}

    Thus, after any possible core--periphery deviation, the core bank's payoff is at most its original payoff $(1-\theta)y$ in every shock state, and hence in expectation.

\subsection{Proof of Proposition~\ref{p:star still stable} (\nameref{p:star still stable})}
 The large-shock condition
    \(
        N\theta y<(N-1)\ell,
    \)
    is equivalent to $\ell-\theta y>\theta y/(N-1)$. Since $F=\ell-\theta y<\theta y/(N-2)\leq \theta y$, a not-shocked peripheral can repay $\ell-\theta y$ in full, which saves the shocked core. Since $\ell-\theta y>\theta y/(N-1)$, the core cannot save all $N-1$ shocked peripheral banks in the original star but it defaults. The core obtains $(1-\theta)y$ in every state without deviation. A peripheral bank, say, B$_2$, obtains a positive payoff only when it is the unique unshocked peripheral; in that state it pays $\ell-\theta y$ to the core, receives no repayment from the core, and obtains
    \(
        y-(\ell-\theta y).
    \)

    By symmetry, consider a deviation $F_{1\to2}$ and $F_{2\to1}$ between B$_1$ and B$_2$. We will show that i) the core bank B$_1$ can only strictly benefit in the state when it is not shocked, and it must be that $F_{1\to2}<\ell-\theta y$, and thus, B$_2$ remains liquidated whenever shocked; ii) for B$_1$ to strictly benefit in expectation, it must not be liquidated in the state when B$_2$ is not shocked, and it requires $F_{2\to 1}\geq \ell-\theta y$, but then B$_2$ cannot strictly gain in that state.

    First, suppose B$_1$ is the unique not-shocked bank. It is the only state in which B$_1$ can strictly gain.\footnote{After deviation, \(B_1\) need not have zero net position when shocked and may not default even when shocked. Nevertheless, only the unique not-shocked peripheral supplies external pledgeable liquidity, bounded by \(\theta y\); payments recycled through shocked peripherals cannot create positive net resources. For \(B_1\) to become solvent, this external liquidity must cover its liquidity shortfall and all \(N-2\) unchanged liabilities, requiring at least \((N-1)(\ell-\theta y)>\theta y\), which is impossible. Hence \(B_1\) cannot strictly gain whenever it is shocked.}  A non-deviating shocked peripheral can receive at most $\ell-\theta y$ from the core; even if it receives exactly $\ell-\theta y$, it only covers its liquidity shortfall and has no resources left to repay the core. Hence, if the core is solvent, its residual pledgeable assets are at most
    \[
        \theta y+(F_{1\to2}-(\ell-\theta y))^+ -F_{1\to2}-(N-2)(\ell-\theta y).
    \]
    If $F_{1\to2}\geq \ell-\theta y$, this upper bound is $\theta y-(N-1)(\ell-\theta y)<0$. If $F_{1\to2}<\ell-\theta y$, it is $\theta y-F_{1\to2}-(N-2)(\ell-\theta y)$. Thus the core can be strictly better than $(1-\theta)y$ in this state only if
    \[
        F_{1\to2}<\theta y-(N-2)(\ell-\theta y).
    \]
    This inequality implies $F_{1\to2}<\ell-\theta y$ because $\ell-\theta y>\theta y/(N-1)$.

    Second, suppose B$_2$ is the unique not-shocked peripheral. If $F_{2\to1}<\ell-\theta y$, then the core is liquidated and loses its baseline payoff $(1-\theta)y$. And it is no less than its possible gain $\theta y-(N-2)(\ell-\theta y)$. Hence, a deviation with $F_{2\to1}<\ell-\theta y$ cannot strictly raise the core's expected payoff given uniform distribution. Suppose instead $F_{2\to1}\geq\ell-\theta y$. If the core is liquidated in this state, then B$_2$ receives no payment from the core; B$_2$ either repays at least $\ell-\theta y$ and obtains a payoff at most $y-(\ell-\theta y)$, or defaults and obtains $(1-\theta)y<y-(\ell-\theta y)$. If the core is not liquidated, let $R_{1\to2}$ and $R_{2\to1}$ be the realized payments between B$_1$ and B$_2$. Since B$_1$ is shocked, any payment it makes to interbank creditors must come from resources left after covering the liquidity shortfall $\ell-\theta y$. The other shocked peripheral banks have no resources to repay the core. Therefore,
    \(
        R_{1\to2}\leq R_{2\to1}-(\ell-\theta y).
    \)
    If B$_2$ is solvent, its payoff is
    \[
        y+R_{1\to2}-R_{2\to1}\leq y-(\ell-\theta y).
    \]
    If B$_2$ defaults, it obtains $(1-\theta)y<y-(\ell-\theta y)$. Thus, B$_2$ cannot strictly gain in the state in which it is the unique not-shocked peripheral when $F_{2\to1}\geq \ell-\theta y$.

    Combining the two points above, any deviation that strictly benefits B$_1$ leaves B$_2$ weakly worse off in every shock state. Hence no bilateral deviation can strictly benefit both banks in expectation.

    Finally, we show B$_1$ can avoid default. Take the deviation $F_{1\to2}=0$ and $F_{2\to1}=\ell-\theta y$. In the state in which all peripheral banks are shocked, the core owes only $\ell-\theta y$ to each of B$_3,\ldots,$ B$_N$, repays all these liabilities in full, has residual pledgeable assets $\theta y-(N-2)(\ell-\theta y)>0$, and obtains more than $(1-\theta)y$. B$_2$ is weakly indifferent across states. Thus pairwise stability does not come from the core always defaulting after deviations.

\subsection{Proof of Proposition~\ref{p:complete stable large shock} (\nameref{p:complete stable large shock})}

Let
\[
    \bar F=\frac{N}{N-2}\theta y,
\]
and fix \(F\geq\bar F\), a shock distribution \(G\), and a bilateral deviation between B$_1$ and B$_2$. In any shock state with \(S=|\bm\sigma|\in\{1,\ldots,N-1\}\), the condition \(\ell-\theta y>(N-1)\theta y\) implies \(L^*=S\): even all \(N-S\) not-shocked banks together have at most \((N-S)\theta y\leq(N-1)\theta y<\ell-\theta y\), which is not enough to save one shocked bank. Hence every shocked bank is liquidated before and after the deviation.

Relabeling the two deviating banks if necessary, suppose \(F_{1\to2}\geq F_{2\to1}\) after the deviation.

In the original uniform complete network, every not-shocked bank defaults in such a state. Indeed, its incoming payment is at most \((N-S-1)F\), while its liabilities are \((N-1)F\), and
\(
    (N-1)F>\theta y+(N-S-1)F
\)
because \(SF\geq F\geq\bar F>\theta y\). Thus every not-shocked bank's net payment is \(\theta y\), and every shocked bank's payoff is zero.

Now suppose exactly one deviating bank, say B$_1$, is shocked. Since B$_1$ is liquidated, B$_1$ cannot strictly gain. We show that B$_2$ cannot strictly gain either. The deviation most favorable to B$_2$ in this state sets B$_2$'s liability to B$_1$ equal to zero; B$_1$'s liability to B$_2$ is irrelevant because B$_1$ is liquidated. Let \(U\) be the set of not-shocked banks outside the deviating pair, and let \(m=|U|=N-S-1\). If \(m>0\), let \(\bar R\) be the largest total repayment among banks in \(U\). A bank in \(U\) receives at most \(F\) from B$_2$ and at most \((m-1)\bar R/(N-1)\) from the other banks in \(U\). Hence
\[
    \bar R\leq \theta y+F+\frac{m-1}{N-1}\bar R,
    \qquad\text{so}\qquad
    \frac{\bar R}{N-1}\leq \frac{\theta y+F}{S+1}.
\]
Thus B$_2$'s incoming payment from \(U\) is at most \(m(\theta y+F)/(S+1)\). If B$_2$ defaults, its net payment is \(\theta y\), as in the original network. If B$_2$ is solvent, its net payment is at least
\[
    (N-2)F-\frac{N-S-1}{S+1}(\theta y+F).
\]
This lower bound is at least \(\theta y\) because
\[
    (N-2)F-\frac{N-S-1}{S+1}(\theta y+F)-\theta y
    =
    \frac{[S(N-1)-1]F-N\theta y}{S+1}
    \geq0,
\]
where the last inequality follows from \(F\geq\bar F=N\theta y/(N-2)\), since \(S(N-1)-1\geq N-2\). Hence B$_2$ cannot strictly gain. The same argument applies when B$_2$ is shocked and B$_1$ is not.

It remains to consider states in which both deviating banks have the same shock status. If both are shocked and at least one bank is not shocked, the argument above shows that both are liquidated. If all banks are shocked, aggregate payment clearing likewise rules out any non-liquidated bank, since such a bank would require a positive net inflow \(\ell-\theta y\) but there is no source of positive net outflow. Thus B$_1$ cannot gain when both deviating banks are shocked. Suppose instead that both are not shocked and that \(S\) banks outside the pair are shocked. If B$_1$ defaults after the deviation, its payoff is \((1-\theta)y\), equal to its original payoff when \(S\geq1\) and below its original payoff \(y\) when \(S=0\). If B$_1$ is solvent, it receives at most \(F_{2\to1}\) from B$_2$ and at most \(F\) from each of the \(N-S-2\) not-shocked banks outside the pair. Since its total liabilities are \((N-2)F+F_{1\to2}\), its net payment is at least
\[
    (N-2)F+F_{1\to2}-\big[(N-S-2)F+F_{2\to1}\big]
    =SF+F_{1\to2}-F_{2\to1}.
\]
For \(S\geq1\), this is at least \(F>\theta y\), so B$_1$'s payoff is no greater than its original payoff \((1-\theta)y\). For \(S=0\), it is at least \(F_{1\to2}-F_{2\to1}\geq0\), so B$_1$'s payoff is no greater than its original payoff \(y\). Thus B$_1$ is weakly worse off in every shock state. Therefore the deviation cannot strictly increase both banks' expected payoffs under any \(G\).

\subsection{Proof of Proposition \ref{p:star inefficient} (\nameref{p:star inefficient})}

Since the system doesn't have enough liquidity to save all shocked banks, when the core bank is not shocked, all shocked banks are peripheral, and they are all liquidated by symmetry. When the core bank is shocked, it is not possible for both the core and all peripheral banks to avoid liquidation. The core can avoid liquidation if it receives enough liquidity from the $(N-S)$ not-shocked peripheral banks, but given the low interbank debt $F<(\ell-\theta y)/(N-S)$, the maximum liquidity the core bank can receive is $(N-S)F<\ell-\theta y$. So it cannot avoid liquidation. Both scenarios are not constrained efficient when $L^*<S$ by definition.

\subsection{Proof of Proposition \ref{p:complete inefficient} (\nameref{p:complete inefficient})}
Since there is not enough liquidity to save all shocked banks given $S\ell>N\theta y$, all shocked banks are liquidated by symmetry and equilibrium uniqueness (Proposition~\ref{p:existence}).

\subsection{Proof of Proposition \ref{p:PP=KP} (\nameref{p:PP=KP})}
We prove each implication in turn.

    \textbf{(i) $\hat{\textbf{t}}$ solves PP $\implies$ $\hat{\textbf{x}}$ solves KP.} We show that $\hat{\textbf{x}}$ is feasible and optimal in turn.

        \begin{itemize}
            \item $\hat{\textbf{x}}$ \emph{is feasible.}
            Observe that, since,  by equation \eqref{eq:pp:liquidity_constraint}, $\hat t_i \geq -\theta y_i$,  $\hat x_i\sigma_i(\ell_i-\theta y_i) \leq (1-\sigma_i)\theta y_i+ \hat t_i$. Thus, using the liquidity conservation constraint (equation~\eqref{eq:pp:liquidity_conservation}),  we have that
                \begin{align}
                    \sum_{i=1}^N \hat x_i \sigma_i(\ell_i-\theta y_i)\leq \sum _{i=1}^N \Big(  (1-\sigma_i)\theta y_i+ \hat t_i \Big) \leq  \sum _{i=1}^N (1-\sigma_i)\theta y_i,
                \end{align}
        confirming that $\hat {\textbf{x}}$ is feasible.

            \item  $\hat{\textbf{x}}$\emph{ is optimal.}
            Suppose, in anticipation of a contradiction, that $\hat{\textbf{x}}$ is not optimal, i.e.\ that there is a feasible $\hat{\textbf{x}}'$ that yields lower deadweight loss (equation \eqref{e:x DWL}).

             We now show that such an $\hat{\textbf{x}}'$ cannot exist, because, if it does, $\hat{\textbf{t}}$  cannot be a solution to PP. Specifically, $\hat{\textbf{t}}\!\!{\phantom{t}}'$ with $\hat t_i' := \sigma_i \hat x_i' (\ell_i - \theta y_i )  - (1 - \sigma_i) \theta y_i$ is feasible and yields a lower objective:

    \begin{itemize}

        \item {Feasibility:}

            \begin{itemize}

                \item  We show that $\hat{\textbf{t}}\!\!{\phantom{t}}'$ satisfies the liquidity constraint \eqref{eq:pp:liquidity_constraint} for $\sigma_i = 0$  and $\sigma_i = 1$ in turn: If $\sigma_i = 0$, then  $\hat t_i' =-\theta y_i =\min\{\sigma_i\ell_i-\theta y_i,0\}$ and if $\sigma_i =1$, then $
                    \hat t_i' = \hat x_i' (\ell_i-\theta y_i) \geq  \min\{\sigma_i\ell_i-\theta y_i,0\}.$

                \item We show that  $\hat{\textbf{t}}\!\!{\phantom{t}}'$  satisfies the liquidity conservation constraint \eqref{eq:pp:liquidity_conservation} as
                                \begin{align}
                            \sum_{i=1}^N \hat t_i' =
             \sum_{i=1}^N \Big(  \sigma_i \hat x_i' (\ell_i-\theta y_i)-(1-\sigma_i)\theta y_i \Big) \leq 0,
                        \end{align}
            since $\hat{\textbf{x}}'$ must satisfy the constraint \eqref{eq:aggregate_pledge} by its definition as a solution to KP.

            \end{itemize}

        \item {Optimality:}          Observe, from the definitions of $\hat x_i'$ and $\hat t_i$ that
            \begin{equation}
            \sum_{i=1}^N \mathbbm{1}_{\{\theta y_i+ \hat t_i' < \ell_i\sigma_i\}}  \Delta_i
            \leq
            \sum_{i=1}^N ( 1- \hat x_i' ) \Delta_i
            <
            \sum_{i =1}^N ( 1 - \hat x_i ) \Delta_i
            = \sum_{i=1}^N \mathbbm{1}_{\{\theta y_i+ \hat t_i < \ell_i\sigma_i\}}  \Delta_i,
            \end{equation}
            contradicting the optimality of $\hat{\textbf{t}}\!\!{\phantom{t}}$.

    \end{itemize}
         \noindent Therefore $\hat {\textbf{x}}$ solves KP.

 \textbf{(ii) $\check{\textbf{x}}$ solves KP $\implies$ $\check{\textbf{t}}$ solves PP.} We show that $\check{\textbf{t}}$ is feasible and optimal in turn.

        \begin{itemize}
            \item $\check{\textbf{t}}$\emph{ is feasible.}
                          Observe that $\check t_i$ satisfies B$_i$'s liquidity  constraint \eqref{eq:pp:liquidity_constraint} by construction:    $\check t_i = \sigma_i \check x_i (\ell_i-\theta y_i)-(1-\sigma_i)\theta y_i\geq \min\{\sigma_i\ell_i-\theta y_i,0\}.$

            \item  $\check{\textbf{t}}$\emph{ is optimal.}
            Suppose, in anticipation of a contradiction, that $\check{\textbf{t}}$ is not optimal, i.e.\ that there is a feasible $\check{\textbf{t}}'$ that yields lower deadweight loss (equation \eqref{e:t DWL}).

             We now show that such an $\check{\textbf{t}}'$ cannot exist because, if it does, $\check{\textbf{x}}$ cannot be a solution to KP. Specifically, $\check{\textbf{x}}'$ with $\check x_i' := \mathbbm{1}_{\{\theta y_i+ \check t_i' -\ell_i\sigma_i\geq 0\}} $ is feasible and yields a lower objective:

    \begin{itemize}

        \item {Feasibility:}  We show that $\check x_i'$ satisfies the liquidity conservation constraint~\eqref{eq:aggregate_pledge} as
                    \begin{equation}
                        \sum_{i=1}^N \check x_i' \sigma_i(\ell_i-\theta y_i)\leq \sum_{i=1}^N \Big( (1-\sigma_i)\theta y_i + \check t_i' \Big) \leq \sum_{i=1}^N (1-\sigma_i)\theta y_i,
                    \end{equation}
                    since $\check{\textbf{t}}'$ must satisfy the constraint \eqref{eq:pp:liquidity_conservation} by its definition as a solution to PP.

        \item {Optimality:}      Observe, from the definitions of $\check x_i'$ and $\check t_i$ that
        \begin{equation}
    \hspace{-.8cm}    \sum_{i=1}^N
        (1 - \check x_i') \Delta_i
        \leq
        \sum_{i=1}^N
    \mathbbm{1}_{\{\theta y_i + \check t_i' < \ell_i \sigma_i\}} \Delta_i
    <
        \sum_{i=1}^N    \mathbbm{1}_{\{ \theta y_i + \check t_i < \ell_i \sigma_i\}} \Delta_i
    \leq     \sum_{i=1}^N
        (1 - \check x_i) \Delta_i ,
        \end{equation}
        contradicting the optimality of $\check{\textbf{x}}.$

    \end{itemize}

              \noindent Therefore $\check {\textbf{t}}$ solves PP. \qed

    \end{itemize}

    \end{itemize}

\subsection{Proof of Proposition \ref{p:exponential_implement_greedy} (\nameref{p:exponential_implement_greedy})}
 First note that Lemma~\ref{l:high debt}, Lemma~\ref{lm:exponentially_dominated}, and Lemma~\ref{lm:total_liquidity} hold with heterogeneous banks. Their proofs are essentially unchanged, with $y$ and $\ell$ replaced with $y_i$ and $\ell_i$ everywhere, except in the last line of the proof of Lemma~\ref{lm:total_liquidity}, when they are replaced by $y_{i^*}$ and $\ell_{i^*}$. We apply these results freely throughout the proof.

Now we prove the result in three steps.
    \begin{itemize}
        \item [Step 1:] \textbf{For all $i$ such that $\sigma_i=0$, $x_i=1.$} This is immediate from equation \eqref{eq:optimal_repayment_cases}, which implies that the bank is not liquidated if it is not shocked.

        \item [Step 2:] \textbf{Existence of critical index.}
    Now show that there exists a critical index $i^*$ such that a shocked bank B$_i$ is liquidated if and only if $i \geq i^*$.
    Suppose, in anticipation of a contradiction, that, to the contrary, there are shocked B$_i$ and B$_j$ with $i<j$ such that $x_i=0$ and $x_j=1$.
    Then, by equation \eqref{eq:optimal_repayment_cases}, it must be that $R_{i\leftleftarrows}<\ell_i - \theta y_i$ and $R_{j\leftleftarrows}\geq \ell_j - \theta y_j$ and, therefore, by  Lemma~\ref{lm:exponentially_dominated}, that
            \begin{equation}
            \ell_j - \theta y_j \leq R_{j\leftleftarrows}\leq s^{j-i}R_{i\leftleftarrows} < s^{j-i}( \ell_i -\theta y_i)  .
            \end{equation}
            The inequality is violated if $s< \sqrt[j-i]{(\ell_j - \theta y_j)/(\ell_i - \theta y_i)}$, a contradiction.

        \item [Step 3:] \textbf{Same critical index.}  We now show that the critical index delivered by the exponential network coincides with that delivered by the greedy algorithm (for every state $\bm \sigma$) or, equivalently, that $i^*$ satisfies the following two inequalities:
            \begin{equation}
                \sum _{i=1}^{i^*}\sigma_i(\ell_i-\theta y_i) > \sum_{i=1}^N (1-\sigma_i) \theta y_i
            \end{equation}
            and
            \begin{equation}
                \sum _{i=1}^{i^*-1}\sigma_i(\ell_i-\theta y_i) \leq  \sum_{i=1}^N(1-\sigma_i)\theta y_i  .
            \end{equation}
            The second is immediate, as it is implied by the aggregate liquidity constraint, which holds strictly by hypothesis.
            To prove the first,  we invoke Lemma~\ref{l:high debt}, which implies that, as long as debts are sufficiently high, each not-shocked bank pays $\theta y_i$ in net payment, so conservation of liquidity requires:
                \begin{align}
                \sum_{i=1}^N (1-\sigma_i)\theta y_i
                & = -\sum_{i \, :\, \sigma_i=0}(R_{i\leftleftarrows}-R_{i\rightrightarrows})
                \\
                & = \sum_{i\, : \, \sigma_i=1} (R_{i\leftleftarrows}-R_{i\rightrightarrows})
                \\
                & = \sum_{i\in\mathscr{L}^c \, :\, \sigma_i=1} (R_{i\leftleftarrows}-R_{i\rightrightarrows}) +\sum_{i\in\mathscr{L}} (R_{i\leftleftarrows}-R_{i\rightrightarrows})
                 \\
                 &\leq \sum_{i\in\mathscr{L}^c \, :\, \sigma_i=1} (R_{i\leftleftarrows}-R_{i\rightrightarrows})
                 + \sum_{i\in\mathscr{L}}R_{i\leftleftarrows}
                 \\
                 &< \sum_{i\in\mathscr{L}^c \, :\, \sigma_i=1} (R_{i\leftleftarrows}-R_{i\rightrightarrows})
                 + \frac{R_{i^* \leftleftarrows}}{1-s}
                 \\
                 &\leq  \sum_{i\in\mathscr{L}^c \, :\, \sigma_i=1} (\ell_i - \theta y_i)
                 + \frac{\ell_{i^*} - \theta y_{i^*}}{1-s}
                 \\
                 &= \sum_{i=1}^{i^*}\sigma_i(\ell_i-\theta y_i)
                 + \frac{s}{1-s} (\ell_{i^*} - \theta y_{i^*})
                \end{align}
                having used  Lemma~\ref{lm:total_liquidity} to bound the sum over liquidated banks.
                Letting $s$ be sufficiently small and recalling that the inequality must be strict by assumption gives the result.
                \qed

    \end{itemize}

\subsection{Proof of Corollary~\ref{c:optimality of greedy} (\nameref{c:optimality of greedy})}

The result is immediate from the definitions of equivalence of the planner's problem to the knapsack problem and of the greedy algorithm to the exponential network; see Proposition \ref{p:PP=KP}, Definition \ref{d:kp}, and Proposition \ref{p:exponential_implement_greedy}.  (See also the discussion following the statement of the corollary.) \qed

\subsection{Proof of Corollary~\ref{c:greedy small} (\nameref{c:greedy small})}

The result follows immediately from the so-called Dantzig bound (\cite{dantzig1957discrete}), which we state as a lemma:

\begin{lemma} \label{l:dantzig} Suppose, w.l.o.g., that banks are ordered by their profitability indices ($\text{PI}_i \geq \text{PI}_j$ for $i \leq j$) and let  $\check{\bm{\mathrm{x}}}$ be a solution of KP (Definition \ref{d:kp}).   We have that
    \begin{equation}  \label{e:dantzig}
        \sum_{i=1}^N   \sigma_i \check x_i \Delta_i \leq \sum_{i = 1}^{i^*- 1} \sigma_i \Delta_i  + \Delta_{i^*} .
    \end{equation}

\end{lemma}

\begin{proof} See \cite{Martello-Toth-1990}, Theorem 2.1. \end{proof}

\noindent The result implies that the difference between the objective at the optimum (represented by the LHS of equation \eqref{e:dantzig}) and at the greedy algorithm's approximation of it (represented by the sum on the RHS) is at most $\Delta_{i^*}$, as desired. \qed

\subsection{Proof of Proposition~\ref{p:endogenous theta}
(\nameref{p:endogenous theta})}

We prove the result in three steps: We first characterize which banks are liquidated as a function of the profile $(\theta_1,\theta_2)$, which reduces the planner's problem to a comparison of four candidate profiles; we then compare those profiles; and we finally check that no bank wants to deviate from any of them. Throughout, we label banks w.l.o.g.\ so that the concentrated profile is $(0,\theta^a)$ rather than its permutation.

\textbf{Liquidation and candidate profiles.} Suppose exactly one bank, say B$_1$, is shocked. Since $F\geq\ell$, it can raise $\theta_1y+\min\{F,\theta_2y\}$, so it survives if and only if $\theta_1+\theta_2\geq\theta^a$. Suppose instead both banks are shocked. A bank with $\theta_i\geq\theta^a$ meets its shock unaided and repays $\theta_iy-\ell$; a bank with $\theta_i<\theta^a$ receives nothing from a liquidated counterparty and so, by equation~\eqref{eq:optimal_repayment_cases}, is itself liquidated unless its counterparty survives and repays enough. Hence both banks survive if and only if $\theta_1+\theta_2\geq2\theta^a$, and otherwise only a bank with $\theta_i\geq\theta^a$ does.

 Up to terms independent of pledgeability, the
planner maximizes
\begin{equation}
-(y-\ell)\,\mathbb{E}\big[\#\text{liquidated banks}\big]
-c(\theta_1)-c(\theta_2).
\end{equation}
By the previous step, profiles fall into four classes by liquidation outcome,
and within each class the planner takes the cheapest. If
$\theta_1+\theta_2<\theta^a$, no bank is ever saved, and $(0,0)$ is cheapest.
If $\theta_1+\theta_2\geq\theta^a$ with $\theta_i<\theta^a$ for both banks,
shocked banks are saved only in single-shock states, and strict convexity
makes the equal split $(\theta^a/2,\theta^a/2)$ uniquely cheapest. If
$\theta_i\geq\theta^a$ for one bank and $\theta_1+\theta_2<2\theta^a$, one
bank is also saved in the double-shock state, and $(0,\theta^a)$ is uniquely
cheapest. If $\theta_1+\theta_2\geq2\theta^a$, both banks are always saved,
and strict convexity again makes $(\theta^a,\theta^a)$ uniquely cheapest,
ruling out $(0,2\theta^a)$ in particular. It is therefore enough to compare
these four profiles.

\textbf{Social optimum.} Taking $(0,0)$ as the baseline, we have the following expression for the incremental
surpluses:
\begin{equation}
\text{incremental surplus}
=
\renewcommand\arraystretch{1.5}
\left\{
\begin{array}{cl}
2p(1-p)(y-\ell)-2c(\theta^a/2) & \text{ for } (\theta^a/2,\theta^a/2),\\
\big[2p(1-p)+p^2\big](y-\ell)-c(\theta^a) & \text{ for } (0,\theta^a)\\
2p(y-\ell)-2c(\theta^a) &\text{ for } (\theta^a,\theta^a)
\end{array}
\right.
\end{equation}
The maintained condition $c(\theta^a)\leq p(y-\ell)$ makes the second at least $p(1-p)(y-\ell)>0$, so $(0,0)$ is not socially optimal. Pairwise comparison of the other three gives the following  classification:

\begin{enumerate}

\item[(i)] $(\theta^a,\theta^a)$ is weakly socially optimal if and only if
$c(\theta^a)\leq p^2(y-\ell)$.

\item[(ii)] $(0,\theta^a)$ is weakly socially optimal if and only if $c(\theta^a)\geq p^2(y-\ell)$ and $c(\theta^a)-2c(\theta^a/2)\leq p^2(y-\ell)$.

\item[(iii)] $(\theta^a/2,\theta^a/2)$ is weakly socially optimal if and only
if $c(\theta^a)-2c(\theta^a/2)\geq p^2(y-\ell)$.
\end{enumerate}
Each displayed inequality is the comparison with $(0,\theta^a)$; the remaining comparison, between $(\theta^a,\theta^a)$ and $(\theta^a/2,\theta^a/2)$, follows from it in each case, since $c(\theta^a/2) \geq 0$ implies $c(\theta^a)-c(\theta^a/2)\leq c(\theta^a)$ in (i) and $c(\theta^a)-c(\theta^a/2)\geq c(\theta^a)-2c(\theta^a/2)$ in (iii). These cases exhaust the parameters satisfying the maintained condition.

\textbf{Best response/equilibrium.} We check each social optimum in turn.

\begin{enumerate}

\item[(i)] At $(\theta^a,\theta^a)$, suppose B$_i$ cuts to $\theta<\theta^a$. The pair still holds at least $\theta^a$, so B$_i$ keeps its single-shock protection, but the pair holds less than $2\theta^a$ and $\theta<\theta^a$, so B$_i$ is liquidated when both banks are shocked, losing $p^2(y-\ell)$. The saving $c(\theta^a)-c(\theta)$ is largest at $\theta=0$, and condition (i) rules that out. Raising pledgeability gains B$_i$ nothing,  as it is already saved in every state.

\item[(ii)] At $(0,\theta^a)$, suppose B$_1$ raises to $\theta$. If $\theta<\theta^a$ it remains liquidated when both banks are shocked, so it gains nothing; if $\theta\geq\theta^a$ it gains $p^2(y-\ell)$ at a cost of at least $c(\theta^a)$, which condition (ii) rules out. Suppose instead B$_2$ cuts to $\theta<\theta^a$. The pair then holds less than $\theta^a$, so B$_2$ is liquidated whenever it is shocked, losing $p(y-\ell)$ and saving at most $c(\theta^a)$; the maintained condition rules that out. Raising its pledgeability further gains B$_2$ nothing, as it is already saved whenever it is shocked.

\item[(iii)] At $(\theta^a/2,\theta^a/2)$, suppose B$_i$ cuts to $\theta<\theta^a/2$. The pair then holds less than $\theta^a$, so B$_i$ is liquidated in the state in which it alone is shocked, losing $p(1-p)(y-\ell)$ and saving at most $c(\theta^a/2)$. Condition (iii) and the maintained condition give
\begin{equation}
c(\theta^a/2)\leq\frac{c(\theta^a)-p^2(y-\ell)}{2}
\leq\frac{p(1-p)(y-\ell)}{2},
\end{equation}
so the deviation is unprofitable. Suppose instead B$_i$ raises its pledgeability to $\theta$. It gains only if it survives when both banks are shocked, which requires $\theta\geq\theta^a$ and hence costs
\begin{equation}
c(\theta)-c(\theta^a/2)\geq c(\theta^a)-c(\theta^a/2)
\geq c(\theta^a)-2c(\theta^a/2)\geq p^2(y-\ell),
\end{equation}
which is the gain.

\end{enumerate}

Thus every weakly socially optimal profile is a Nash equilibrium. \qed

\section{Ring Example} \label{a:ring}

Here we show that the ring network is not constrained efficient with four banks. We keep $y = 2$ and $\theta = 1/2$ but increase the shock to $\ell = 12/5$, so that no single bank's liquidity covers a shortfall, $\theta y = 1 < \ell - \theta y = 7/5$, although pooling two banks' liquidity does, $\ell - \theta y \leq 2 \theta y$. With $S = 2$ banks shocked, exactly one bank is liquidated in the constrained-efficient outcome, $L^* = 1$ (Lemma \ref{l:constrained efficiency}).

In the ring, each B$_i$ owes $F$ to B$_{i+1}$ (indices modulo four) and nothing to the others:
\begin{equation}
\renewcommand\arraystretch{.75}
\mathbf{F} = F
\left[
\begin{array}{cccc}
0 & 1 & 0 & 0 \\
0 & 0 & 1 & 0 \\
0 & 0 & 0 & 1 \\
1 & 0 & 0 & 0
\end{array}
\right].
\end{equation}
Whether the ring is efficient depends on where the shocks fall.

Suppose first that the shocked banks are adjacent, say B$_1$ and B$_2$. The clearing vector solves
\begin{equation}
\renewcommand\arraystretch{1.5}
\left\{ \begin{array}{l}
R_{1\rightrightarrows} =
\max\Big\{ \, 0 \, , \, \min\left\{ \, - \frac 7 5 + R_{4 \rightrightarrows}
\, , \,
F
\, \right\} \Big\},
\\
R_{2\rightrightarrows} =
\max\Big\{ \, 0 \, , \, \min\left\{ \, - \frac 7 5 + R_{1 \rightrightarrows}
\, , \,
F
\, \right\} \Big\},
\\
R_{3\rightrightarrows} =
\max\Big\{ \, 0 \, , \, \min\left\{ \, 1 + R_{2 \rightrightarrows}
\, , \,
F
\, \right\} \Big\},
\\
R_{4\rightrightarrows} =
\max\Big\{ \, 0 \, , \, \min\left\{ \, 1 + R_{3 \rightrightarrows}
\, , \,
F
\, \right\} \Big\},
\end{array}
\right.
\end{equation}
with solution
$\mathbf{R}_{\rightrightarrows} = \left( \frac{3}{5}, 0, 1, 2 \right)$
for $F \geq 2$. The not-shocked banks' payments chain through the ring---B$_3$ pays B$_4$, allowing B$_4$ to pay $2 \theta y = 2$ to B$_1$---so B$_1$ survives and only B$_2$ is liquidated: The ring is constrained efficient in this state.

Now suppose the shocked banks alternate with the not-shocked ones, say B$_1$ and B$_3$. The clearing vector solves
\begin{equation}
\renewcommand\arraystretch{1.5}
\left\{ \begin{array}{l}
R_{1\rightrightarrows} =
\max\Big\{ \, 0 \, , \, \min\left\{ \, - \frac 7 5 + R_{4 \rightrightarrows}
\, , \,
F
\, \right\} \Big\},
\\
R_{2\rightrightarrows} =
\max\Big\{ \, 0 \, , \, \min\left\{ \, 1 + R_{1 \rightrightarrows}
\, , \,
F
\, \right\} \Big\},
\\
R_{3\rightrightarrows} =
\max\Big\{ \, 0 \, , \, \min\left\{ \, - \frac 7 5 + R_{2 \rightrightarrows}
\, , \,
F
\, \right\} \Big\},
\\
R_{4\rightrightarrows} =
\max\Big\{ \, 0 \, , \, \min\left\{ \, 1 + R_{3 \rightrightarrows}
\, , \,
F
\, \right\} \Big\},
\end{array}
\right.
\end{equation}
with solution
$\mathbf{R}_{\rightrightarrows} = \left( 0, 1, 0, 1 \right)$.
Each not-shocked bank pays a different shocked bank, so each shocked bank receives only $\theta y = 1 < \frac 7 5$ and both are liquidated. The planner instead pools the two not-shocked banks' liquidity, ($2 \theta y=2$), and transfers enough to one shocked bank to save it. The ring is thus efficient for some states ($\bm{\sigma}$) but not for others. Its problem is that the direction of liquidity depends on where the shocked and not-shocked banks happen to lie in the ring: When they alternate, liquidity is divided between the two shocked banks rather than concentrated on one of them.

\section{Notations}

To the extent possible, we use bold face letters for matrices and vectors and use italics for scalars;  we use single-arrow subscripts for liabilities from one bank to another and double-arrow subscripts for total liabilities from one to many banks. E.g., $\mathbf{F} = [F_{i \to j}]_{ij}$ is the matrix of interbank liabilities between individual banks; $\mathbf{F}_{\rightrightarrows} = [F_{i \rightrightarrows}]_i$ is the vector of banks' total interbank liabilities, i.e.\ the vector of row sums of $\mathbf{F}$. We use B$_i$ for individual banks and script letters for sets; $\text{B}_i \in \mathscr{B}$  and $i \in \mathscr{B}$ are synonymous. We summarize our notations in Table~\ref{t:notations}, separating those used in the main text from those used only in extensions or proofs.

\footnotesize \setlength{\LTcapwidth}{\textwidth}
\begin{longtable}{rll}
\caption{Notations. \label{t:notations}}  \\
\toprule
Notation & Meaning & Parametric restriction\\
\midrule
\endfirsthead
Notations (continued)\\
\midrule
Notation & Meaning & Parametric restriction \\
\midrule
\endhead

\bottomrule{{Continued on next page}} \\
\endfoot

\bottomrule
\endlastfoot

        $y$ & Long-term real asset value & $y > 0$\\
        $\ell$ & Size of liquidity shock & $\theta y<\ell<y$\\
        $\theta$ & Pledgeable fraction of $y$ &  $0 < \theta < 1$\\
        $\sigma_i$& Indicator of B$_i$'s shock & $\sigma_i\in\{0,1\}$\\
        $\bm{\sigma}\equiv \{\sigma_i\}_i$ &Vector of shocks/Aggregate state  &$\bm{\sigma}\in \{0,1\}^N$\\
        $\mathrm{B}_{i}$ &  $i$th bank & \\
        $\mathscr{B}$ &  A set of banks & \\
        $\mathscr{B}^c$ & Complement of $\mathscr{B}$& \\
        $\mathscr{L}$ & Set of banks that are liquidated  & $\mathscr{L}\subset \mathscr{D}$\\
        $N$ & Number of banks & \\
        $S = \sum \sigma_i $ & Number of shocked banks
        &  \\
        $L = \lvert \mathscr{L} \rvert$ & Number of liquidated banks & \\
        $L^*$ & Minimum $L$ (Lemma \ref{l:constrained efficiency}) & \\
        $F_{i\to j}$& $\mathrm{B}_i$'s liability to $\mathrm{B}_j$& $F_{i\to j} \geq 0$\\
        $\textbf{F} \equiv [F_{i\to j}]_{ij}$ & Matrix of interbank debts&\\
        $F_{i\rightrightarrows}\equiv \sum_{j\neq i}F_{i\to j}$ & $\mathrm{B}_i$'s total  interbank liabilities &\\
        $F$ & Each bank's total liabilities $F \equiv F_{i \rightrightarrows} $  in a regular network; &\\
         & per-link face value in the star and uniform-complete sections &\\
        $F_{i\leftleftarrows}\equiv \sum_{j\neq i}F_{j\to i}$ & $\mathrm{B}_i$'s total interbank claims &\\
        $\mathbf{F}_{\rightrightarrows}\equiv \{F_{i\rightrightarrows}\}_i $& Vector of each bank's total interbank liabilities&\\
        $\hat{F}_{i\to j}\equiv F_{i\to j}/F_{i\rightrightarrows}$ & $\mathrm{B}_i$'s liability to $\mathrm{B}_j$ as a fraction of its total liabilities& $0\leq\hat{F}_{i\to j}\leq 1$\\
        $\mathbf{\hat{F}}\equiv[\hat{F}_{i\to j}]_{ij}$& Matrix of  interbank debts & $\sum_i \hat F_{i \to j} = 1$ if $F_{i\rightrightarrows} > 0$\\
        $R_{i\to j}$& $\mathrm{B}_i$'s equilibrium repayment to $\mathrm{B}_j$& $0\leq R_{i\to j} \leq F_{i\to j}$\\
        $R_{i\rightrightarrows}\equiv\sum_{j\neq i}R_{i\to j}$ & $\mathrm{B}_i$'s total repayment to other banks &\\
        $R_{i\leftleftarrows}\equiv\sum_{j\neq i}R_{j\to i}$ & $\mathrm{B}_i$'s total repayment received from other banks &\\
        $\mathbf{R}_{\rightrightarrows}\equiv\{R_{i\rightrightarrows}\}_i $& Vector of each bank's total equilibrium repayment& $\mathbf{0}\leq \mathbf{R}_{\rightrightarrows}\leq \mathbf{F}_{\rightrightarrows}$\\
        $\mathbf{R}_{\leftleftarrows}\equiv\{R_{i\leftleftarrows}\}_i $& Vector of each bank's total payment received & $\mathbf{R}_{\leftleftarrows}=\mathbf{\hat{F}}^{\top}\mathbf{R}_{\rightrightarrows}$\\
        $\alpha$& Scale of debts used in, e.g., Proposition~\ref{p:netting} & $\alpha>0$\\
        $\beta$ & Bottleneck parameter (Definition \ref{d:bottleneck})& \\
        $d_{i\to j}$ & Harmonic distance from $\mathrm{B}_i$ to $\mathrm{B}_j$ (Definition~\ref{d:HD})& $d_{i\to j}\geq 0$\\
        $\delta$ & Connectedness parameter (Definition \ref{d:delta})& $0<\delta<1$\\
        $d^{ST}, d^{LT}$& Default and salvation radii in Lemma~\ref{l:ST radius} and Proposition~\ref{p:radius}&\\
        $\beta^{ST},\beta_{ST},\beta^{LT},\beta_{LT}$ & Thresholds in Lemma~\ref{l:ST bottleneck} and Proposition~\ref{p:bottleneck}&\\
            $s$ & Dominance parameter (Definition \ref{d:s-dominance})& $0<s<1$\\
        $s^*$ &Threshold in Proposition \ref{p:constrained efficiciency}&\\
        $t_i$&Transfer to $\mathrm{B}_i$ in Definition \ref{d:constrained efficiency}&\\
        $i^*$ &Index of highest-ranked liquidated bank in Section~\ref{s:exp} and Section~\ref{s:hetero}&\\
        $\pi_i$ &Permutation of banks keeping $\mathrm{B}_i$ fixed  &\\
        $[\cdot]^+=\max\{\cdot,0\}$& Maximum of variable and zero&\\
        $\lceil\cdot\rceil,\lfloor\cdot\rfloor$ &Ceiling and floor functions&\\
        &&\\
            \hline
        \multicolumn{3}{c}{Notations Used Only in Extensions}
        \\ \hline
                $y^*$ & Efficient liquidation threshold in Section~\ref{s:risky assets}&\\
                $\Delta_i$ & Efficiency loss if $\mathrm{B}_i$ is liquidated in Section~\ref{s:hetero}&\\
                $\textbf{x}\equiv \{x_i\}_i$ & Vector of indicators of banks not being liquidated in Section~\ref{s:hetero} & $\textbf{x} \in \{0,1\}^N$\\
        $\hat{\textbf{t}},\check{\textbf{t}},\hat{\textbf{x}},\check{\textbf{x}}$& Optimizers in Section~\ref{s:hetero} & \\
                $\pi$ &Ranking of banks for greedy algorithm in Section~\ref{s:hetero}  &\\
            $y_i$ & B$_i$'s long-term real asset value in Section~\ref{s:hetero} & $y_i > 0$\\
        $\ell_i$ & Size of B$_i$'s  liquidity shock in Section \ref{s:hetero} & $\theta y_i<\ell_i<y_i$\\          $\theta_i$ & B$_i$'s chosen pledgeability in Section~\ref{s:endogenous theta} & $0\leq\theta_i\leq 1$\\         $\theta^a\equiv\ell/y$ & Pledgeability needed to meet a shock unaided, Section~\ref{s:endogenous theta} &\\         $c(\cdot)$ & Cost of pledgeability in Section~\ref{s:endogenous theta} & $c(0)=0$, $c',c''>0$\\         $p$ & Probability a bank is shocked in Section~\ref{s:endogenous theta} & $0<p<1$\\
                \ & &  \\
        \hline
        \multicolumn{3}{c}{Notations Used Only in Proofs}
        \\ \hline
     $\textbf{Q}$ & A matrix, usually shorthand for  $\mathbf{\hat{F}}^{\top}$   &\\
        $\textbf{Q}_{\mathscr{B}_1,\mathscr{B}_2}\equiv[Q_{ij}]_{i\in\mathscr{B}_1,j\in\mathscr{B}_2}$ & Block matrix with rows in $\mathscr{B}_1$ and columns in $\mathscr{B}_2$&\\
        $\mathbf{\tilde Q}$&Matrix in Lemma~\ref{l:radius 2}&\\
         $    \mathbf{0} , \mathbf{1}$ & Vectors of zeros and ones ($(0,..., 0)$ and $(1,..., 1)$) & \\
        $\mathbf{d}_{i\to \mathscr{B}}\equiv\{d_{i\to j}\}_{j\in \mathscr{B}}$&Vector of $\mathrm{B}_i$'s harmonic distance to banks in $\mathscr{B}$&\\
        $D_i\equiv F_{i\rightrightarrows}-R_{i\rightrightarrows}$& $\mathrm{B}_i$'s shortfall&\\
        $\mathbf{D}\equiv\mathbf{F}_{\rightrightarrows}-\mathbf{R}_{\rightrightarrows}$&Vector of each bank's shortfall&\\
                $\mathscr{D}$ & Set of banks that default & \\
       $\Phi^{\alpha},\Psi^{\alpha}$& Mappings  used in  Lemma~\ref{l:Phi} and Lemma~\ref{lm:shortfall}&\\
        $\mathscr{I}^{\alpha}, \mathscr{H}^{\alpha}$ & Restricted domains of   $\Phi^{\alpha}$ and  $\Psi^{\alpha}$ & \\ $\prod_{i=1}^N X_i=X_1\times\cdots\times X_N$& Cartesian product of sets $X_1,...,X_N$ &
        \\
         & & \\
\end{longtable}

\newpage
\bibliographystyle{chicago}

\bibliography{Netting}

\begin{thebibliography}{}

\bibitem[\protect\citeauthoryear{Acemoglu, Ozdaglar, and
  Tahbaz-Salehi}{Acemoglu et~al.}{2015}]{AOT}
Acemoglu, D., A.~Ozdaglar, and A.~Tahbaz-Salehi (2015).
\newblock Systemic risk and stability in financial networks.
\newblock {\em American Economic Review\/}~{\em 105\/}(2), 564--608.

\bibitem[\protect\citeauthoryear{Allen, Babus, and Carletti}{Allen
  et~al.}{2009}]{allen2009financial}
Allen, F., A.~Babus, and E.~Carletti (2009).
\newblock Financial crises: Theory and evidence.
\newblock {\em Annu. Rev. Financ. Econ.\/}~{\em 1\/}(1), 97--116.

\bibitem[\protect\citeauthoryear{Allen, Babus, and Carletti}{Allen
  et~al.}{2012}]{ALLE/BABU/CARL/12}
Allen, F., A.~Babus, and E.~Carletti (2012).
\newblock {Asset commonality, debt maturity and systemic risk}.
\newblock {\em Journal of Financial Economics\/}~{\em 104\/}(3), 519--534.

\bibitem[\protect\citeauthoryear{Allen and Gale}{Allen and
  Gale}{1998}]{Allen-Gale-1998}
Allen, F. and D.~Gale (1998).
\newblock Optimal financial crises.
\newblock {\em The Journal of Finance\/}~{\em 53}, 1245--1284.

\bibitem[\protect\citeauthoryear{Allen and Gale}{Allen and
  Gale}{2000}]{Allen-Gale-2000}
Allen, F. and D.~Gale (2000).
\newblock {Financial contagion}.
\newblock {\em Journal of Political Economy\/}~{\em 108\/}(1), 1--33.

\bibitem[\protect\citeauthoryear{Allen and Walther}{Allen and
  Walther}{2021}]{Allen-Walther-2021}
Allen, F. and A.~Walther (2021).
\newblock Financial architecture and financial stability.
\newblock {\em Annual Review of Financial Economics\/}~{\em 13\/}(1), 129--151.

\bibitem[\protect\citeauthoryear{Angeletos}{Angeletos}{2002}]{Angeletos-2002}
Angeletos, G.-M. (2002).
\newblock Fiscal policy with noncontingent debt and the optimal maturity
  structure.
\newblock {\em The Quarterly Journal of Economics\/}~{\em 117\/}(3),
  1105--1131.

\bibitem[\protect\citeauthoryear{Berman and Plemmons}{Berman and
  Plemmons}{1979}]{Berman-Plemmons-1979}
Berman, A. and R.~Plemmons (1979).
\newblock {\em Nonnegative Matrices in the Mathematical Sciences}.
\newblock Elsevier Inc.

\bibitem[\protect\citeauthoryear{Bernard, Capponi, and Stiglitz}{Bernard
  et~al.}{2022}]{Bernard-Capponi-Stiglitz-2022}
Bernard, B., A.~Capponi, and J.~E. Stiglitz (2022).
\newblock Bail-ins and bailouts: Incentives, connectivity, and systemic
  stability.
\newblock {\em Journal of Political Economy\/}~{\em 130\/}(7), 1805--1859.

\bibitem[\protect\citeauthoryear{Bjerre}{Bjerre}{1999}]{Bjerre-1999}
Bjerre, C.~S. (1999).
\newblock Secured transactions inside out: Negative pledge covenants, property
  and perfection.
\newblock {\em Cornell Law Review\/}~{\em 84\/}(2), 305–341.

\bibitem[\protect\citeauthoryear{Bluhm, Georg, and Krahnen}{Bluhm
  et~al.}{2016}]{Bluhm-et-al-2016}
Bluhm, M., C.-P. Georg, and J.-P. Krahnen (2016).
\newblock {Interbank intermediation}.
\newblock Working paper, Deutsche Bundesbank, Research Centre.

\bibitem[\protect\citeauthoryear{Bolton and Oehmke}{Bolton and
  Oehmke}{2019}]{Bolton-Oehmke-2019}
Bolton, P. and M.~Oehmke (2019).
\newblock Bank resolution and the structure of global banks.
\newblock {\em The Review of Financial Studies\/}~{\em 32\/}(6), 2384--2421.

\bibitem[\protect\citeauthoryear{Calvin and Leung}{Calvin and
  Leung}{2003}]{calvin2003average}
Calvin, J.~M. and J.~Y.-T. Leung (2003).
\newblock Average-case analysis of a greedy algorithm for the 0/1 knapsack
  problem.
\newblock {\em Operations Research Letters\/}~{\em 31\/}(3), 202--210.

\bibitem[\protect\citeauthoryear{Capponi, Corell, and Stiglitz}{Capponi
  et~al.}{2022}]{Capponi-Corell-Stiglitz-2022}
Capponi, A., F.~Corell, and J.~E. Stiglitz (2022).
\newblock Optimal bailouts and the doom loop with a financial network.
\newblock {\em Journal of Monetary Economics\/}~{\em 128}, 35--50.

\bibitem[\protect\citeauthoryear{Craig and Ma}{Craig and
  Ma}{2022}]{Craig-Ma-2021}
Craig, B. and Y.~Ma (2022).
\newblock Intermediation in the interbank lending market.
\newblock {\em Journal of Financial Economics\/}~{\em 145\/}(2), 179--207.

\bibitem[\protect\citeauthoryear{Craig and Von~Peter}{Craig and
  Von~Peter}{2014}]{Craig-vonPeter-2014}
Craig, B. and G.~Von~Peter (2014).
\newblock Interbank tiering and money center banks.
\newblock {\em Journal of Financial Intermediation\/}~{\em 23\/}(3), 322--347.

\bibitem[\protect\citeauthoryear{Cs{\'o}ka and Herings}{Cs{\'o}ka and
  Herings}{2021}]{csoka2021axiomatization}
Cs{\'o}ka, P. and P.~J.-J. Herings (2021).
\newblock An axiomatization of the proportional rule in financial networks.
\newblock {\em Management Science\/}~{\em 67\/}(5), 2799--2812.

\bibitem[\protect\citeauthoryear{Dantzig}{Dantzig}{1957}]{dantzig1957discrete}
Dantzig, G.~B. (1957).
\newblock Discrete-variable extremum problems.
\newblock {\em Operations Research\/}~{\em 5\/}(2), 266--288.

\bibitem[\protect\citeauthoryear{DeMarzo and Fishman}{DeMarzo and
  Fishman}{2007}]{DeMarzo-Fishman-2007b}
DeMarzo, P.~M. and M.~J. Fishman (2007).
\newblock Optimal long-term financial contracting.
\newblock {\em Review of Financial Studies\/}~{\em 20\/}(6), 2079--2128.

\bibitem[\protect\citeauthoryear{Diamond}{Diamond}{1993}]{Diamond-1993}
Diamond, D.~W. (1993).
\newblock {Seniority and maturity of debt contracts}.
\newblock {\em Journal of Financial Economics\/}~{\em 33\/}(3), 341--368.

\bibitem[\protect\citeauthoryear{Donaldson, Gromb, and Piacentino}{Donaldson
  et~al.}{2020}]{Paradox}
Donaldson, J.~R., D.~Gromb, and G.~Piacentino (2020).
\newblock The paradox of pledgeability.
\newblock {\em Journal of Financial Economics\/}~{\em 137\/}(3), 591--605.

\bibitem[\protect\citeauthoryear{Donaldson, Gromb, and Piacentino}{Donaldson
  et~al.}{2021}]{Reallocation}
Donaldson, J.~R., D.~Gromb, and G.~Piacentino (2021).
\newblock Collateral reallocation.
\newblock Working paper, Washington University in St. Louis.

\bibitem[\protect\citeauthoryear{Donaldson, Gromb, and Piacentino}{Donaldson
  et~al.}{2025}]{Donaldson-Gromb-Piacentino-2025}
Donaldson, J.~R., D.~Gromb, and G.~Piacentino (2025).
\newblock Conflicting priorities: A theory of covenants and collateral.
\newblock {\em The Journal of Finance\/}~{\em 80\/}(3), 1739--1768.

\bibitem[\protect\citeauthoryear{Donaldson and Micheler}{Donaldson and
  Micheler}{2018}]{donaldson2018resaleable}
Donaldson, J.~R. and E.~Micheler (2018).
\newblock Resaleable debt and systemic risk.
\newblock {\em Journal of Financial Economics\/}~{\em 127\/}(3), 485--504.

\bibitem[\protect\citeauthoryear{Donaldson and Piacentino}{Donaldson and
  Piacentino}{2018}]{Netting}
Donaldson, J.~R. and G.~Piacentino (2018).
\newblock Netting.
\newblock Working paper, Washington University in St. Louis.

\bibitem[\protect\citeauthoryear{Donaldson and Piacentino}{Donaldson and
  Piacentino}{2026}]{Collateral}
Donaldson, J.~R. and G.~Piacentino (2026).
\newblock Collateral.
\newblock {\em Annual Review of Financial Economics\/}.

\bibitem[\protect\citeauthoryear{Dubey, Geanakoplos, and Shubik}{Dubey
  et~al.}{1988}]{Dubey-et-al-1988}
Dubey, P., J.~Geanakoplos, and M.~Shubik (1988).
\newblock Default and efficiency in a general equilibrium model with incomplete
  markets.
\newblock Cowles Foundation Discussion Papers 879R, Cowles Foundation for
  Research in Economics, Yale University.

\bibitem[\protect\citeauthoryear{Duffie and Skeel}{Duffie and
  Skeel}{2012}]{Duffie-Skeel-2012}
Duffie, D. and D.~Skeel (2012).
\newblock A dialogue on the costs and benefits of automatic stays for
  derivatives and repurchase agreements.
\newblock In {\em Bankruptcy Not Bailout}. Hoover Institution, Stanford.

\bibitem[\protect\citeauthoryear{Eisenberg and Noe}{Eisenberg and
  Noe}{2001}]{Eisenberg-Noe-2001}
Eisenberg, L. and T.~Noe (2001).
\newblock Systemic risk in financial systems.
\newblock {\em Management Science\/}~{\em 47\/}(2), 236–249.

\bibitem[\protect\citeauthoryear{Elliott, Golub, and Jackson}{Elliott
  et~al.}{2014}]{ELLI/GOLU/JACK/14}
Elliott, M., B.~Golub, and M.~Jackson (2014).
\newblock Financial networks and contagion.
\newblock {\em American Economic Review\/}~{\em 104\/}(10), 3115--3153.

\bibitem[\protect\citeauthoryear{Erol}{Erol}{2019}]{Erol-2019}
Erol, S. (2019).
\newblock Network hazard and bailouts.
\newblock Working paper, Carnegie Mellon University.

\bibitem[\protect\citeauthoryear{Farboodi}{Farboodi}{2023}]{Farboodi-2021}
Farboodi, M. (2023).
\newblock Intermediation and voluntary exposure to counterparty risk.
\newblock {\em Journal of Political Economy\/}~{\em 131\/}(12), 3267--3309.

\bibitem[\protect\citeauthoryear{Farboodi, Jarosch, and Shimer}{Farboodi
  et~al.}{2017}]{Farboodi-et-al-2017}
Farboodi, M., G.~Jarosch, and R.~Shimer (2017).
\newblock The emergence of market structure.
\newblock {NBER} Working Paper 23234, National Bureau of Economic Research.

\bibitem[\protect\citeauthoryear{Flannery}{Flannery}{2014}]{Flannery-2014}
Flannery, M.~J. (2014).
\newblock Contingent capital instruments for large financial institutions: A
  review of the literature.
\newblock {\em Annual Review of Financial Economics\/}~{\em 6}, 225--240.

\bibitem[\protect\citeauthoryear{Gabrieli and Georg}{Gabrieli and
  Georg}{2014}]{gabrieli2014network}
Gabrieli, S. and C.-P. Georg (2014).
\newblock A network view on interbank market freezes.
\newblock Bundesbank Discussion Paper 44/2014, Deutsche Bundesbank.

\bibitem[\protect\citeauthoryear{Gale}{Gale}{1990}]{Gale-1990}
Gale, D. (1990).
\newblock The efficient design of public debt.
\newblock In R.~Dornbusch and M.~Draghi (Eds.), {\em Public Debt Management:
  Theory and History}, pp.\  14–47. Cambridge University Press.

\bibitem[\protect\citeauthoryear{Glasserman and Young}{Glasserman and
  Young}{2016}]{glasserman2016contagion}
Glasserman, P. and H.~P. Young (2016).
\newblock Contagion in financial networks.
\newblock {\em Journal of Economic Literature\/}~{\em 54\/}(3), 779--831.

\bibitem[\protect\citeauthoryear{Goyal}{Goyal}{2005}]{Goyal-2005}
Goyal, V.~K. (2005).
\newblock Market discipline of bank risk: Evidence from subordinated debt
  contracts.
\newblock {\em Journal of Financial Intermediation\/}~{\em 14\/}(3), 318--350.

\bibitem[\protect\citeauthoryear{Hart and Moore}{Hart and
  Moore}{1995}]{Hart-Moore-1995}
Hart, O. and J.~Moore (1995).
\newblock {Debt and Seniority: An Analysis of the Role of Hard Claims in
  Constraining Management}.
\newblock {\em American Economic Review\/}~{\em 85\/}(3), 567--85.

\bibitem[\protect\citeauthoryear{He and Li}{He and Li}{2022}]{He-Li-2022}
He, Z. and J.~Li (2022).
\newblock Intermediation via {{Credit Chains}}.
\newblock NBER Working Paper 29632, National Bureau of Economic Research.

\bibitem[\protect\citeauthoryear{Holmstr\"{o}m and Tirole}{Holmstr\"{o}m and
  Tirole}{1998}]{Holmstrom-Tirole-1998}
Holmstr\"{o}m, B. and J.~Tirole (1998).
\newblock Private and public supply of liquidity.
\newblock {\em Journal of Political Economy\/}~{\em 106\/}(1), 1--40.

\bibitem[\protect\citeauthoryear{Holmstr\"{o}m and Tirole}{Holmstr\"{o}m and
  Tirole}{2011}]{Holmstrom-Tirole-2011}
Holmstr\"{o}m, B. and J.~Tirole (2011).
\newblock {\em Inside and Outside Liquidity}.
\newblock The MIT Press.

\bibitem[\protect\citeauthoryear{Jackson and Pernoud}{Jackson and
  Pernoud}{2021}]{jackson2021systemic}
Jackson, M.~O. and A.~Pernoud (2021).
\newblock Systemic risk in financial networks: A survey.
\newblock {\em Annual Review of Economics\/}~{\em 13}, 171--202.

\bibitem[\protect\citeauthoryear{Jackson and Pernoud}{Jackson and
  Pernoud}{2024}]{Jackson-Pernoud-2024}
Jackson, M.~O. and A.~Pernoud (2024).
\newblock Credit freezes, equilibrium multiplicity, and optimal bailouts in
  financial networks.
\newblock {\em The Review of Financial Studies\/}~{\em 37\/}(7), 2017--2062.

\bibitem[\protect\citeauthoryear{Jorion}{Jorion}{2000}]{Jorion-2000}
Jorion, P. (2000).
\newblock {Risk management lessons from Long-Term Capital Management}.
\newblock {\em European Financial Management\/}~{\em 6\/}(3), 277--300.

\bibitem[\protect\citeauthoryear{Kanik}{Kanik}{2020}]{Kanik-2020}
Kanik, Z. (2020).
\newblock From {Lombard} {Street} to {Wall} {Street}: Systemic risk, rescues,
  and stability in financial networks.
\newblock Working paper, University of Glasgow.

\bibitem[\protect\citeauthoryear{Kanik}{Kanik}{2022}]{Kanik-2022}
Kanik, Z. (2022).
\newblock A new era for financial networks: Mandatory bail-ins.
\newblock Working paper, University of Glasgow.

\bibitem[\protect\citeauthoryear{Kuo, Skeie, Vickery, and Youle}{Kuo
  et~al.}{2014}]{Kuo-et-al-2014}
Kuo, D., D.~Skeie, J.~Vickery, and T.~Youle (2014).
\newblock Identifying term interbank loans from {Fedwire} payments data.
\newblock Staff Report 603, Federal Reserve Bank of New York.

\bibitem[\protect\citeauthoryear{Kusnetsov and Veraart}{Kusnetsov and
  Veraart}{2019}]{Kusnetsov-Veraart-2018}
Kusnetsov, M. and L.~Veraart (2019).
\newblock Interbank clearing in financial networks with multiple maturities.
\newblock {\em {SIAM} Journal on Financial Mathematics\/}~{\em 10\/}(1),
  37--67.

\bibitem[\protect\citeauthoryear{Leitner}{Leitner}{2005}]{Leitner-2005}
Leitner, Y. (2005).
\newblock Financial networks: Contagion, commitment, and private sector
  bailouts.
\newblock {\em The Journal of Finance\/}~{\em 60\/}(6), 2925--2953.

\bibitem[\protect\citeauthoryear{Martello and Toth}{Martello and
  Toth}{1990}]{Martello-Toth-1990}
Martello, S. and P.~Toth (1990).
\newblock {\em Knapsack Problems: Algorithms and Computer Implementations}.
\newblock John Wiley \& Sons, Inc.

\bibitem[\protect\citeauthoryear{Perotti and Spier}{Perotti and
  Spier}{1993}]{Perotti-Spier-1993}
Perotti, E. and K.~Spier (1993).
\newblock Capital structure as a bargaining tool: The role of leverage in
  contract renegotiation.
\newblock {\em American Economic Review\/}~{\em 83\/}(5), 1131--41.

\bibitem[\protect\citeauthoryear{Plemmons}{Plemmons}{1977}]{Plemmons-1977}
Plemmons, R. (1977).
\newblock {$M$-matrix characterizations. I---Nonsingular $M$-matrices}.
\newblock {\em Linear Algebra and its Applications\/}~{\em 18\/}(2), 175--188.

\bibitem[\protect\citeauthoryear{Rogers and Veraart}{Rogers and
  Veraart}{2013}]{Rogers-Veraart-2013}
Rogers, L. C.~G. and L.~A.~M. Veraart (2013).
\newblock Failure and rescue in an interbank network.
\newblock {\em Management Science\/}~{\em 59\/}(4), 882--898.

\bibitem[\protect\citeauthoryear{Rose, Bergstresser, and Lane}{Rose
  et~al.}{2009}]{BearCase}
Rose, C.~S., D.~B. Bergstresser, and D.~Lane (2009).
\newblock The tip of the iceberg: {JP Morgan Chase and Bear Stearns (A)}.

\bibitem[\protect\citeauthoryear{Roukny, Battiston, and Stiglitz}{Roukny
  et~al.}{2018}]{Roukny-Battiston-Stiglitz-2018}
Roukny, T., S.~Battiston, and J.~E. Stiglitz (2018).
\newblock Interconnectedness as a source of uncertainty in systemic risk.
\newblock {\em Journal of Financial Stability\/}~{\em 35}, 93--106.

\bibitem[\protect\citeauthoryear{Schwartz}{Schwartz}{1989}]{Schwartz-1989}
Schwartz, A. (1989).
\newblock A theory of loan priorities.
\newblock {\em The Journal of Legal Studies\/}~{\em 18\/}(2), 209–261.

\bibitem[\protect\citeauthoryear{Shu}{Shu}{2025}]{Shu-2025}
Shu, C. (2025).
\newblock Endogenous risk exposure and systemic instability.
\newblock {\em Management Science\/}~{\em 71\/}(7), 5511--5528.

\bibitem[\protect\citeauthoryear{Stulz and Johnson}{Stulz and
  Johnson}{1985}]{Stulz-Johnson-1985}
Stulz, R.~M. and H.~Johnson (1985).
\newblock {An analysis of secured debt}.
\newblock {\em Journal of Financial Economics\/}~{\em 14\/}(4), 501--521.

\bibitem[\protect\citeauthoryear{Sundaresan and Wang}{Sundaresan and
  Wang}{2015}]{Sundaresan-Wang-2015}
Sundaresan, S.~M. and Z.~Wang (2015).
\newblock On the design of contingent capital with a market trigger.
\newblock {\em The Journal of Finance\/}~{\em 70\/}(2), 881--920.

\bibitem[\protect\citeauthoryear{Upper and Worms}{Upper and
  Worms}{2004}]{upperEstimatingBilateralExposures2004}
Upper, C. and A.~Worms (2004).
\newblock Estimating bilateral exposures in the {German} interbank market: Is
  there a danger of contagion?
\newblock {\em European Economic Review\/}~{\em 48\/}(4), 827--849.

\bibitem[\protect\citeauthoryear{Valukas}{Valukas}{2010}]{LehmanReport}
Valukas, A.~R. (2010).
\newblock {Lehman Brothers Holdings Inc.} ch.\ 11 proceedings.

\bibitem[\protect\citeauthoryear{Zame}{Zame}{1993}]{Zame-1993}
Zame, W.~R. (1993).
\newblock Efficiency and the role of default when security markets are
  incomplete.
\newblock {\em American Economic Review\/}~{\em 83\/}(5), 1142--1164.

\end{thebibliography}

\end{document}